\newcommand{\paperversion}{arxiv}
\providecommand{\paperversion}{arxiv}
\providecommand{\paperversion}{finaldraft}
\providecommand{\anonymoussubmission}{1}
\RequirePackage{paperversions}

\ifweb
  \reporttrue
  \finalformattrue
\fi

\PassOptionsToPackage{dvipsnames}{xcolor}

\documentclass[conference]{IEEEtran}

\usepackage[utf8]{inputenc}

\iffinalformat
  \usepackage{pbalance}
\fi

\usepackage{microtype}
\usepackage[english]{babel}
\newcommand{\hyp}{\babelhyphen{hard}}
\usepackage{inconsolata}

\usepackage{enumitem}
\usepackage{kvoptions}

\usepackage{amsmath}
\usepackage{amssymb}
\usepackage{amsthm}
\usepackage{mathtools}
\mathtoolsset{mathic}

\usepackage{mathpartir}
\usepackage{proof}
\usepackage{pl-syntax}

\usepackage{tikz}
  \usetikzlibrary{arrows, calc, intersections, positioning}
  \newlength{\nodedistance}
\ifCLASSOPTIONcompsoc
  \usepackage[caption=false,font=normalsize,labelfont=sf,textfont=sf]{subfig}
\else
  \usepackage[caption=false,font=footnotesize]{subfig}
\fi

\newif\ifappendixproofs
\appendixproofstrue

\iffinalformat
 \ifreport
  \appendixproofstrue
 \else
  \appendixproofsfalse
 \fi
\fi

\ifappendixproofs
  \usepackage[bibliography=common]{apxproof}
\else
  \usepackage[bibliography=common, appendix=strip]{apxproof}
\fi

  \renewcommand{\appendixsectionformat}[2]{Details for Section #1 (#2)}

\newcommand{\apxref}[1]{\ifappendixproofs Appendix~\ref{#1}\else the technical report~\cite{viaduct-formal-tr}\fi
}

\iffinalformat
  \usepackage[pages=16]{checkpagelimit}
\else
  \usepackage[pages=12]{checkpagelimit}
\fi
\ifshowcomments
  \usepackage{aplcomments}
\else
  \usepackage[disabled]{aplcomments}
\fi

\usepackage{hyperref}
\ifweb\relax\else
  \hypersetup{hidelinks}
\fi
\usepackage{cleveref}
  \ifreport
    \crefformat{section}{Section~#2#1#3}
    \Crefformat{section}{Section~#2#1#3}
  \else
    \crefformat{section}{\S#2#1#3}
  \fi

\usepackage{overrides}

\usepackage{abstract-syntax}
\usepackage{dynamic-semantics}
\usepackage{functions}
\usepackage{general}
\usepackage{information-flow}
\usepackage{judgment}
\usepackage{meta-variables}
\usepackage{partitioning}
\usepackage{secure-compilation}
\usepackage{static-semantics}
\usepackage{theorem}

\usepackage[numbers]{natbib}
\providecommand\authorname[1]{\IEEEauthorblockN{#1}}
\providecommand\affiliation[1]{\IEEEauthorblockA{#1}}
\providecommand\orcid[1]{\relax}
\providecommand\institution[1]{#1}
\providecommand\department[1]{\relax}
\providecommand\city[1]{\\#1, }
\providecommand\state[1]{#1, }
\providecommand\postcode[1]{\relax}
\providecommand\country[1]{#1}
\providecommand\email[1]{\\\texttt{#1}}

\newcommand{\cornell}[1]{\affiliation{\institution{Cornell University}\department{Department of Computer Science}\city{Ithaca}\state{NY}\postcode{14850}\country{USA}\email{#1}}
}

\newcommand{\cmu}[1]{\affiliation{\institution{Carnegie Mellon University}\department{Department of Computer Science}\city{Pittsburgh}\state{PA}\postcode{15213}\country{USA}\email{#1}}
}

\ifblinded
\author{}
\else
\author{\authorname{Co\c{s}ku Acay}
  \orcid{0000-0002-0487-1167}
  \cornell{cacay@cs.cornell.edu}
  \and
  \authorname{Joshua Gancher}
  \cmu{jgancher@andrew.cmu.edu}
  \and
  \authorname{Rolph Recto}
  \cornell{rr729@cornell.edu}
  \and
  \authorname{Andrew C. Myers}
  \orcid{0000-0001-5819-7588}
  \cornell{andru@cs.cornell.edu}
}
\fi
  \title{
  Secure Synthesis of Distributed Cryptographic Applications
\ifreport
  {\\ \LARGE (Technical Report)}
\fi
}
 
\ifshowpagenumbers
\fi

\begin{document}
\maketitle
\begin{abstract}
Developing secure distributed systems is difficult, and
even harder when advanced cryptography must be used
to achieve security goals.
Following prior work, we advocate using \emph{secure program partitioning} to
synthesize cryptographic applications:
instead of implementing a system of communicating processes,
the programmer implements a \emph{centralized, sequential} program,
which is automatically compiled into a secure distributed version
that uses cryptography.

While this approach is promising, formal results for the security of such compilers
are limited in scope. In particular, no security proof yet simultaneously
addresses subtleties essential for robust, efficient applications:
multiple cryptographic mechanisms, malicious
corruption, and asynchronous communication.

In this work, we develop a compiler security proof that handles these subtleties.
Our proof relies on a novel unification of
simulation-based security, information-flow control, choreographic programming,
and sequentialization techniques for concurrent programs.
While our proof targets hybrid protocols, which abstract cryptographic
mechanisms as idealized functionalities, our approach offers a clear path toward
leveraging Universal Composability to obtain end-to-end, modular security
results with fully instantiated cryptographic mechanisms.

Finally, following prior observations about simulation-based security,
we prove that our result guarantees \emph{robust hyperproperty preservation},
an important criterion for compiler correctness that preserves
all source-level security properties in target programs.
 \end{abstract}

\ifblinded
  \bstctlcite{IEEEtran:override}
\fi

\section{Introduction}
\label{sec:intro}

Ensuring security for modern distributed applications remains a difficult
challenge, as such systems can cross administrative boundaries and involve
parties that do not fully trust each other.
To defend their security policies, some applications employ sophisticated
mechanisms such as complex distributed protocols~\cite{paxos,pbft},
trusted hardware~\cite{intelsgx1, dflate, sanctum},
and advanced uses of cryptography.
These technologies add significant complexity to software
development and require expertise to use
successfully~\cite{Egele:2013, georgievIJABS12, decker2014bitcoin}.

To ease the development of secure distributed applications,
prior work leverages compilers that translate high-level programs into
distributed protocols that employ advanced security mechanisms.
Unfortunately, most compilers only target a single mechanism---such as multiparty computation~\cite{fairplay,oblivm,ScaleMamba,wysteria},
zero-knowledge proofs~\cite{pinocchio,geppetto,buffet,xjsnark},
or homomorphic encryption~\cite{CHET,EVA,porcupine}---and thus do not support
secure combinations of mechanisms.
On the other hand, compilers that perform
\emph{secure program partitioning}~\cite{zznm02,zcmz03,swift07,fr-popl08,fgr09,viaduct-pldi21}
do combine mechanisms, but come with limited or informal correctness guarantees.

In this work, we give the first formal security result for program partitioning
that targets multiple cryptographic mechanisms, arbitrary corruption, and
adversarially controlled scheduling.
Our work proves the correctness of a reasonably faithful model of the compilation process used in the Viaduct compiler~\cite{viaduct-pldi21}.
We formalize our result in the \emph{simulation-based} security framework, which establishes a modular foundation for cryptographic protocol security~\cite{UC}.
Our security proof is primarily concerned with program partitioning itself, and
thus does not reason about the concrete instantiation of cryptographic mechanisms;
however, we discuss how to extend our results to reason about concrete
mechanisms.

Our security proof incorporates multiple techniques for
simulation\hyp{}based security:
information\hyp{}flow type systems~\cite{sm-jsac} to define the security policy
and to guide partitioning,
choreographies~\cite{Montesi23} to define global programs for
distributed executions,
and a novel information\hyp{}flow guided technique for concurrent program
sequentialization~\cite{canonical-sequentialization}.
\begin{itemize}
  \item
    We formalize a variant of Simplified Universal Composability (SUC)~\cite{SUC}
    enriched with information flow, allowing us to capture distributed protocols
    in the presence of adversarial scheduling and corruption.

  \item
    We show how to model secure program partitioning using
    security-typed choreographies~\cite{Montesi23}.
    The input to program partitioning is a sequential program
    representing an idealized execution on a single, trusted security domain,
    while the output is a distributed protocol with message-passing concurrency
    between mutually distrusting agents.
  \item
    We prove simulation-based security for our model of program partitioning.
    Informally, we show that any adversary interacting with the
    compiled distributed program is no more powerful than a corresponding adversary
    (a \emph{simulator}) interacting with the source program.
  \item
    We show that, in our setting, simulation implies
    \emph{robust hyperproperty preservation}~\cite{abateB0HPT19},
    a strong criterion for compiler correctness that ensures security conditions
    of source programs are preserved in target programs.
\end{itemize}

\ifreport
This technical report is an expanded version of the corresponding
conference paper~\cite{viaduct-formal}. In particular, it includes
the proof details that had to be omitted for reasons of space, found
in Appendices A--G.
\fi
 
\begin{figure*}
  \definecolor{lightGray}{RGB}{180,180,180}
  \Crefformat{section}{\S#2#1#3}
\begin{tikzpicture}[
  auto,
  thick,
  node distance=\nodedistance,
  every node/.style={align=center},
  program/.style={draw, rounded corners, text width=9em, minimum height=8ex},
  simulation/.style={},
  transformation/.style={->, >=stealth', rounded corners=0.5mm},
]
  \node[program] (1)
    {Sequential Source \\ $\source{\process}$, $\idealstepsto{}$};
  \node[simulation] (12) [right=of 1]
    {\Cref{sec:correctness-of-synthesis} \\ $\simulatedBy$};
  \node[program] (2) [right=of 12]
    {Choreography \\ $\corrupt{\process}$, $\realstepsto!{}$};
  \node[simulation] (23) [right=of 2]
    {\Cref{sec:endpoint-projection} \\ $\simulatedBy$};
  \node[program] (3) [right=of 23]
    {Hybrid Distributed \\ $\corrupt{\partition{\process}}$, $\realstepsto{}$};
  \node[simulation] (34) [right=of 3]
    {\Cref{sec:uc} \\ {\color{lightGray} $\simulatedBy$}};
  \node[program] (4) [dashed, right=of 34]
    {Concrete Distributed};

  \pgfmathsetmacro{\arrowshift}{0.7em}

  \draw[transformation]
    (1.north)
    -- ++(0, \nodedistance)
    -| node [pos=0.25, above] {Protocol Synthesis} ([xshift=-\arrowshift]2.north);

  \draw[transformation]
    ([xshift=\arrowshift]2.north)
    -- ++(0, \nodedistance)
    -| node [pos=0.25, above] {Endpoint Projection} ([xshift=-\arrowshift]3.north);

  \draw[transformation, dashed]
    ([xshift=\arrowshift]3.north)
    -- ++(0, \nodedistance)
    -| node [pos=0.25, above] {Cryptographic Instantiation} (4.north);
\end{tikzpicture}
   \caption{
    Overview of compilation and the correctness proof.
Left-to-right arrows are compilation steps; $\simulatedBy$ are proof steps.
Term \process is a choreography,
    $\partition{\cdot}$ is endpoint projection,
    $\source{\cdot}$ is the inverse of protocol synthesis,
    and $\corrupt{\cdot}$ models corruption.
    Dashed components represent proof sketches.
  }
  \label{fig:compilation-overview}
\end{figure*}
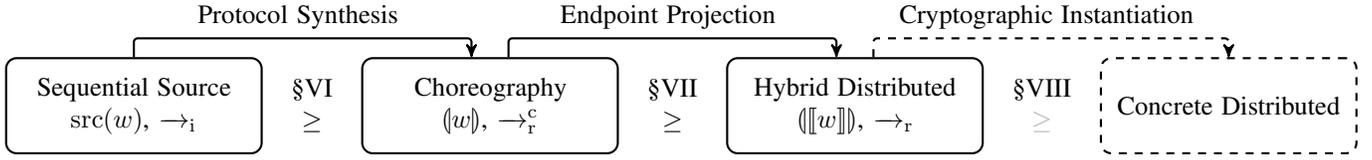

\begin{figure}
  \recovercaptionskip \small \begin{subfigure}{\linewidth}{Source program with \iflabel{information-flow labels}.}
  \(
    \begin{program}
      \slethead{a : \lalice}{\einputh{\alice}}
      \\
      \slethead{b : \lbob}{\einputh{\bob}}
      \\
      \slethead{x}{
        \edeclassify
          {\dendorse{a} < \dendorse{b}}
          {\\ \hspace{9em} \iflabel{A \wedge B}}
          {\iflabel{A \meet B}}}
      \\
      \eoutputh{x}{\alice}
      \\
      \eoutputh{x}{\bob}
    \end{program}
  \)
  \label{fig:example-source}
\end{subfigure}

\begin{subfigure}{\linewidth}{Choreography with explicit communication and synchronization.}
  \(
    \begin{program}
      \slethhead{a}{\alice}{\einput{\alice}}
      \\
      \smovehead{\alice}{a}{\MPCalicebob}{a'}
      \\
      \slethhead{b}{\bob}{\einput{\bob}}
      \\
      \smovehead{\bob}{b}{\MPCalicebob}{b'}
      \\
      \slethhead{x}{\MPCalicebob}{\\ \hspace{4em}\ddeclassify{(\dendorse{a'} < \dendorse{b'})}}
      \\
      \smovehead{\MPCalicebob}{x}{\alice}{x_1}
      \\
      \smovehead{\MPCalicebob}{x}{\bob}{x_2}
      \\
      \eoutputh{x_1}{\alice}
      \\
      \smovehead{\alice}{\vunit}{\bob}{\_}~\scomment{Sync outputs}
      \\
      \eoutputh{x_2}{\bob}
    \end{program}
  \)
  \label{fig:example-choreography}
\end{subfigure}

\begin{subfigure}{\linewidth}{Hybrid distributed program derived by projecting choreography.}
  \centering
  \(
    \begin{program}
      \scomment{\alice}
      \\
      \slethead{a}{\einput{\alice}}
      \\
      \esend{a}{\MPCalicebob}
      \\
      \slethead{x_1}{\ereceive{\MPCshort}}
      \\
      \eoutput{x_1}{\alice}
      \\
      \esend{\vunit}{\bob}~\scomment{Sync}
    \end{program}
  \)
  \hfill
  \(
    \begin{program}
      \scomment{\bob}
      \\
      \slethead{b}{\einput{\bob}}
      \\
      \esend{b}{\MPCalicebob}
      \\
      \slethead{x_2}{\ereceive{\MPCshort}}
      \\
      \slethead{\_}{\ereceive{\alice}}~\scomment{Sync}
      \\
      \eoutput{x_2}{\bob}
    \end{program}
  \)

  \addvspace{\belowdisplayskip}
  \(
    \begin{program}
      \scomment{\MPCalicebob}
      \\
      \slethead{a'}{\ereceive{\alice}}
      \\
      \slethead{b'}{\ereceive{\bob}}
      \\
      \slethead{x}{\ddeclassify{(\dendorse{a'} < \dendorse{b'})}}
      \\
      \esend{x}{\alice}
      \\
      \esend{x}{\bob}
    \end{program}
  \)
  \label{fig:example-target}
\end{subfigure}

\begin{subfigure}{\linewidth}{Concrete distributed program derived by instantiating idealized hosts.}
  \(
    \begin{program}
      \scomment{\alice}
      \\
      \slethead{a}{\einput{\alice}}
      \\
      \scomment{Calls to MPC library}
      \\
      \eoutput{x_1}{\alice}
      \\
      \esend{\vunit}{\bob}~\scomment{Sync}
    \end{program}
  \)
  \hfill
  \(
    \begin{program}
      \scomment{\bob}
      \\
      \slethead{b}{\einput{\bob}}
      \\
      \scomment{Calls to MPC library}
      \\
      \slethead{\_}{\ereceive{\alice}}~\scomment{Sync}
      \\
      \eoutput{x_2}{\bob}
    \end{program}
  \)
  \label{fig:example-concrete}
\end{subfigure}
 \caption{Compiling the Millionaires' Problem}
  \label{fig:example}
\end{figure}

\section{Overview}

We illustrate compilation via the classic Millionaires' Problem~\cite{yao82},
expressed as the source program in \cref{fig:example-source}.
Here, \alice and \bob learn who is richer without revealing their net
worth to each other.
To do so, the program collects inputs from \alice and \bob
representing their net worth (lines~1 and~2);
compares these (line~3),
and outputs the result to \alice and \bob (lines~4 and~5).

\subsection{Information Flow Control}

Source programs prevent insecure information flows using a security type
system~\cite{sm-jsac}, which assigns a \emph{label} to every variable.
Labels track the \emph{confidentiality} and \emph{integrity} of data.
Our security type system follows prior work~\cite{msz04,nmifc}
in using \emph{downgrading mechanisms}---\kdeclassify and \kendorse
expressions---to
selectively allow information flows that would otherwise be deemed insecure.
As in prior work~\cite{nmifc},
the type system constrains these downgrading mechanisms
to prevent improper usage. These constraints turn out to be crucial.

In \cref{fig:example-source}, the \kdeclassify expression explicitly allows revealing the
result of the comparison $a < b$ to \alice and \bob, which is by default disallowed
since the computation reads secrets from both parties.
Dually, the \kendorse expressions allow untrusted data coming from
\alice and \bob to influence the output from the comparison,
which must be trusted since the value is output to both parties.

Downgrading require explicit \emph{source} and \emph{target} labels.
\Cref{fig:example-source} suppresses these labels for the \kendorse
expressions, but shows the \kdeclassify expression that moves from
$\iflabel{A \wedge B}$ to $\iflabel{A \sqcap B}$.
Label $\iflabel{A \wedge B}$ is too secret for either \alice or \bob
to see the value; label $\iflabel{A \sqcap B}$ allows \emph{both} parties
to see it.

\subsection{Compilation}

Source programs serve as \emph{specifications} of intended behavior,
and correspond to \emph{ideal functionalities} from Universal Composability~\cite{UC}.
Source programs act as trusted third parties, perfectly and securely executing
the program on behalf of the involved hosts.
The source language is high-level by design, and its simple, sequential
semantics facilitate reasoning.
Our compiler generates a concurrent distributed program that correctly
implements the same behavior.

\Cref{fig:compilation-overview} gives an overview of the compilation pipeline.
First, \emph{protocol synthesis} compiles the source program into
a \emph{choreography}, a single, centralized program that represents a
distributed computation between many hosts.
In addition to \emph{parties} such as \alice and \bob,
choreographies may mention \emph{idealized hosts} such as \MPCalicebob,
which represents a maliciously secure multiparty computation protocol between
\alice and \bob.
Idealized hosts can perform computations that require high confidentiality
or integrity.
Next, \emph{endpoint projection},
a standard procedure in choreographic programming~\cite{Montesi13, FilipeM20},
partitions the choreography into a distributed program,
where hosts run in parallel and interact via message passing.
The distributed program still contains idealized hosts,
so it corresponds to a \emph{hybrid program} in UC.
Finally, \emph{cryptographic instantiation} replaces idealized hosts with
concrete cryptographic algorithms.

\Cref{fig:example-choreography} shows a choreography where
\alice and \bob perform their respective \kinput and \koutput statements,
while \MPCalicebob does the comparison.
Explicit communication statements move data between hosts.
Choreographies have \emph{concurrent} semantics, so
statements at different hosts may step out of program order.
The penultimate line
has \emph{synchronization} between \alice and \bob: \bob must
wait on an input from \alice before delivering output.
This synchronization step is necessary for the distributed program to match the
sequential source program, in which \bob's output happens after \alice's.
\Cref{fig:example-target} shows the distributed program obtained by
projecting the choreography in \cref{fig:example-choreography}.
Endpoint projection converts communication statements
to \ksend/\kreceive pairs,
and projects local computations to their corresponding hosts.
Finally, \cref{fig:example-concrete} shows the result of cryptographic
instantiation.
Each \ksend/\kreceive statement that interacts with \MPCalicebob is
replaced with a call to a cryptographic library.

\subsection{Defining Correctness}
\label{sec:overview-threat-model}

Our main contribution is a proof that compilation is correct.
A correct compiler preserves properties of source programs in generated
target programs.
For generality, we demand that the compiler
preserve all \emph{hyperproperties}~\cite{cs08} guaranteed by the source program.
Hyperproperties capture many common notions of security,
including secure information flow.

Formally, preserving all hyperproperties is defined via
\emph{robust hyperproperty preservation} (RHP)~\cite{abateB0HPT19}.
Following prior observations~\cite{rhp-uc, rhp-uc-extended},
we do not prove RHP directly, but instead prove \emph{simulation},
which we show implies RHP in our framework.
Simulation requires every attack by an adversary against the
target program to be possible against the source program.
Ideally, the source program is ``obviously secure,''
meaning it has straightforward, sufficiently abstract semantics
and a narrow attack surface.
In contrast, the target program faithfully models
real code and has a realistic attack surface.

As with UC, our framework abstracts concrete cryptographic
mechanisms as idealized hosts, yielding an \emph{extensible}
proof that is generic over the set of supported cryptographic mechanisms.
Indeed, we sketch how our framework may be embedded into UC
to leverage the UC composition theorem~\cite{UC};
using the composition theorem, we can instantiate idealized hosts with
cryptographic mechanisms proven secure separately.

\paragraph{Threat Model}
Simulation demands we carefully define the capabilities of adversaries for
each language in the pipeline.

Adversaries are characterized by two sets of labels \public and \untrusted
representing \emph{public} and \emph{untrusted} labels.
These label sets induce a partitioning of hosts into three sets:
\emph{honest} (secret and trusted),
\emph{semi-honest} (public and trusted),
and \emph{malicious} (public and untrusted).
Intuitively, malicious hosts are fully controlled by the adversary,
and semi-honest hosts follow the protocol
but leak all their data to the adversary~\cite{Lindell17}.

We say a host is \emph{dishonest} if it is semi-honest or malicious,
and \emph{nonmalicious} if it is honest or semi-honest.

Source programs are fully sequential, so adversaries do not control scheduling.
Moreover, adversaries cannot read or change intermediate data within a source
program.
However, adversaries can read messages from \kinput/\koutput expressions
involving dishonest hosts,
read the results of \kdeclassify expressions with public target labels,
and change the results of \kendorse expressions with untrusted source labels.

Choreographies and hybrid distributed programs have the same semantics,
so their adversaries have equal power.
These programs are concurrent and the adversary controls all scheduling.
The adversary also fully controls malicious hosts
and can read messages involving at least one semi-honest host.
However, the adversary cannot carry out computational attacks,
since cryptographic mechanisms are modeled as idealized hosts.

The adversary can view all message \emph{headers} (source and destination),
but not necessarily message \emph{content}.
The adversary may not drop, duplicate, or modify messages.
This abstraction of secure channels can be realized by standard techniques,
such as TLS~\cite{TLS}.
In our model as in most models of cryptographic protocols~\cite{UC, RSIM},
the adversary can exploit \emph{timing} and \emph{progress} channels
since it controls scheduling.
These channels make secret control flow insecure:
any discrepancy in timing or progress behavior between different
control-flow paths can be detected by the adversary.
Therefore, we only prove security for programs that make control flow
decisions based on public, trusted data.

\subsection{Roadmap of Correctness Proof}

To define and prove our compilation pipeline secure, we make use of
simulation (\cref{sec:simulation}), which defines a relation $\simulates$
between semantic configurations. Intuitively, $\config_1 \leq \config_2$
means that any adversary against configuration $\config_1$ is no more powerful than
an equivalent adversary against $\config_2$. Crucially, the attacker often has more
choices in $\config_1$ (e.g., low-level scheduling decisions); the role of
simulation is to show that these extra choices are benign, and thus $\config_1$
is as secure as $\config_2$.

Our proof exploits the transitivity of simulation, using
multiple intermediate simulations depicted in
\Cref{fig:compilation-overview}. (To follow the flow of compilation, the figure
uses $\simulatedBy$, which is defined as expected.)
\paragraph{Correctness of Protocol Synthesis}
We first prove in \cref{sec:correctness-of-synthesis} that
protocol synthesis is correct:
sequential source programs (e.g., \cref{fig:example-source})
are simulated by their choreographies (e.g., \cref{fig:example-choreography}).

There is a wide semantic gap between source programs and choreographies:
while source programs are sequential and use \kdeclassify/\kendorse expressions
to interface with the adversary, choreographies are concurrent and allow the
adversary to read and corrupt data controlled by dishonest hosts.
To bridge this gap, we break the protocol synthesis proof itself into
three steps (\cref{fig:proof-overview}).
These steps allow us to reason separately about the
semantic features of choreographies.

Aside from employing transitivity, each proof step $\config_1 \simulates \config_2$
requires us to define an appropriate \emph{simulator} $\simulator(\cdot)$ such
that for any adversary $\adv$, $\config_1$ running alongside $\adv$ is identical
in behavior to $\simulator(\adv)$ running alongside $\config_2$.
Thus, after defining the simulator as a labeled transition system,
we show that the two resulting configurations are bisimilar using standard
information-flow arguments (e.g., by defining an appropriate notion of
\emph{low-equivalence} between the configurations).

\paragraph{Correctness of Protocol Synthesis}
Second, we prove in \cref{sec:correctness-of-projection} that choreographies are
simulated by their corresponding distributed programs (e.g.,
\cref{fig:example-concrete}). Our proof (\cref{thm:correctness-of-projection}) largely follows
the choreographic programming literature~\cite{Montesi23, Montesi13, FilipeM20,
FilipeMP22, HirschG22},
but deals with extra complications arising from an adversary who may corrupt hosts.

\paragraph{Cryptographic Instantiation}
Finally, we sketch in \cref{sec:uc}
how hybrid distributed programs are simulated by concrete
distributed programs, which make use of actual cryptographic mechanisms.
In particular,
we show how to embed our framework in the SUC~\cite{SUC} framework,
which in turn embeds into the full UC framework.
We can then leverage the UC composition theorem to instantiate idealized hosts
one at a time, appealing to existing correctness proofs for cryptographic
mechanisms.
 
\section{Semantic Framework}
\label{sec:simulation}

\begin{figure}
  \begin{syntax}
  \categoryFromSet[Hosts]{\h}{\h!}
  \categoryFromSet[Endpoints]{\ch}{\ch! = \set{\advprot, \envprot} \cup \h!}
  \categoryFromSet[Values]{\val}{\val! = \set{\vfalse, \dots}}

  \category[Messages]{\msg \in \msg!}
  \alternative{\mglobal{\ch_1}{\ch_2}{\val}}

  \category[Actions]{\action \in \action!}
  \alternative{\ainput{\msg}}
  \alternative{\aoutput{\msg}}
\end{syntax}

\begin{rules}{\actor{\action} = \ch}
  \actor{\aglobalreceive{\ch_1}{\ch_2}{\val}}
  =
  \ch_2

  \actor{\aglobalsend{\ch_1}{\ch_2}{\val}}
  =
  \ch_1
\end{rules}
   \caption{Syntax of messages and actions.}
  \label{fig:action-syntax}
\end{figure}

We capture the semantics of programs using labeled transition systems (LTSs),
where labels are \emph{actions} \action drawn from the grammar in
\cref{fig:action-syntax}.
An action \action is either the input $\ainput{\msg}$ or the
output $\aoutput{\msg}$ of a message \msg.
A message \msg specifies the endpoints $\ch_1$ and $\ch_2$ of communication
and carries a value \val.
An endpoint is either a host \h, the adversary \advprot, or the external environment
\envprot.
Values are drawn from an arbitrary set \val!, which we assume contains
at least $\vfalse$.
We define $\actor{\action}$ as the host performing \action:
the sender performs output actions and the receiver performs input actions.
Internal steps are represented as self-communication $\aglobalinternal{\h}$,
which identifies a host \h making progress without requiring a new syntactic
form.

\paragraph{Configurations and Parallel Composition}

\begin{figure}
  \begin{syntax}
  \abstractCategory[Processes]{\process}

  \category[Configurations]{\config}
  \alternative{\process_1 \parallel \dots \parallel \process_n}
\end{syntax}
 
  \bigskip
  \begin{rules}{\config \stepsto{\action} \config'}
\inferrule
    [\named{\config-Input}{config:input}]
    {
      \forall i \centerdot
        \process_i \stepsto{\ainput{\msg}} \process_i'
    }
    {
      \process_1 \parallel \dots \parallel \process_n
      \stepsto{\ainput{\msg}}
      \process_1' \parallel \dots \parallel \process_n'
    }

\inferrule
    [\named{\config-Output}{config:output}]
    {
      \process_i \stepsto{\aoutput{\msg}} \process_i'
      \\
      \forall j \neq i \centerdot
        \process_j \stepsto{\ainput{\msg}} \process_j'
    }
    {
      \process_1 \parallel \dots \parallel \process_n
      \stepsto{\aoutput{\msg}}
      \process_1' \parallel \dots \parallel \process_n'
    }
\end{rules}
   \caption{Syntax and semantics of configurations.}
  \label{fig:parallel}
\end{figure}

A configuration, \config, is the parallel composition of a finite set of
processes $\process_i$, which are arbitrary LTSs.
Following prior work~\cite{LynchT87,rafnssonS14}, processes must be \emph{input-total}:
for every state $\process$ and input message $\ainput{\msg}$,
there exists a state $\process'$ such that
$\process \stepsto{\ainput{\msg}} \process'$.
\Cref{fig:parallel} gives the semantics of configurations.
A configuration \config steps with an input $\ainput{\msg}$ if all processes
in \config do, and steps with an output $\aoutput{\msg}$ if one of the
processes outputs \msg, and the rest input \msg.

\paragraph{Adversaries}

As with processes, an adversary \adv or \simulator is an arbitrary LTS.
The rules for running an adversary in parallel with a configuration
are the same as in \cref{fig:parallel}.
In contrast to processes, adversaries are \emph{not} input-total,
which enables adversaries to control scheduling:
to schedule an endpoint $\ch_1$, \adv refuses to step with actions of the form
$\aglobalreceive{\ch_1'}{\ch_2}{\msg}$ where $\ch_1' \neq \ch_1$,
but accepts actions $\aglobalreceive{\ch_1}{\ch_2}{\msg}$.

Due to the definition of parallel composition,
a copy of every message from the configuration and the environment is delivered
to the adversary;
and any output of the adversary is delivered to the configuration and the
environment.

However, the adversary can only read a message if at least one endpoint is
dishonest, and can only forge messages from malicious hosts.

\begin{definition}[Adversary Interface]
  \label{def:adversary-interface}
  For all \adv:
  \begin{enumerate}
    \item
      If \( \ch_1 \) and \( \ch_2 \) are honest,
      then
      \( \adv \stepsto{\aglobalreceive{\ch_1}{\ch_2}{\val_1}} \adv' \)
      if and only if
      \( \adv \stepsto{\aglobalreceive{\ch_1}{\ch_2}{\val_2}} \adv' \)
      for all \( \val_1 \) and \( \val_2 \).

    \item
      If
      \( \adv \stepsto{\aglobalsend{\ch_1}{\ch_2}{\val}} \),
      then either \( \ch_1 = \advprot \) or \( \ch_1 \) is malicious.
  \end{enumerate}
\end{definition}

\paragraph{Determinism}

To match UC, the adversary must resolve all nondeterminism,
so that $\adv \parallel \config$ is deterministic.
We ensure determinism with the following restrictions.
\begin{itemize}
  \item
    Configurations and adversaries are \emph{internally deterministic}:
    if
    \( \process \stepsto{\action} \process_1 \)
    and
    \( \process \stepsto{\action} \process_2 \),
    then
    \( \process_1 = \process_2 \).

  \item
    Adversaries are \emph{output deterministic}:
    if
    \( \adv \stepsto{\aoutput{\msg_1}} \)
    and
    \( \adv \stepsto{\aoutput{\msg_2}} \),
    then
    \( \msg_1 = \msg_2 \).

  \item
    Configurations are \emph{output deterministic per channel}:
    if
    \( \process \stepsto{\aoutput{\msg_1}} \),
    \( \process \stepsto{\aoutput{\msg_2}} \),
    and
    \( \actor{\aoutput{\msg_1}} = \actor{\aoutput{\msg_2}} \),
    then
    \( \msg_1 = \msg_2 \).

  \item
    Adversaries are \emph{channel selective}:
    if
    \( \process \stepsto{\ainput{\msg_1}} \)
    and
    \( \process \stepsto{\ainput{\msg_2}} \),
    then
    \( \actor{\ainput{\msg_1}} = \actor{\ainput{\msg_2}} \).
\end{itemize}

\paragraph{Simulation}

Simulation determines when a configuration $\config_2$ securely realizes
configuration $\config_1$:
that is, if \emph{every} adversary \adv interacting with
$\config_2$ can be simulated by another adversary \simulator (with the same
interface) running against $\config_1$~\cite{UC}.
The latter adversary is called a \emph{simulator}.

\begin{definition}[Simulation]
  \label{def:simulation}
  \( \config_1 \) is simulated by \( \config_2 \),
  written \( \config_1 \simulatedBy \config_2 \),
when the two systems are indistinguishable to
\emph{any} external environment:
  \[
    \forall \adv \centerdot
    \exists \simulator \centerdot
      \envtraces{\simulator \parallel \config_1}
      =
      \envtraces{\adv \parallel \config_2}
  \]
\end{definition}

Here, \( \envtraces{\config} \) is the set of \emph{traces} of \config but
containing only the actions that communicate with the environment.
Given trace
$\tr = \action_1, \dots, \action_n$, we have
\(
  \envtraces{\config}
  =
  \setdef{\restrict{\tr}{\envprot}}{\config \stepsto{\tr}}
\),
where restriction $\restrict{\tr}{\envprot}$ removes all actions in \tr where
neither the source nor the destination is \envprot.

Our definition of simulation guarantees \emph{perfect}
(i.e., information\hyp{}theoretic) security.
In \cref{sec:uc}, we discuss how to transfer our results to the SUC
framework, which is based on computational security.

\section{Specifying Security Policies}
\label{sec:labels}

To succinctly capture both security policies and the adversary's power,
we use a label model that can describe confidentiality
and integrity simultaneously~\cite{ml-tosem, dclabels, flam}.

A security label $\lbl \in \lbl!$ is a pair of the form $\pair{\cc}{\ic}$
where \cc and \ic are elements of an arbitrary bounded distributive lattice \p!.
Here, \cc describes confidentiality and \ic describes integrity.
Elements of \p! are called \emph{principals}.
Principals can be thought of as negation-free boolean formulas over
a set $\set{\palice, \pbob, \pchuck, \dots}$ of \emph{atomic principals}.

The \emph{acts-for} relation ($\actsfor$) orders principals by authority,
and coincides with logical implication: for example,
$\p \wedge \q \actsfor \p$ and $\q \actsfor \p \vee \q$.
The most powerful principal is $\strongest$ and the least powerful,
$\weakest$, so we have $\strongest \actsfor \p \actsfor \weakest$
for any principal $\p$.

We lift $\wedge$, $\vee$, and $\actsfor$ to labels in the obvious pointwise manner.
Whenever appropriate, we write $\p$ for the security label $\pair{\p}{\p}$.
Confidentiality and integrity projections
$\confidentiality{\lbl}$ and $\integrity{\lbl}$
completely weaken the other component of a label:
$
  \confidentiality{\pair{\cc}{\ic}}
  =
  \pair{\cc}{\weakest}
$
and
$
  \integrity{\pair{\cc}{\ic}}
  =
  \pair{\weakest}{\ic}
  .
$

As in FLAM~\cite{flam},
the authority ordering on principals defines secure information flow.
Flow policies become more restrictive as they become
\emph{more} secret and \emph{less} trusted:
\begin{align*}
  \pair{\cc_1}{\ic_1}
  \flowsto
  \pair{\cc_2}{\ic_2}
  &\iff
  \cc_2 \actsfor \cc_1
  \text{ and }
  \ic_1 \actsfor \ic_2
  \\
  \pair{\cc_1}{\ic_1}
  \join
  \pair{\cc_2}{\ic_2}
  &=
  \pair{\cc_1 \wedge \cc_2}{\ic_1 \vee \ic_2}
  \\
  \pair{\cc_1}{\ic_1}
  \meet
  \pair{\cc_2}{\ic_2}
  &=
  \pair{\cc_1 \vee \cc_2}{\ic_1 \wedge \ic_2}
\end{align*}
The least restrictive information flow policy is $\ifbottom$ (``public
trusted''), describing information that can be used anywhere, while
the most restrictive is $\iftop$ (``secret untrusted'').

\begin{toappendix}
  \label{sec:ifc-details}

  Formally, an attack is specified by picking two sets of \emph{principals}:
  public principals $\p* \subseteq \p!$ and
  untrusted principals $\q* \subseteq \p!$.
Some common-sense conditions must hold on the sets \p* and \q*~\cite{nmifc}.
We state the conditions for \p* but they apply equally to \q*.
The adversary always controls the weakest principal,
  but never controls the strongest:
  $\weakest \in \p*$ and $\strongest \not\in \p*$.
If the adversary controls a principal, then it controls all weaker principals:
  if $\p \in \p*$ and $\p \actsfor \q$, then $\q \in \p*$.
Attacking principals may collude:
  if $\p, \q \in \p*$, then $\p \wedge \q \in \p*$.
Combining secret/trusted principals leads to secret/trusted principals:
  if $\p \vee \q \in \p*$, then either $\p \in \p*$ or $\q \in \p*$.
  
Together, these conditions imply that \p* and \q* are sensible truth assignments
  to elements \p!: sensible in the sense that they play nicely with $\wedge$ and
  $\vee$.\footnote{For those familiar with order theory,
    \p* and \q* must be \emph{prime filters} of \p!.}

  We can derive the set of public/secret and trusted/untrusted \emph{labels}
  from \p* and \q*:
  \begin{mathpar}
    \public
    =
    \setdef
      {\pair{\cc}{\ic} \in \lbl!}
      {\cc \in \p*}

    \secret = \lbl! \setminus \public

    \trusted
    =
    \setdef
      {\pair{\cc}{\ic} \in \lbl!}
      {\ic \not\in \q*}

    \untrusted = \lbl! \setminus \trusted
    .
  \end{mathpar}

  Additionally, we require that attacks compromise at least as much confidentiality
  as integrity.

  \begin{definition}[Valid Attack]
    \label{def:valid-attack}
    Attack \( \pair{\p*}{\q*} \) is valid if all untrusted principals are public:
    \( \q* \subseteq \p* \).
  \end{definition}
\end{toappendix}

\subsection{Authority of Hosts}
\label{sec:host-auth}

Protocol synthesis places computations on hosts that have enough authority
to securely execute them.
The authority of each host \h is captured with a label $\labelof{\h}$~\cite{zznm02}.
For our example, we take $\labelof{\alice} = \iflabel{A}$
and $\labelof{\bob} = \iflabel{B}$.
Following an insight from Viaduct~\cite{viaduct-pldi21},
idealized hosts like $\MPCalicebob$ have a derived label that
conservatively approximates the
security guarantees of the cryptographic mechanism.
For maliciously secure MPC, a reasonable label is
\(
  \labelof{\MPCalicebob}
  =
  \labelof{\alice} \wedge \labelof{\bob}
  =
  \iflabel{A \wedge B}
\),
meaning that \MPCalicebob may view secrets of \alice and \bob,
and also has enough integrity to compute values for them.

\subsection{Capturing Attacks with Labels}
\label{sec:labels-attack}

The power of the adversary is determined by partitioning labels \lbl! across
the two axes: public/secret and trusted/untrusted;
we denote these sets as \public/\secret and \trusted/\untrusted, respectively.
We only consider sets that form \emph{valid attacks}~\cite{nmifc}.
Intuitively, a valid attack requires that all untrusted labels are public, so
that the adversary cannot modify secret data;
we define valid attacks formally in \apxref{sec:ifc-details}.
The rest of our development is parameterized over a valid partitioning of labels.

Recalling the threat model from \cref{sec:overview-threat-model},
an honest host has a secret, trusted label ($\honest{\h}$);
a semi-honest host has a public, trusted label ($\semihonest{\h}$);
and a malicious host has a public, untrusted label ($\malicious{\h}$).
A host with a secret, untrusted label does not make any sense:
an untrusted host is fully controlled by the adversary,
so it cannot hide information from the adversary.
We rule out such corruptions by requiring all host labels to be
\emph{uncompromised}~\cite{zsm19}.
A valid partitioning never classifies an uncompromised label as secret and
untrusted.

\begin{definition}[Uncompromised Label]
  Label \label{def:uncompromised-label}
  \( \lbl = \pair{\cc}{\ic} \) is uncompromised,
  written \( \uncompromised{\lbl} \),
  if it is at least as trusted as it is secret:
  \( \ic \actsfor \cc \).
\end{definition}

\begin{theoremrep}
  Under a valid attack, if \( \uncompromised{\lbl} \),
  then we have \( \lbl \not\in \secret \cap \untrusted \).
\end{theoremrep}
\begin{proof}
  Let
  \( \pair{\p*}{\q*} \) be a valid attack
  and
  \( \lbl = \pair{\p}{\q} \).
Assume
  \( \uncompromised{\lbl} \)
  and
  \( \lbl \in \untrusted \).
By definition, we have \( \q \in \q* \).
Unfolding \( \uncompromised{\lbl} \),
  we have \( \q \actsfor \p \), and since \q* is upward closed,
  we have \( \p \in \q* \).
Finally, \( \p \in \q* \subseteq \p* \),
  so \( \lbl \in \public \).
\end{proof}
 
\section{Protocol Synthesis}

The first program transformation, protocol synthesis, takes a sequential
source program to a choreography.

\subsection{Source Language}

\begin{figure}
  \begin{syntax}
  \groupleft{
    \categoryFromSet[Variables]{\x}{\x!}
    \hspace{1em}\hfill
    \categoryFromSet[Labels]{\lbl}{\lbl!}
    \hspace{1em}\hfill
    \categoryFromSet[Operators]{\f}{\f!}
  }

  \separate

  \category[Atomic Expr.]{\aexp}
  \alternative{\val}
  \alternative{\x}

  \category[Expressions]{\e}
  \alternative{\eapplyop{\f}{\manyargs{\aexp}{n}}}
  \\
  \alternative{\edeclassify{\aexp}{\fromlbl}{\tolbl}}
  \\
  \alternative{\eendorse{\aexp}{\fromlbl}{\tolbl}}
  \\
  \alternative{\einput{\h}}
  \alternative{\eoutput{\aexp}{\h}}

  \separate

  \category[Statements]{\s}
  \alternative{\sleth{\x}{\h}{\e}{\s}}
  \\
  \alternative{\sifh{\aexp}{\h}{\s_1}{\s_2}}
  \\
  \alternative{\sskip}

  \separate

  \categoryFromSet[Buffers]{\buffer}{\ch! \times \ch! \to \val!^*}

  \category[Processes]{\process}
  \alternative{\proc{\h!}{\buffer}{\s}}
\end{syntax}
   \caption{Syntax of the source language.}
  \label{fig:source-syntax}
\end{figure}

\Cref{fig:source-syntax} gives the syntax of source programs.
The language supports an abstract set of operators \f over values.
We distinguish pure, atomic expressions \aexp
from expressions \e that may have side effects.
The \kdeclassify expression marks locations where
private data is explicitly allowed to flow to public data.
Similarly, \kendorse marks where untrustworthy data may
influence trusted data.
The \kinput/\koutput expressions allow programs to interact with the external
environment~\cite{oneillCC06, clarkH08}.

Statement $\sleth{\x}{\h}{\e}{\s}$ performs the local
computation \e at host \h, binds the result to variable \x,
and continues as \s.
In source programs, \h is only relevant for \kinput and \koutput expressions:
we write
\(\dlethhead{\x}{\h}{\einput{\h}}\)
and
\(\dlethhead{\x}{\h}{\eoutput{\aexp}{\h}}\)
in contrast to our example in \cref{fig:example-source},
where we write
\(\dlethead{\x}{\einputh{\h}}\)
and
\(\dlethead{\x}{\eoutputh{\aexp}{\h}}\).
For all other expressions, \h is \idealprot,
a single fully trusted host.
Representing source programs with host annotations
allows smoothly extending the language later on.

A source-program configuration is a logically centralized process
$\proc{\h!}{\buffer}{\s}$.
The component \h! indicates the process acts for all hosts.
The second component, \buffer, is a \emph{buffer} mapping pairs of endpoints to
first-in-first-out queues of values.
Processes buffer input so that output on other processes
(the adversary and the environment) is nonblocking.

\subsection{Choreography Language}

\begin{figure}
  \begin{syntax}
  \category[Expressions]{\e}
  \alternative{\dots}
  \alternative{\ereceive{\h}}
  \alternative{\esend{\aexp}{\h}}

  \category[Statements]{\s}
  \alternative{\dots}
  \alternative{\smove{\h_1}{\aexp}{\h_2}{\x}{\s}}
  \alternative{\sselect{\h_1}{\val}{\h_2}{\s}}

  \category[Processes]{\process}
  \alternative{\proc{\h* \subseteq \h!}{\buffer}{\s}}
\end{syntax}
   \caption{
    Syntax of choreographies as an extension to source syntax
    (\cref{fig:source-syntax}).
The $\ksend/\kreceive$ expressions are only relevant for the
    security proof.}
  \label{fig:choreography-syntax}
\end{figure}

Choreographies are centralized representations of distributed computations.
Unlike source programs, they make explicit the location of
computations, storage, and communication.

\Cref{fig:choreography-syntax} gives the syntax of choreographies,
which we present as an extension of the source syntax.
Host annotations on \klet and \kif statements are no longer restricted to
\idealprot and can be any host $\h \in \h!$.
The \emph{global communication} statement $\smove{\h_1}{\aexp}{\h_2}{\x}{\s}$
represents host $\h_1$ sending the value of \aexp to $\h_2$,
which stores it in variable \x.
The \emph{selection} statement $\sselect{\h_1}{\val}{\h_2}{\s}$
communicates control flow decisions,
and is used to establish \emph{knowledge of choice}~\cite{Montesi23, Montesi13}.
We extend expressions with \ksend and \kreceive,
however, the compiler never generates choreographies with these expressions;
they are only used in proofs to model malicious corruption
(\cref{sec:modeling-dishonesty}).

As for source programs, configurations are single
processes $\proc{\h*}{\buffer}{\s}$, but \h* may be a strict subset of
\h! and does not contain the ideal process \idealprot.

\subsection{Operational Semantics of Choreographies}
\label{sec:operational-semantics}

\begin{figure}
  \begin{rules}{\h \says \e \idealstepsto{\action} \val}
\inferrule
    {
      \issecret{\fromlbl}
      \\
      \ispublic{\tolbl}
    }
    {
      \h \says
      \edeclassify{\val}{\fromlbl}{\tolbl}
      \idealstepsto{\aleak{\h}{\val}}
      \val
    }

\inferrule
    {
      \isuntrusted{\fromlbl}
      \\
      \istrusted{\tolbl}
    }
    {
      \h \says
      \eendorse{\val}{\fromlbl}{\tolbl}
      \idealstepsto{\amaul{\h}{\val'}}
      \val'
    }

\inferrule
    {
      \nonmalicious{\h}
    }
    {
      \h \says
      \einput{\h}
      \idealstepsto{\aglobalreceive{\envprot}{\h}{\val}}
      \val
    }

\inferrule
    {
      \nonmalicious{\h}
    }
    {
      \h \says
      \eoutput{\val}{\h}
      \idealstepsto{\aglobalsend{\h}{\envprot}{\val}}
      \vunit
    }
\end{rules}
 
  \bigskip
  \begin{judgments}
  \judgment{\s \realstepsto{\action} \s'}
\end{judgments}
\begin{mathpar}
\smove{\h_1}{\val}{\h_2}{\x}{\s}
  \realstepsto{\aglobalsend{\h_1}{\h_2}{\val}}
  \subst{\x}{\val}{\s}

\sselect{\h_1}{\val}{\h_2}{\s}
  \realstepsto{\aglobalsend{\h_1}{\h_2}{\val}}
  \s
\end{mathpar}
 
  \smallskip
  \begin{rules}{\s \polystepsto!{\action} \s'}
\inferrule
    {
      \s \polystepsto!{\action} \s'
      \\
      \actor{\action} \notin \hosts{\ectx}
    }
    {
      \sleth{\x}{\h}{\e}{\s}
      \polystepsto!{\action}
      \sleth{\x}{\h}{\e}{\s'}
    }
\end{rules}
 \caption{Select ideal, real, and concurrent stepping rules.}
  \label{fig:select-stepping}
\end{figure}

Following \cref{sec:simulation}, we give operational semantics to programs
using labeled transition systems.
Since choreographies strictly extend the syntax of source programs,
it suffices to define a semantics for choreographies.

Following \cref{fig:compilation-overview}, we define two transition relations:
\emph{ideal} stepping $\idealstepsto{}$ gives meaning to source programs
and to idealized choreographies (an intermediate language for our simulation proof),
and \emph{real} stepping $\realstepsto{}$ gives meaning to choreographies
and distributed programs.
Additionally, we lift ideal and real stepping to concurrent versions,
written $\idealstepsto!{}$ and $\realstepsto!{}$,
to capture the concurrent semantics of choreographies.
\Cref{fig:select-stepping} gives a selection of key rules;
we defer full definitions to \apxref{sec:operational-semantics-details}.

\subsubsection{Ideal Semantics}

We write $\h \says \e \idealstepsto{\action} \val$ to mean expression \e evaluates
to value \val at host \h with action \action.
We assume operators are total;
partial operators (like division) can be made total using defaults.
Formally, we give meaning to operator application assuming a denotation function
$\evaluate : \f! \times \val!^{*} \to \val!$.
We model \kdeclassify and \kendorse expressions as \emph{interactions} with
the adversary endpoint \advprot.
When a value is declassified from a secret label to a public one,
the program \emph{outputs} the value to \advprot.
Dually, when a value is endorsed from an untrusted label to a trusted one,
the program takes \emph{input} from \advprot, and uses that value instead.
When the confidentiality/integrity of the value does not change,
these expressions act as the identity function and take internal steps.
The \kinput/\koutput expressions communicate with the environment endpoint \envprot,
except on malicious hosts;
there, they step internally and always return \vunit.
In source programs and idealized choreographies,
\kreceive/\ksend expressions only communicate with malicious hosts;
we suppress them (they take internal steps and always return \vunit)
to give less power to the adversary in the ideal setting.

For statements,
we write $\s \idealstepsto{\action} \s'$ to mean statement \s
steps to $\s'$ with action \action.
Statement stepping rules are as expected:
\klet statements step using substitution,
\kif statements pick a branch based on their conditional,
and communication and selection statements step internally,
naming the ``sending host'' as the host performing the action.

\subsubsection{Real Semantics}

Real stepping rules modify ideal stepping rules.
The \kdeclassify/\kendorse expressions always step internally instead of
communicating with \advprot.
The \kreceive/\ksend expressions communicate a value with the specified host.
Finally, communication and selection statements step with a visible action
instead of internally.

\subsubsection{Concurrent Lifting for Choreographies}
\label{sec:concurrent-lifting}

Concurrent stepping rules, written $\s \polystepsto!{\action} \s'$,
lift an underlying statement-stepping judgment
($\idealstepsto{}$ or $\realstepsto{}$).
Concurrent stepping allows choreographies to step statements at different hosts
out of program order as long as there are no dependencies between the hosts,
and is the standard way choreographies model the behavior of a
distributed system~\cite{Montesi23}.
The key rule allows skipping over \klet statements to step a statement
in the middle of a program.
This rule requires the actor of the performed action to be different from the
hosts of the statements being skipped over,
matching the behavior of target programs where code running on a single
host is single-threaded.

\paragraph{Synchronous vs. Asynchronous Choreographies}

When skipping over \klet statements,
requiring only the \emph{actor} to be missing from the context
leads to an \emph{asynchronous} semantics~\cite{async-choreo}.
In a \emph{synchronous} setting, the side condition would require both
endpoints to be missing:
$\hosts{\action} \cap \hosts{\ectx} = \emptyset$.
Consider the following program:
\begin{program}
  \slethhead{\x_1}{\alice}{\einput{\alice}}
  \\
  \smovehead{\bob}{\vunit}{\alice}{\x_2}
\end{program}
\alice is waiting on an input,
so is not ready to receive from \bob.
In a synchronous setting, these statements must execute in program order
since \bob can only send if \alice is ready to receive.
In an asynchronous setting, sends are nonblocking, so the second statement
can execute first.

\subsubsection{Processes}
\label{sec:process-semantics}

A buffer \buffer behaves as a FIFO queue for each channel:
it can input a message by appending the received value at the end of the
corresponding queue,
and can output the value at the beginning of any queue.
Buffers guarantee in-order delivery within a single channel $\ch_1 \ch_2$,
but messages across different channels may be reordered.
A process \process forwards its input to its buffer if the message is
addressed to a relevant host; otherwise, \process discards the message.
A process takes an internal step when its buffer delivers a message to its statement,
and an output step when its statement outputs.

\subsection{Compiling to Choreographies}
\label{sec:compiling-to-choreographies}

\begin{figure}
  \begin{minipage}{\linewidth}
    
\begingroup \newcommand{\sourceH}[1]{{\source{#1}}}\begin{judgments}
  \judgment{\source{\s} = \s'}
\end{judgments}
\vspace{-0.5em}
\begin{function}{\sourceH}
\case{\sleth{\x}{\h}{\e}{\s}}
  \begin{cases}
    \sleth{\x}{\h}{\e}{\sourceH{\s}}
      & \phantom{\neg}\io{\e}
    \\
    \sleth{\x}{\idealprot}{\e}{\sourceH{\s}}
      & \neg\io{\e}
  \end{cases}

\case{\smove{\h_1}{\aexp}{\h_2}{\x}{\s}}
  \subst{\x}{\aexp}{\source{\s}}

\case{\sselect{\h_1}{\val}{\h_2}{\s}}
  \sourceH{\s}

\case{\sifh{\aexp}{\h}{\s_1}{\s_2}}
  \sifh{\aexp}{\idealprot}{\sourceH{\s_1}}{\sourceH{\s_2}}

\case{\sskip}
  \sskip
\end{function}
\[
  \io{\e}
  =
  (\e = \einput{\h})
  \vee
  (\exists \aexp \centerdot \e = \eoutput{\aexp}{\h})
\]
\endgroup
   \end{minipage}
  \caption{Canonical source program from a choreography.}
  \label{fig:choreography-to-source}
\end{figure}

Instead of committing to a specific algorithm,
we give validity criteria for the output of protocol synthesis,
which generalizes our results to and beyond prior
work~\cite{zcmz03,viaduct-pldi21,efficient-mpc}.
Because a source program can be realized as many different choreographies,
protocol synthesis cannot be modeled as a function from source programs to
choreographies.
Instead, we capture a valid protocol synthesis as a
mapping from choreographies to source programs.

\begin{definition}[Valid Protocol Synthesis]
  \label{def:valid-protocol-synthesis}
  Choreography \( \s' \) is a valid result of protocol synthesis on
  source program \( \s \) if
  \( \source{\s'} = \s \),
  \( \ssecure[\emptylist]{\s'} \),
  and
  \( \ssynched[\sctx]{\s'} \)
  for some \sctx.
  
\end{definition}

\Cref{fig:choreography-to-source} defines the function $\source{\cdot}$,
which maps a choreography to its canonical source program
by removing communication and selection statements
and replacing all host annotations with \idealprot
(except those associated with \kinput and \koutput).
The judgment $\ssecure{\s}$ denotes that choreography \s has secure information
flows, and $\ssynched{\s}$ denotes \s is well-synchronized.
We define these judgments next.

\subsubsection{Information-Flow Type System}
\label{sec:if-checking}

\begin{figure*}
  \begin{minipage}{\linewidth}
    \begin{judgments}
  \judgment{\aesecure{\aexp}{\lbl}}
  \and
  \judgment{\esecure{\e}{\lbl}}
\end{judgments}
\begin{mathpar}
\inferrule
    [\named{\lbl-Value}{aesecure:value}]
    { }
    {
      \aesecure{\val}{\lbl}
    }

\inferrule
    [\named{\lbl-Variable}{aesecure:variable}]
    {
      \lbl' \flowsto \lbl
    }
    {
      \aesecure[\ctx \munion \csing{\x}{\h}{\lbl'}]{\x}{\lbl}
    }

\inferrule
    [\named{\lbl-Operator}{esecure:operator}]
    {
      \forall i \centerdot
      \aesecure{\aexp_i}{\lbl}
    }
    {
      \esecure{\eapplyop{\f}{\manyargs{\aexp}{n}}}{\lbl}
    }

\inferrule
    [\named{\lbl-Declassify}{esecure:declassify}]
    {
      \aesecure{\aexp}{\fromlbl}
      \\
      \integrity{\fromlbl} = \integrity{\tolbl}
      \\\\
      \uncompromised{\fromlbl}
      \\
      \uncompromised{\tolbl}
      \\
      \tolbl \flowsto \lbl
    }
    {
      \esecure{\edeclassify{\aexp}{\fromlbl}{\tolbl}}{\lbl}
    }

\inferrule
    [\named{\lbl-Endorse}{esecure:endorse}]
    {
      \aesecure{\aexp}{\fromlbl}
      \\
      \confidentiality{\fromlbl} = \confidentiality{\tolbl}
      \\\\
      \uncompromised{\fromlbl}
      \\
      \uncompromised{\tolbl}
      \\
      \tolbl \flowsto \lbl
    }
    {
      \esecure{\eendorse{\aexp}{\fromlbl}{\tolbl}}{\lbl}
    }

\inferrule
    [\named{\lbl-Input}{esecure:input}]
    {
      \labelof{\h} \flowsto \lbl
    }
    {
      \esecure{\einput{\h}}{\lbl}
    }

\inferrule
    [\named{\lbl-Output}{esecure:output}]
    {
      \aesecure{\aexp}{\labelof{\h}}
    }
    {
      \esecure{\eoutput{\aexp}{\h}}{\lbl}
    }

\inferrule
    [\named{\lbl-Receive}{esecure:receive}]
    {
      \integrity{\labelof{\h'}} \flowsto \lbl
    }
    {
      \esecure{\ereceive{\h'}}{\lbl}
    }

\inferrule
    [\named{\lbl-Send}{esecure:send}]
    {
      \aesecure{\aexp}{\confidentiality{\labelof{\h'}}}
    }
    {
      \esecure{\esend{\aexp}{\h'}}{\lbl}
    }
\end{mathpar}
 
    \bigskip
    \begin{judgments}
  \judgment{\ssecure{\s}}
\end{judgments}
\begin{mathpar}
\inferrule
    [\named{\lbl-Let}{ssecure:let}]
    {
      \esecure{\e}{\lbl}
      \\
      \labelof{\h} \actsfor \lbl
      \\\\
      \ssecure[\ctx \munion \csing{\x}{\h}{\lbl}]{\s}
    }
    {
      \ssecure{\sleth{\x}{\h}{\e}{\s}}
    }

\inferrule
    [\named{\lbl-Communicate}{ssecure:move}]
    {
      \esecure{\aexp}{\lbl}[\h_1]
      \\
      \labelof{\h_2} \actsfor \lbl
      \\\\
      \ssecure[\ctx \munion \csing{\x}{\h_2}{\lbl}]{\s}
    }
    {
      \ssecure{\smove{\h_1}{\aexp}{\h_2}{\x}{\s}}
    }

\inferrule
    [\named{\lbl-Select}{ssecure:select}]
    {
      \integrity{\labelof{\h_1}} \flowsto \integrity{\labelof{\h_2}}
      \\
      \ssecure{\s}
    }
    {
      \ssecure{\sselect{\h_1}{\val}{\h_2}{\s}}
    }

\inferrule
    [\named{\lbl-If}{ssecure:if}]
    {
      \aesecure{\aexp}{\ifbottom}[\h]
      \\\\
      \ssecure{\s_1}
      \\
      \ssecure{\s_2}
    }
    {
      \ssecure{\sifh{\aexp}{\h}{\s_1}{\s_2}}
    }

\inferrule
    [\named{\lbl-Skip}{ssecure:skip}]
    { }
    {
      \ssecure{\sskip}
    }
\end{mathpar}
   \end{minipage}
  \caption{Information-flow typing rules for expressions and
  statements in choreographies.}
  \label{fig:if-checking}
\end{figure*}

First, we give a type system for choreographies based on information-flow
control~\cite{zznm02,zcmz03,swift07,viaduct-pldi21} which validates that
hosts have enough authority to execute their assigned statements.

\Cref{fig:if-checking} gives the typing rules.
A label context \ctx maps a variable to its host and label, as a
set of bindings $\csing{\x}{\h}{\lbl}$.
The judgment $\esecure{\e}{\lbl}$,
means that \e at host \h has label \lbl in the context \ctx.
\Cref{aesecure:variable} ensures hosts only use variables they own.
\Cref{esecure:declassify,esecure:endorse} enforce nonmalleable information
flow control (NMIFC)~\cite{nmifc} by requiring source and target labels
to be uncompromised~\cite{zsm19, nmifc}.
NMIFC requires declassified data to be trusted, enforcing \emph{robust declassification},
and endorsed data to be public, enforcing \emph{transparent endorsement}.
These restrictions prevent the adversary from exploiting downgrades.
Enforcing NMIFC is crucial for our simulation result, which we discuss in
\cref{sec:correctness-of-ideal-execution}.

In choreographies, \kreceive/\ksend expressions model communication with
malicious hosts.
Choreographies exclude code for malicious hosts, which
exhibit arbitrary behavior;
thus, labels for \kreceive/\ksend expressions must be approximated.

\Cref{esecure:receive} ensures data coming from malicious hosts is considered
untrusted; it treats the data as fully public since we do not care about
preserving the confidentiality of malicious hosts.
\Cref{esecure:send} ensures secret data is not sent to malicious hosts;
it ignores integrity since malicious hosts are untrusted.

Statement checking rules have the form $\ssecure{\s}$;
they are largely standard~\cite{sm-jsac},
but do not track program counter labels since we require programs to only
branch on public, trusted values.
\Cref{ssecure:let,ssecure:move} check that the host storing a variable has
enough authority to do so.

This is the key condition governing secure host selection and prevents,
for example, \bob's secret data being placed on \alice,
or high-integrity data being placed on an untrusted host.
\Cref{ssecure:select} ensures that if host $\h_1$ informs $\h_2$ of a
branch being taken, then $\h_1$ has at least as much integrity as $\h_2$.
So malicious hosts cannot influence control flow on nonmalicious hosts.
Finally, \cref{ssecure:if} requires control flow to be public and trusted.

\subsubsection{Synchronization Checking}
\label{sec:synchronization}

\begin{figure*}
  \begin{minipage}{\linewidth}
    \begin{rules} {\ssynched{\s}}
\inferrule
    [\named{Sync-External}{ssynched:external}]
    {
      \external{\e}
      \\
      \ssynched[\hreset{\h}]{\s}
      \\
      \outputting{\e} \implies \hissynched{\h}
    }
    {
      \ssynched{\sleth{\x}{\h}{\e}{\s}}
    }

\inferrule
    [\named{Sync-Internal}{ssynched:internal}]
    {
      \internal{\e}
      \\
      \ssynched{\s}
    }
    {
      \ssynched{\sleth{\x}{\h}{\e}{\s}}
    }

\inferrule
    [\named{Sync-Communicate}{ssynched:move}]
    {
      \ssynched[\hsync{\h_1}{\h_2}]{\s}
    }
    {
      \ssynched{\smove{\h_1}{\aexp}{\h_2}{\x}{\s}}
    }

\inferrule
    [\named{Sync-Select}{ssynched:selection}]
    {
      \ssynched[\hsync{\h_1}{\h_2}]{\s}
    }
    {
      \ssynched{\sselect{\h_1}{\val}{\h_2}{\s}}
    }

\inferrule
    [\named{Sync-If}{ssynched:if}]
    {
      \ssynched{\s_1}
      \\
      \ssynched{\s_2}
    }
    {
      \ssynched{\sifh{\aexp}{\h}{\s_1}{\s_2}}
    }

\inferrule
    [\named{Sync-Skip}{ssynched:skip}]
    { }
    {
      \ssynched{\sskip}
    }
\end{rules}

\bigskip
\begin{judgments}
  \judgment{\hissynched{\h}}
  \and
  \judgment{\hreset{\h} = \sctx'}
  \and
  \judgment{\hsync{\h_1}{\h_2} = \sctx'}
\end{judgments}
\begin{mathpar}
  \hissynched{\h}
  =
  \forall \h' \centerdot
    \sctx(\h', \h) \flowsto \labelof{\h'} \vee \labelof{\h}

  \hreset{\h}
  =
  \mupdate{\h, \h}{\labelof{\h}}
    {\mupdate{\h, *}{\weakest}{\sctx}}

  \hsync{\h_1}{\h_2}
  =
  \mupdate{*, \h_2}
    {\sctx(*, \h_2) \wedge (\sctx(*, \h_1) \vee \labelof{\h_2})}
    {\sctx}
\end{mathpar}
   \end{minipage}
  \caption{Checking that a concurrent choreography has sequential behavior.}
  \label{fig:synchronization-checking}
\end{figure*}

Next, we define a novel \emph{synchronization} judgment, $\ssynched{\s}$,
which guarantees that all external actions in $\s$ happen in sequential
program order.
For example, any $\kendorse$ statement that happens after a $\kdeclassify$
must logically \emph{depend} on the $\kdeclassify$. Since these statements
may be run on different hosts, the $\kdeclassify$ could happen before the
$\kendorse$, violating program order. To prevent this, we require that
the host running the $\kendorse$ \emph{synchronizes} with the host running
the $\kdeclassify$.

Synchronization becomes more complex with \emph{corruption}.
For example, if \alice and \bob synchronize through another host \h
($\alice \communicate \h \communicate \bob$)
and \h is malicious, \h might give \bob the go-ahead before confirming with
\alice.
We use integrity labels to ensure synchronization even under corruption.

\Cref{fig:synchronization-checking} defines the synchronization-checking
judgment $\ssynched{\s}$.
Intuitively, a choreography is well-synchronized when for any
\emph{external} (input or output) expression \e,
a high-integrity communication path exists from \e to all
\emph{output} expressions following \e in the program order.\footnote{Input expressions are \kinput and \kendorse;
output expressions are \koutput and \kdeclassify.}
Integrity of a communication path
$\h_1 \communicate \cdots \communicate \h_n$
is determined by the hosts in the path:
\[
 \labelof{\h_1 \communicate \cdots \communicate \h_n}
 =
 \integrity{\labelof{\h_1}}
 \vee \dots \vee
 \integrity{\labelof{\h_n}}
\]
Hosts can be malicious, so each host on the path weakens integrity,
which is captured by disjunction ($\vee$).
Multiple paths between the same hosts increase integrity,
which we capture by taking the conjunction ($\wedge$) of path labels:
\[
  \labelof{\paths{\h_1}{\h_2}}
  =
  \bigwedge_{\cpath \in \paths{\h_1}{\h_2}}{
    \labelof{\cpath}
  }
\]
We track the integrity of paths using the context \sctx, which maps pairs of hosts
$\sctx(\h_1, \h_2)$ to the integrity label $\labelof{\paths{\h_1}{\h_2}}$.

\Cref{ssynched:external} checks a \klet statement that executes an external
expression \e on \h.
The continuation is checked under a context where the label of all paths
\emph{from} \h to any other host are set to $\weakest$
(this corresponds to removing the paths),
since these hosts now need to synchronize with \h.
Also, if \e is an output expression, \h must be synchronized
with all hosts through the following condition, which ensures that if
neither $\h_1$ nor $\h_2$ is malicious, a communication path exists from
$\h_1$ to $\h_2$ that could not have been influenced by the adversary:
\begin{equation}
  \label{def:is-synched}
  \labelof{\paths{\h_1}{\h_2}}
  \flowsto
  \labelof{\h_1} \vee \labelof{\h_2}
\end{equation}

\Cref{ssynched:move,ssynched:selection} update \sctx using
$\hsync{\h_1}{\h_2}$.
The function captures that a path from some $\h$ to $\h_1$ implies
there is a path from $\h$ to $\h_2$ that goes through $\h_1$.
Further, all existing paths are still valid.

\subsection{Modeling Malicious Corruption}
\label{sec:modeling-dishonesty}

\begin{figure}
\small
  \begingroup \newcommand{\corruptH}[1]{{\corrupt{#1}}}\newcommand{\ow}{\text{o/w}}\renewcommand{\nonmalicious}[1]{\istrusted{#1}}
\begin{judgments}
  \judgment{\corrupt{\s} = \s'}
\end{judgments}
\begin{function}{\corruptH}
\case{\sleth{\x}{\h}{\e}{\s}}
  \begin{cases}
    \sleth{\x}{\h}{\e}{\corruptH{\s}}
      & \nonmalicious{\h}
    \\
    \corruptH{\s}
      & \ow
  \end{cases}

\case{\smove{\h_1}{\aexp}{\h_2}{\x}{\s}}
  \begin{cases}
    \smove{\h_1}{\aexp}{\h_2}{\x}{\corruptH{\s}}
      & \nonmalicious{\h_1, \h_2}
    \\
    \sleth{\_}{\h_1}{\esend{\aexp}{\h_2}}{\corruptH{\s}}
      & \nonmalicious{\h_1}
    \\
    \sleth{\x}{\h_2}{\ereceive{\h_1}}{\corruptH{\s}}
      & \nonmalicious{\h_2}
    \\
    \corruptH{\s}
      & \ow
  \end{cases}

\case{\sselect{\h_1}{\val}{\h_2}{\s}}
  \begin{cases}
    \sselect{\h_1}{\val}{\h_2}{\corruptH{\s}}
      & \nonmalicious{\h_1, \h_2}
    \\
    \sleth{\_}{\h_1}{\esend{\val}{\h_2}}{\corruptH{\s}}
      & \nonmalicious{\h_1}
    \\
    \corruptH{\s}
      & \ow
  \end{cases}

\case{\sifh{\aexp}{\h}{\s_1}{\s_2}}
  \begin{cases}
    \sifh{\aexp}{\h}{\corruptH{\s_1}}{\corruptH{\s_2}}
      & \nonmalicious{\h}
    \\
    \bot
      & \ow
  \end{cases}

\case{\sskip}
  \sskip
\end{function}

\begin{rules}{ \corrupt{\process} = \process' }
  \corrupt{\proc{\h*}{\buffer}{\s}} =
    \proc{\setdef{\h \in \h*}{\nonmalicious{\h}}}{\buffer}{\corrupt{\s}}
\end{rules}
\endgroup
   \caption{Modeling malicious corruption.
    We write \( \istrusted{\h} \) for \( \nonmalicious{\h} \).}
  \label{fig:malicious-corruption}
\end{figure}

Malicious hosts are fully controlled by the adversary.
Following UC~\cite{UC}, we entirely remove processes that correspond to
malicious hosts in hybrid distributed programs,
and allow the adversary to forge arbitrary messages in their stead.
This is reflected in choreographies by rewriting them to elide statements
that involve malicious hosts.

\Cref{fig:malicious-corruption} defines the \emph{corruption} $\corrupt{\s}$
of a choreography \s.
The operation considers each statement in turn.
If \emph{all} hosts involved in a statement are nonmalicious,
the statement stays as is.
If \emph{all} hosts involved in a statement are malicious,
the statement is removed entirely.
Otherwise, only \emph{some} involved hosts are malicious,
and we rewrite the statement.
Communication statements become either a \ksend or a \kreceive,
depending on whether the sending or the receiving host is nonmalicious.
Selection statements are similar,
except we cannot have a malicious sending host and a nonmalicious receiving
host (\cref{ssecure:select}).
Similarly, \kif statements cannot be at a malicious host since we require
trusted control flow (\cref{ssecure:if}).
   \nosectionappendix
  \begin{toappendix}
    \appendixfor{sec:operational-semantics}
    \label{sec:operational-semantics-details}
    \begin{figure*}
  \newcommand{\figskip}{\vspace*{1.5em}}\begin{subfigure}{\linewidth}{Ideal stepping rules for expressions and statements.}
    \begin{judgments}
  \judgment{\h \says \e \idealstepsto{\action} \val}
\end{judgments}
\begin{mathpar}
\inferrule
    [\named{\e-Operator}{expression:operator}]
    {
      \val = \evaluate(\f, \manyargs{\val}{n})
    }
    {
      \h \says
      \eapplyop{\f}{\manyargs{\val}{n}}
      \idealstepsto{\aglobalinternal{\h}}
      \val
    }

\inferrule
    [\named{\e-Declassify}{expression:declassify}]
    {
      \issecret{\fromlbl}
      \\
      \ispublic{\tolbl}
    }
    {
      \h \says
      \edeclassify{\val}{\fromlbl}{\tolbl}
      \idealstepsto{\aleak{\h}{\val}}
      \val
    }

\inferrule
    [\named{\e-Declassify-Skip}{expression:declassify-skip}]
    {
      \ispublic{\fromlbl}
      \vee
      \issecret{\tolbl}
    }
    {
      \h \says
      \edeclassify{\val}{\fromlbl}{\tolbl}
      \idealstepsto{\aglobalinternal{\h}}
      \val
    }

\inferrule
    [\named{\e-Endorse}{expression:endorse}]
    {
      \isuntrusted{\fromlbl}
      \\
      \istrusted{\tolbl}
    }
    {
      \h \says
      \eendorse{\val}{\fromlbl}{\tolbl}
      \idealstepsto{\amaul{\h}{\val'}}
      \val'
    }

\inferrule
    [\named{\e-Endorse-Skip}{expression:endorse-skip}]
    {
      \istrusted{\fromlbl}
      \vee
      \isuntrusted{\tolbl}
    }
    {
      \h \says
      \eendorse{\val}{\fromlbl}{\tolbl}
      \idealstepsto{\aglobalinternal{\h}}
      \val
    }

\inferrule
    [\named{\e-Input}{expression:input}]
    {
      \nonmalicious{\h}
    }
    {
      \h \says
      \einput{\h}
      \idealstepsto{\aglobalreceive{\envprot}{\h}{\val}}
      \val
    }

\inferrule
    [\named{\e-Input-Malicious}{expression:input-malicious}]
    {
      \malicious{\h}
    }
    {
      \h \says
      \einput{\h}
      \idealstepsto{\aglobalinternal{\h}}
      \vunit
    }

\inferrule
    [\named{\e-Output}{expression:output}]
    {
      \nonmalicious{\h}
    }
    {
      \h \says
      \eoutput{\val}{\h}
      \idealstepsto{\aglobalsend{\h}{\envprot}{\val}}
      \vunit
    }

\inferrule
    [\named{\e-Output-Malicious}{expression:output-malicious}]
    {
      \malicious{\h}
    }
    {
      \h \says
      \eoutput{\val}{\h}
      \idealstepsto{\aglobalinternal{\h}}
      \vunit
    }

\inferrule
    [\named{\e-Receive}{expression:receive}]
    {}
    {
      \h \says
      \ereceive{\h'}
      \idealstepsto{\aglobalinternal{\h}}
      \vunit
    }

\inferrule
    [\named{\e-Send}{expression:send}]
    {}
    {
      \h \says
      \esend{\val}{\h'}
      \idealstepsto{\aglobalinternal{\h}}
      \vunit
    }
\end{mathpar}
 
    \bigskip
    \begin{judgments}
  \judgment{\s \idealstepsto{\action} \s'}
\end{judgments}
\begin{mathpar}
\inferrule
    [\named{\s-Let}{statement:let}]
    {
      \h \says \e \idealstepsto{\action} \val
    }
    {
      \sleth{\x}{\h}{\e}{\s}
      \idealstepsto{\action}
      \subst{\x}{\val}{\s}
    }

\inferrule
    [\named{\s-Communicate}{statement:move}]
    {}
    {
      \smove{\h_1}{\val}{\h_2}{\x}{\s}
      \idealstepsto{\aglobalinternal{\h_1}}
      \subst{\x}{\val}{\s}
    }

\inferrule
    [\named{\s-Select}{statement:select}]
    {}
    {
      \sselect{\h_1}{\val}{\h_2}{\s}
      \idealstepsto{\aglobalinternal{\h_1}}
      \s
    }

\inferrule
    [\named{\s-If}{statement:if}]
    {
      i =
      \text{ if } \val \neq \vfalse
      \text{ then }
      1
      \text{ else }
      2
    }
    {
      \sifh{\val}{\h}{\s_1}{\s_2}
      \idealstepsto{\aglobalinternal{\h}}
      \s_i
    }
\end{mathpar}
 \label{fig:ideal-stepping}
  \end{subfigure}

  \figskip
  \begin{subfigure}{\linewidth}{
    Real stepping rules for expressions and statements.
    These override the rules in \cref{fig:ideal-stepping}.
  }
    \begin{judgments}
  \judgment{\h \says \e \realstepsto{\action} \val}
\end{judgments}
\begin{mathpar}
\inferrule
    [\named{\e-Declassify-Real}{expression:declassify-real}]
    {
      \issecret{\fromlbl}
      \\
      \ispublic{\tolbl}
    }
    {
      \h \says
      \edeclassify{\val}{\fromlbl}{\tolbl}
      \realstepsto{\aglobalinternal{\h}}
      \val
    }

\inferrule
    [\named{\e-Endorse-Real}{expression:endorse-real}]
    {
      \isuntrusted{\fromlbl}
      \\
      \istrusted{\tolbl}
    }
    {
      \h \says
      \eendorse{\val}{\fromlbl}{\tolbl}
      \realstepsto{\aglobalinternal{\h}}
      \val
    }

\inferrule
    [\named{\e-Receive-Real}{expression:receive-real}]
    {}
    {
      \h \says
      \ereceive{\h'}
      \realstepsto{\aglobalreceive{\h'}{\h}{\val}}
      \val
    }

\inferrule
    [\named{\e-Send-Real}{expression:send-real}]
    {}
    {
      \h \says
      \esend{\val}{\h'}
      \realstepsto{\aglobalsend{\h}{\h'}{\val}}
      \vunit
    }
\end{mathpar}
 
    \bigskip
    \begin{judgments}
  \judgment{\s \realstepsto{\action} \s'}
\end{judgments}
\begin{mathpar}
\inferrule
    [\named{\s-Communicate-Real}{statement:move-real}]
    {}
    {
      \smove{\h_1}{\val}{\h_2}{\x}{\s}
      \realstepsto{\aglobalsend{\h_1}{\h_2}{\val}}
      \subst{\x}{\val}{\s}
    }

\inferrule
    [\named{\s-Select-Real}{statement:select-real}]
    {}
    {
      \sselect{\h_1}{\val}{\h_2}{\s}
      \realstepsto{\aglobalsend{\h_1}{\h_2}{\val}}
      \s
    }

\inferrule
    [\named{\s-Case}{statement:case}]
    {}
    {
      \scase{\h_1}{\h_2}{\set{\val \mapsto \s, \dots}}
      \realstepsto{\aglobalreceive{\h_1}{\h_2}{\val}}
      \s
    }
\end{mathpar}
 \label{fig:real-stepping}
  \end{subfigure}

  \figskip
  \begin{subfigure}{\linewidth}{Concurrent lifting of ideal/real stepping rules.}
    \begin{judgments}
  \judgment{\s \polystepsto!{\action} \s'}
\end{judgments}
\begin{mathpar}
\inferrule
    [\named{\s-Sequential}{statement:sequential}]
    {
      \s \polystepsto{\action} \s'
    }
    {
      \s
      \polystepsto!{\action}
      \s'
    }

\inferrule
    [\named{\s-Delay}{statement:delay}]
    {
      \s \polystepsto!{\action} \s'
      \\
      \actor{\action} \notin \hosts{\ectx}
    }
    {
      \ectxof{\s}
      \polystepsto!{\action}
      \ectxof{\s'}
    }

\inferrule
    [\named{\s-If-Delay}{statement:if-delay}]
    {
      \s_1 \polystepsto!{\action} \s_1'
      \\
      \s_2 \polystepsto!{\action} \s_2'
      \\
      \actor{\action} \neq \h
    }
    {
      \sifh{\aexp}{\h}{\s_1}{\s_2}
      \polystepsto!{\action}
      \sifh{\aexp}{\h}{\s_1'}{\s_2'}
    }
\end{mathpar}
 
    \bigskip
    \begin{judgments}
      \judgment{\ectx}
      \and
      \judgment{\hosts{\ectx} = \h*}
    \end{judgments}
    \begin{syntax}
  \category[Evaluation Contexts]{\ectx}
  \alternative{\sleth{\x}{\h}{\e}{\hole}}
  \alternative{\smove{\h_1}{\aexp}{\h_2}{\x}{\hole}}
  \alternative{\sselect{\h_1}{\val}{\h_2}{\hole}}
\end{syntax}
 \begin{mathpar}
  \hosts{\sleth{\x}{\h}{\e}{\hole}}
  =
  \set{\h}

  \hosts{\smove{\h_1}{\aexp}{\h_2}{\x}{\hole}}
  =
  \set{\h_1, \h_2}

  \hosts{\sselect{\h_1}{\val}{\h_2}{\hole}}
  =
  \set{\h_1, \h_2}
\end{mathpar}
 \label{fig:concurrent-stepping}
  \end{subfigure}
  \caption{Ideal, real, and concurrent stepping rules for expressions
    and statements.}
  \label{fig:choreography-semantics}
\end{figure*}

\begin{figure*}
  \begin{minipage}{\linewidth}
    \begin{judgments}
  \judgment{\buffer \stepsto{\action} \buffer'}
\end{judgments}
\begin{mathpar}
\inferrule
    [\named{\buffer-Input}{buffer:input}]
    {}
    {
      \mupdate{\ch_1 \ch_2}{\val*}{\buffer}
      \stepsto{\aglobalreceive{\ch_1}{\ch_2}{\val}}
      \mupdate{\ch_1 \ch_2}{\val* \append \val}{\buffer}
    }

\inferrule
    [\named{\buffer-Output}{buffer:output}]
    {}
    {
      \mupdate{\ch_1 \ch_2}{\val \append \val*}{\buffer}
      \stepsto{\aglobalsend{\ch_1}{\ch_2}{\val}}
      \mupdate{\ch_1 \ch_2}{\val*}{\buffer}
    }
\end{mathpar}
 
    \bigskip
    \begin{judgments}
  \judgment{\process \stepsto{\action} \process'}
\end{judgments}
\begin{mathpar}
\inferrule
    [\named{\process-Input}{process:input}]
    {
      \ch_1 \not\in \h*
      \\
      \ch_2 \in \h*
      \\\\
      \buffer
      \stepsto{\aglobalreceive{\ch_1}{\ch_2}{\val}}
      \buffer'
    }
    {
      \proc{\h*}{\buffer}{\s}
      \stepsto{\aglobalreceive{\ch_1}{\ch_2}{\val}}
      \proc{\h*}{\buffer'}{\s}
    }

\inferrule
    [\named{\process-Discard}{process:discard}]
    {
      \ch_1 \in \h*
      \vee
      \ch_2 \not\in \h*
    }
    {
      \proc{\h*}{\buffer}{\s}
      \stepsto{\aglobalreceive{\ch_1}{\ch_2}{\val}}
      \proc{\h*}{\buffer}{\s}
    }

\inferrule
    [\named{\process-Internal}{process:internal}]
    {
      \buffer \stepsto{\aglobalsend{\ch_1}{\ch_2}{\val}} \buffer'
      \\
      \s \stepsto{\aglobalreceive{\ch_1}{\ch_2}{\val}} \s'
    }
    {
      \proc{\h*}{\buffer}{\s}
      \stepsto{\aglobalinternal{\ch_2}}
      \proc{\h*}{\buffer'}{\s'}
    }

\inferrule
    [\named{\process-Output}{process:output}]
    {
      \s \stepsto{\aoutput{\msg}} \s'
    }
    {
      \proc{\h*}{\buffer}{\s}
      \stepsto{\aoutput{\msg}}
      \proc{\h*}{\buffer}{\s'}
    }
\end{mathpar}
   \end{minipage}
  \caption{Stepping rules for buffers and processes.}
  \label{fig:process-stepping}
\end{figure*}

\Cref{fig:choreography-semantics} gives the full set ideal, real, and concurrent
stepping rules for expressions and statements.
\Cref{fig:process-stepping} gives buffer and process stepping rules.
Concurrent stepping rules refer to \emph{evaluation contexts}---statements
containing a single hole---and the function $\hosts{\cdot}$,
which returns the set of hosts that appear in an evaluation context.
\Cref{statement:delay} allows skipping over \klet, communication, and selection
statements to step a statement in the middle of a program.
\Cref{statement:if-delay} allows stepping the body of an \kif statement without
resolving the conditional as long as both branches step with the same action.
Both rules require the actor of the performed action to be different from the
hosts of the statements being skipped over,
matching the behavior of target programs where code running on a single
host is single-threaded.
   \end{toappendix}

  \nosectionappendix
  \begin{toappendix}
    \section{Properties of the Choreography Language}

\subsection{Typing and Synchronization}

Typing ensures robust declassification and transparent endorsement,
which guarantee that declassified values are always trusted,
and that endorsed values are always public.

\begin{lemma}[Robust Declassification]
  \label{thm:robust-declassification}
  If
  \(\esecure{\edeclassify{\aexp}{\fromlbl}{\tolbl}}{\lbl}\),
  \(\issecret{\fromlbl}\),
  and
  \(\ispublic{\tolbl}\),
  then
  \(\istrusted{\fromlbl}\).
\end{lemma}

\begin{lemma}[Transparent Endorsement]
  \label{thm:transparent-endorsement}
  If
  \(\esecure{\eendorse{\aexp}{\fromlbl}{\tolbl}}{\lbl}\),
  \(\isuntrusted{\fromlbl}\),
  and
  \(\istrusted{\tolbl}\),
  then
  \(\ispublic{\fromlbl}\).
\end{lemma}

Typing has standard properties.

\begin{definition}[Refinement]
  \label{def:refinement}
  Define
  \begin{itemize}
    \item
      \( \ctx_1 \flowsto \ctx_2 \)
      if
      \( (\csing{\x}{\h}{\lbl_2}) \in \ctx_2 \)
      implies
      \( (\csing{\x}{\h}{\lbl_1}) \in \ctx_1 \)
      for some \( \lbl_1 \) such that
      \( \lbl_1 \flowsto \lbl_2 \).

    \item
      \( \sctx_1 \flowsto \sctx_2 \)
      if
      \( \sctx_1(\h_1 h_2) \flowsto \sctx_2(h_1 h_2) \)
      for all \( h_1, \h_2 \in \h* \).
  \end{itemize}
\end{definition}

\begin{lemma}[Subsumption]
  \label{thm:subsumption}
  We have
  \begin{enumerate}
    \item
      If
      \( \esecure[\ctx]{\e}{\lbl} \),
      \( \ctx' \flowsto \ctx \),
      and
      \( \lbl \flowsto \lbl' \),
      then
      \( \esecure[\ctx']{\e}{\lbl'} \).

    \item
      If
      \( \ssecure[\ctx]{\s} \)
      and
      \( \ctx' \flowsto \ctx \),
      then
      \( \ssecure[\ctx']{\s} \).

    \item
      If
      \( \ssynched[\sctx]{\s} \)
      and
      \( \sctx' \flowsto \sctx \),
      then
      \( \ssynched[\ctx']{\s} \).
  \end{enumerate}
\end{lemma}

\begin{lemma}[Substitution]
  \label{thm:substitution}
  Substitution preserves typing:
  \begin{enumerate}
    \item
      If
      \( \esecure[(\ctx \munion \csing{\x}{\h'}{\lbl'})]{\e}{\lbl} \),
      then
      \( \esecure{\subst{\x}{\val}{\e}}{\lbl} \).

    \item
      If
      \( \ssecure[(\ctx \munion \csing{\x}{\h}{\lbl})]{\s} \),
      then
      \( \ssecure[\ctx]{\subst{\x}{\val}{\s}} \).

\end{enumerate}
\end{lemma}

A well-typed program remains well typed under execution and all corruption.

\begin{lemma}[Type Preservation]
  \label{thm:ssecure-preservation}
  If
  \( \ssecure{\s} \)
  and
  \( \s \stepsto{\action} \s' \),
  then
  \( \ssecure{\s'} \).
\end{lemma}

\begin{lemma}[Robust Typing]
  \label{thm:ssecure-robustness}
  If
  \( \ssecure{\s} \),
  then
  \( \ssecure{\corrupt{\s}} \).
\end{lemma}

A well-synchronized program remains well synchronized under execution and
all corruption.

\begin{lemma}[Synchrony Preservation]
  \label{thm:ssynched-preservation}
  If
  \( \ssynched{\s} \)
  and
  \( \s \stepsto{\action} \s' \),
  then
  \( \ssynched{\s'} \).
\end{lemma}

\begin{lemma}[Robust Synchrony]
  \label{thm:ssynched-robustness}
  If
  \( \ssynched{\s} \),
  then
  \( \ssynched{\corrupt{\s}} \).
\end{lemma}

A host can only output if it is synchronized with all previous external actions.

\begin{lemma}[Output Synchronization]
  \label{thm:ssynched-blocked-output}
  If
  \( \ssynched[\sctx]{\s} \)
  and
  \( \s \stepsto{\aoutput{\msg}} \)
  with \( \msg \) external, then
  \( \hissynched[\sctx]{\h} \).
\end{lemma}

\subsection{Operational Semantics}

Below, we write $\stepsto{}$ to stand for any of
$\idealstepsto{}$, $\realstepsto{}$, $\idealstepsto!{}$, or $\realstepsto!{}$.

Processes never refuse input.

\begin{lemma}[Input Totality]
  \label{thm:input-totality}
  For all \( \process \) and \( \msg \),
  there exists \( \process' \) such that
  \( \process \stepsto{\ainput{\msg}} \process' \).
\end{lemma}

The stepping judgments are nondeterministic since inputs are externally
controlled (different input values lead to different states),
and, for concurrent judgments,
outputs and internal actions are independent across hosts.
However, processes are fully deterministic when the action is fixed.

\begin{lemma}[Internal Determinism]
  \label{thm:internal-determinism}
  If
  \( \process \stepsto{\action} \process_1 \)
  and
  \( \process \stepsto{\action} \process_2 \),
  then
  \( \process_1 = \process_2 \).
\end{lemma}

\begin{lemma}[Output Determinism]
  \label{thm:output-determinism}
  If
  \( \process \stepsto{\aoutput{\msg_1}} \process_1 \),
  \( \process \stepsto{\aoutput{\msg_2}} \process_2 \),
  and
  \( \actor{\aoutput{\msg_1}} = \actor{\aoutput{\msg_2}} \),
  then
  \( \msg_1 = \msg_2 \).
\end{lemma}

These results lift to configurations \config as long as the configuration does
not contain duplicate hosts.
   \end{toappendix}

\section{Correctness of Protocol Synthesis}
\label{sec:correctness-of-synthesis}

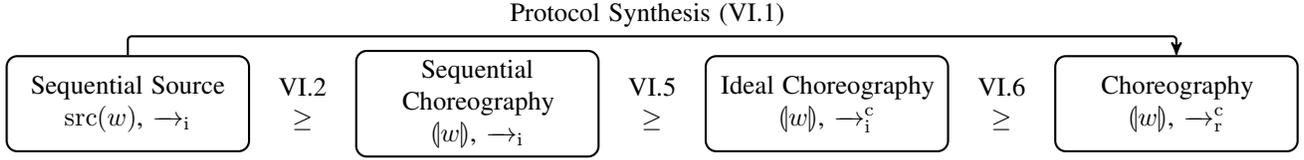
\begin{figure*}
  \centering
  \begin{tikzpicture}[
  auto,
  thick,
  node distance=\nodedistance,
  every node/.style={align=center},
  program/.style={draw, rounded corners, text width=8.5em, minimum height=8ex},
  simulation/.style={},
  transformation/.style={->, >=stealth', rounded corners=0.5mm},
]
  \node[program] (1)
    {Sequential Source \\ $\source{\process}$, $\idealstepsto{}$};
  \node[simulation] (12) [right=of 1]
    {\ref{thm:correctness-of-host-selection} \\ $\simulatedBy$};
  \node[program] (2) [right=of 12]
    {Sequential Choreography \\ $\corrupt{\process}$, $\idealstepsto{}$};
  \node[simulation] (23) [right=of 2]
    {\ref{thm:correctness-of-sequentialization} \\ $\simulatedBy$};
  \node[program] (3) [right=of 23]
    {Ideal Choreography \\ $\corrupt{\process}$, $\idealstepsto!{}$};
  \node[simulation] (34) [right=of 3]
    {\ref{thm:correctness-of-ideal-execution} \\ $\simulatedBy$};
  \node[program] (4) [right=of 34]
    {Choreography \\ $\corrupt{\process}$, $\realstepsto!{}$};

  \pgfmathsetmacro{\arrowshift}{0.7em}

  \draw[transformation]
    (1.north)
    -- ++(0, \nodedistance)
    -| node [pos=0.25, above] {Protocol Synthesis
    (\ref{thm:correctness-of-partitioning})} (4.north);
\end{tikzpicture}
   \caption{
    Intermediate simulation steps for proving the correctness of protocol
    synthesis.
  }
  \label{fig:proof-overview}
\end{figure*}

We prove the correctness of protocol synthesis by demonstrating
a simulation between source programs and their corresponding choreographies.
For $\process = \proc{\h*}{\buffer}{\s}$,
we write $\ssecure{\process}$ and $\ssynched{\process}$ if
$\ssecure{\s}$ and $\ssynched{\s}$, respectively,
and define $\source{\process} = \proc{\h*}{\buffer}{\source{\s}}$.

\begin{theorem}
  \label{thm:correctness-of-partitioning}
  If
  \( \ssecure[\emptylist]{\process} \),
  and
  \( \ssynched[\sctx]{\process} \) for some \( \sctx \),
  then
  \(
    \using{\source{\process}}{\idealstepsto{}}
    \simulatedBy
    \using{\corrupt{\process}}{\realstepsto{}}
  \).
\end{theorem}

We prove \cref{thm:correctness-of-partitioning} through a series of intermediate
simulations, following \cref{fig:proof-overview} from left to right.
First, in \cref{sec:correctness-of-host-selection},
we show idealized, sequential choreographies simulate their canonical source
programs.
Then, we show in \cref{sec:correctness-of-sequentialization} that our synchronization
judgment ensures all externally visible actions happen in program order.
Finally, in \cref{sec:correctness-of-ideal-execution}, we move from
the ideal semantics $\idealstepsto!{}$ to the real semantics $\realstepsto!{}$.

For each simulation, we define a simulator that emulates the adversary
``in its head''.
We ensure that the emulated adversary's view is the same as the real adversary
even though the simulator only has access to public information.
Concretely, we establish a (weak) bisimulation relation~\cite{NainV07, HennessyM85}
between the ideal world (simulator running against ideal configuration)
and the real world (adversary running against real configuration).

\subsection{Correctness of Host Selection}
\label{sec:correctness-of-host-selection}

First,
we show that the original source program is simulated by the sequential
choreography:

\begin{theorem}
  \label{thm:correctness-of-host-selection}
  If
  \( \ssecure[\emptylist]{\process} \),
  then
  \(
    \using{\source{\process}}{\idealstepsto{}}
    \simulatedBy
    \using{\corrupt{\process}}{\idealstepsto{}}
  \).
\end{theorem}

This is shown via two simpler simulations.
First, we add host annotations and explicit communication to the source program:
\begin{lemma}
  \label{thm:choreography-to-source}
  If
  \( \ssecure[\emptylist]{\process} \),
  then
  \(
    \using{\source{\process}}{\idealstepsto{}}
    \simulatedBy
    \using{\process}{\idealstepsto{}}
  \).
\end{lemma}
\begin{proof}
  Statements removed by $\source{\cdot}$ only produce internal actions,
  which the simulator can recreate.
Host annotations affect program behavior only by changing the source
  and destination of internal actions and actions generated by
  \kdeclassify/\kendorse;
  the simulator must recover the original host names before forwarding messages
  from/to the adversary.

  The simulator maintains a public view of \process and runs the
  adversary against this view.
When the adversary steps \process, the simulator steps $\source{\process}$
  only if the statement is preserved by $\source{\cdot}$; it does nothing
  otherwise.
To handle \kdeclassify,
  whenever the simulator receives a message of the form
  $\mglobal{\idealprot}{\advprot}{\val}$,
  the simulator inspects its copy of \process to determine the sending
  host \h, and sends $\mglobal{\h}{\advprot}{\val}$ to the adversary instead.
Similarly, to handle \kendorse,
  the simulator replaces $\h$ with $\idealprot$ in messages
  $\mglobal{\advprot}{\h}{\val}$ from the adversary.

\end{proof}
Second, we add corruption, obtaining a choreography which steps sequentially:
\begin{lemma}
  \label{thm:remove-corruption}
  If
  \( \ssecure[\emptylist]{\process} \),
  then
  \(
    \using{\process}{\idealstepsto{}}
    \simulatedBy
    \using{\corrupt{\process}}{\idealstepsto{}}
  \).
\end{lemma}
\begin{proof}
  Corruption only removes statements at malicious hosts,
  however, these statements only generate internal actions:
  \kinput/\koutput expressions always step internally using ideal rules,
  and typing ensures \kdeclassify/\kendorse expressions step internally.
This means $\corrupt{\process}$ and $\process$ have the same external behavior,
  except $\process$ takes extra internal steps.
The simulator follows the control flow and acts like the adversary,
  but whenever the adversary schedules $\corrupt{\process}$,
  the simulator schedules $\process$ multiple times until the head statement
  is at a nonmalicious host; then, it schedules $\process$ again.

  A small caveat: in $\corrupt{\process}$, all data from
  malicious hosts is explicitly replaced with \vunit,
  whereas malicious hosts may store arbitrary data in $\process$.
Since data from malicious hosts is untrusted, our type system ensures
  this data does not influence trusted data,
  which includes all output messages.
Formally, we only require and maintain that $\corrupt{\process}$
  and $\process$ agree on \emph{trusted} values.
\end{proof}

\subsection{Correctness of Sequentialization}
\label{sec:correctness-of-sequentialization}

Next, we show that if the choreography is well\hyp{}synchronized,
then the choreography stepping sequentially is simulated by itself
stepping concurrently:

\begin{theorem}
  \label{thm:correctness-of-sequentialization}
  If
  \( \ssecure[\emptylist]{\process} \),
  and
  \( \ssynched[\sctx]{\process} \) for some \( \sctx \),
  then
  \(
    \using{\corrupt{\process}}{\idealstepsto{}}
    \simulatedBy
    \using{\corrupt{\process}}{\idealstepsto!{}}
  \).
\end{theorem}

The adversary, interacting with the concurrent version of the choreography,
can schedule a statement that is not next in program order.

If the statement produces an externally visible action,
the simulator must schedule the same statement.
Since the simulator interacts with the sequential version,
it must ``unwind'' the choreography by scheduling every statement leading up
to the desired statement.
Synchronization ensures unwinding does not fail due to a statement blocked on
input (\kinput or \kendorse), or a statement that performs a different visible
action (\koutput or \kdeclassify).

The concurrent and sequential choreographies necessarily fall out of sync
during simulation:
the adversary may schedule steps for the concurrent choreography that the
simulator cannot immediately match,
and the simulator might schedule steps for the sequential choreography while
unwinding, steps the adversary did not schedule.
Nevertheless, the two choreographies remain \emph{joinable}:
they can reach a common choreography via only internal actions.
We prove choreographies are \emph{confluent}~\cite{Newman42, CRosser36, BNipkow98},
which ensures joinable processes remain joinable throughout the simulation.

\begin{proofsketchof}{thm:correctness-of-sequentialization}
  The simulator maintains a public view of the concurrent process,
  and runs the adversary against this view.
When the adversary schedules an output, the simulator schedules
  the sequential process until it performs the same output;
  the simulator does nothing for input and internal actions.
Well-synchronization guarantees the sequential program can perform the output.
The primary invariant, that the concurrent and sequential processes
  remain joinable, is ensured by confluence.
See \apxref{sec:correctness-of-sequentialization-details} for details.
\end{proofsketchof}

\nosectionappendix
\begin{toappendix}
  \appendixfor{sec:correctness-of-sequentialization}
  \label{sec:correctness-of-sequentialization-details}
  \begin{definition}[Joinable Processes]
  \label{def:joinable}
  We write
  \( \joinable{\process_1}{\process_2} \)
  if there exist traces
  \( \tr_1 \) and \( \tr_2 \)
  containing only internal actions such that
  \( \process_1 \idealstepsto!{\tr_1} \process \)
  and
  \( \process_2 \idealstepsto!{\tr_2} \process \)
  for some \( \process \).
Diagrammatically:
  
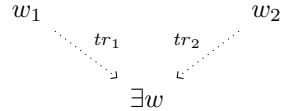
\[
    \begin{tikzcd}
  \process_1
    \sarrow![dr, dotted, "{\tr_1}"]
  &
  &
  \process_2
    \sarrow![dl, dotted, "{\tr_2}"']
  \\
  & \exists \process
\end{tikzcd}
   \]
\end{definition}

We prove confluence through a \emph{diamond} lemma,
which allows reordering \emph{independent} actions.

\begin{definition}[Independent Actions]
  \label{def:independent-actions}
  Actions \( \action_1 \) and \( \action_2 \) are independent,
  written \( \aindep{\action_1}{\action_2} \),
  if one is an input while the other is an output,
  or they are on different channels.
We write
  \( \aindep{\tr_1}{\tr_2} \)
  if
  \( \aindep{\action_1}{\action_2} \)
  for all \( \action_1 \in \tr_1 \) and \( \action_2 \in \tr_2 \).
\end{definition}

\begin{lemma}[Diamond for Processes]
  \label{thm:diamond-process}
  If
  \( \process \idealstepsto!{\tr_1} \process_1 \),
  \( \process \idealstepsto!{\tr_2} \process_2 \),
  and
  \( \aindep{\tr_1}{\tr_2} \),
  then
  \( \process_1 \idealstepsto!{\tr_2} \process' \)
  and
  \( \process_2 \idealstepsto!{\tr_1} \process' \)
  for some \( \process' \).
Diagrammatically:
  
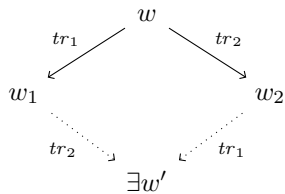
\[
    \begin{tikzcd}
  &
  \process
    \sarrow![dl, "{\tr_1}"']
    \sarrow![dr, "{\tr_2}"]
  &
  \\
  \process_1
    \sarrow![dr, dotted, "{\tr_2}"']
  & &
  \process_2
    \sarrow![dl, dotted, "{\tr_1}"]
  \\
  &
  \exists \process'
  &
\end{tikzcd}
   \]
\end{lemma}

\Cref{thm:diamond-process} does the heavy lifting when proving multiple
confluence results below, and requires quite a bit of work to show.
We first prove a diamond lemma for statements,
and then lift it to processes.

\begin{lemma}[Half Diamond for Statements]
  \label{thm:half-diamond-statement}
  If
  \( \s \idealstepsto{\action_1} \s_1 \),
  \( \s \idealstepsto!{\action_2} \s_2 \),
  and
  \( \aindep{\action_1}{\action_2} \),
  then
  \( \s_1 \idealstepsto!{\action_2} \s' \)
  and
  \( \s_2 \idealstepsto{\action_1} \s' \)
  for some \( \s' \).
\end{lemma}
\begin{proof}
  By case analysis on
  \( \s \idealstepsto!{\action_2} \s_2 \).
\begin{pcases}
    \prule{statement:sequential}
      Contradicts \( \aindep{\action_1}{\action_2} \).

    \prule{statement:delay}
      By case analysis on the evaluation context followed by inversion
      on \( \s \idealstepsto{\action_1} \s_1 \).
The step for $\action_1$ involves only the head statement and
      ignores all future statements, whereas the step for $\action_2$
      ignores the head statement and involves only a statement in the future.
Thus, they can be performed in sequence in either order without changing
      the end result.

    \prule{statement:if-delay}
      The step for $\action_2$ steps both branches of the \kif statement,
      whereas the step for $\action_1$ selects a branch.
They can be performed in sequence in either order.
    \qedhere
  \end{pcases}
\end{proof}

\begin{lemma}[Diamond for Statements]
  \label{thm:diamond-statement}
  If
  \( \s \idealstepsto!{\action_1} \s_1 \),
  \( \s \idealstepsto!{\action_2} \s_2 \),
  and
  \( \aindep{\action_1}{\action_2} \),
  then
  \( \s_1 \idealstepsto!{\action_2} \s' \)
  and
  \( \s_2 \idealstepsto!{\action_1} \s' \)
  for some \( \s' \).
\end{lemma}
\begin{proof}
  By induction on the derivations of both stepping judgments.
If either is by \cref{statement:sequential},
  we conclude by \cref{thm:half-diamond-statement}.
Otherwise, both steps ignore the head of \s
  using the same delay rule.
We appeal to the induction hypothesis,
  and use the same delay rule to get a complete derivation.
\end{proof}

\begin{proofof}{thm:diamond-process}
  We prove the statement when $\tr_1$ and $\tr_2$ are single actions;
  the more general statement follows straightforwardly by induction on
  $\tr_1$ followed by induction on $\tr_2$.

  We proceed by case analysis on both stepping judgments.
  \begin{itemize}
    \item Both steps are input (\cref{process:input,process:discard}).
      Since $\aindep{\tr_1}{\tr_2}$,
      the input messages are added at the end of two different queues.
Both actions can be performed in either order without affecting
      the end result.

    \item One step is input, the other is by \cref{process:internal}.
      The input step adds a message to a queue,
      while \cref{process:internal} pops a message from a queue and feeds it
      to the choreography.
The queues must be different since $\aindep{\tr_1}{\tr_2}$,
      so the steps are independent.

    \item One step is input, the other is by \cref{process:output}.
      The input step only affects the queue,
      and the output step only affects the choreography,
      so the steps are independent.

    \item Both steps are output (\cref{process:internal,process:output}).
      If either step is by \cref{process:internal},
      then we use $\aindep{\tr_1}{\tr_2}$ as before to show we pull messages
      out of different queues.
This allows reordering changes to the buffer.
\Cref{thm:diamond-statement} finishes the proof.
    \qedhere
  \end{itemize}
\end{proofof}

The proof of \cref{thm:diamond-process} reasons generically about buffers,
and appeals to a diamond lemma for statements in a black-box manner.
This means we can generalize \cref{thm:diamond-process} to arbitrary
(combinations of) stepping relations without extra work
as long as a diamond property for the same relations holds for statements.

\begin{lemma}[Generalized Diamond for Processes]
  \label{thm:general-diamond-process}
  Assume the diamond property holds for statements with stepping relations
  \( \stepsto{}_1 \)
  and
  \( \stepsto{}_2 \).
If
  \( \process \stepsto{\tr_1}_1 \process_1 \),
  \( \process \stepsto{\tr_2}_2 \process_2 \),
  and
  \( \aindep{\tr_1}{\tr_2} \),
  then
  \( \process_1 \stepsto{\tr_2}_2 \process' \)
  and
  \( \process_2 \stepsto{\tr_1}_1 \process' \)
  for some \( \process' \).
\end{lemma}
\begin{proof}
  Same as the proof of \cref{thm:diamond-process}.
\end{proof}

Processes remain joinable after taking internal or matching steps.

\begin{lemma}[Internal Action]
  \label{thm:joinable-after-internal}
  If
  \( \joinable{\process_1}{\process_2} \)
  and
  \( \process_1 \idealstepsto!{\action} \process_1' \)
  for \( \action \) internal,
  then
  \( \joinable{\process_1'}{\process_2} \).
Diagrammatically:
  
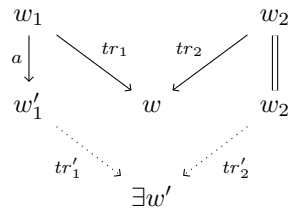
\[
    \begin{tikzcd}
  \process_1
    \sarrow![d, "{\action}"']
    \sarrow![dr, "{\tr_1}"]
  &
  &
  \process_2
    \sarrow![dl, "{\tr_2}"']
    \arrow[d, equal]
  \\
  \process_1'
    \sarrow![dr, dotted, "{\tr_1'}"']
  &
  \process
  &
  \process_2
    \sarrow![dl, dotted, "{\tr_2'}"]
  \\
  &
  \exists \process'
  &
\end{tikzcd}
   \]
\end{lemma}
\begin{proof}
  Since $\process_1$ and $\process_2$ are joinable,
  there exist $\process$ and internal $\tr_1, \tr_2$ such that
  \( \process_1 \idealstepsto!{\tr_1} \process \)
  and
  \( \process_2 \idealstepsto!{\tr_2} \process \).
We case on whether
  \( \aindep{\action}{\tr_1} \).
  \begin{pcases}
    \pcase{\( \aindep{\action}{\tr_1} \)}
      \Cref{thm:diamond-process} gives $\process'$ such that
      \( \process_1' \idealstepsto!{\tr_1} \process' \)
      and
      \( \process \idealstepsto!{\action} \process' \).
Since $\action$ and $\tr_2$ are internal, so is $\tr_2 \append \action$.
      Thus,
      \( \joinable{\process_1'}{\process_2} \)
      through $\tr_1$ and $\tr_2 \append \action$.
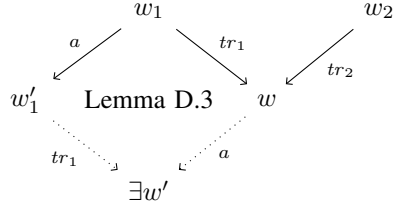
\[
        \begin{tikzcd}
  &
  \process_1
    \sarrow![dl, "{\action}"']
    \sarrow![dr, "{\tr_1}"]
  &
  &
  \process_2
    \sarrow![dl, "{\tr_2}"]
  \\
  \process_1'
    \sarrow![dr, dotted, "{\tr_1}"']
    \arrow[rr, phantom, "\text{\Cref{thm:diamond-process}}" description]
  &
  &
  \process
    \sarrow![dl, dotted, "{\action}"]
  &
  \\
  &
  \exists \process'
  &
  &
\end{tikzcd}
       \]

    \pcase{\( \anindep{\action}{\tr_1} \)}
      Let $\action_1$ be the first action in $\tr_1$ such that
      \( \anindep{\action}{\action_1} \), that is,
      \( \tr_1 = \tr_1' \append \action_1 \append \tr_1'' \)
      with
      \( \aindep{\action}{\tr_1'} \).
We have,
      \(
        \process_1
        \idealstepsto!{\tr_1'}
        \process'
        \idealstepsto!{\action_1}
        \process''
        \idealstepsto!{\tr_1''}
        \process
      \).
\Cref{thm:diamond-process} gives $\process_1''$ such that
      \( \process_1' \idealstepsto!{\tr_1'} \process_1'' \)
      and
      \( \process' \idealstepsto!{\action} \process_1'' \).
We now have
      \( \process' \idealstepsto!{\action} \process_1'' \)
      and
      \( \process' \idealstepsto!{\action} \process'' \),
      however,
      the stepping judgment is deterministic on dependent internal actions,
      thus \cref{thm:output-determinism} gives
      \( \process_1'' = \process'' \).\footnote{
        More specifically, $\action$ and $\action_1$ are internal actions,
        which are represented as outputs.
Since
        \( \anindep{\action}{\action_1} \),
        we have
        \( \actor{\action} = \actor{\action_1} \),
        so \cref{thm:output-determinism} applies.
      }
Finally,
      \( \joinable{\process_1'}{\process_2} \)
      through $\tr_1' \append \tr_1''$ and $\tr_2$.
\[
        \begin{tikzcd}
  &
  \process_1
    \sarrow![dl, "{\action}"']
    \sarrow![dr, "{\tr_1'}"]
  &
  &
  &
  &
  &
  &
  \process_2
    \sarrow![dddlll, "{\tr_2}"']
  \\
  \process_1'
    \sarrow![dr, dotted, "{\tr_1'}"]
    \arrow[rr, phantom, "\text{\Cref{thm:diamond-process}}" description]
  &
  &
  \process'
    \sarrow![dl, dotted, "{\action}"']
    \sarrow![dr, "{\action_1}"]
  \\
  &
  \exists \process_1''
    \arrow[rr, equal, "\text{\Cref{thm:output-determinism}}"]
  &
  &
  \process''
    \sarrow![dr, "{\tr_1''}"]
  & &
  \\
  & & & &
  \process
\end{tikzcd}
       \]
  \end{pcases}
\end{proof}

\begin{lemma}[Matching Actions]
  \label{thm:joinable-after-matching}
  If
  \( \joinable{\process_1}{\process_2} \),
  \( \process_1 \idealstepsto!{\action} \process_1' \),
  and
  \( \process_2 \idealstepsto!{\action} \process_2' \),
  then
  \( \joinable{\process_1'}{\process_2'} \).
Diagrammatically:
  
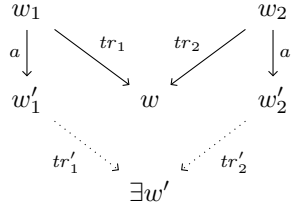
\[
    \begin{tikzcd}
  \process_1
    \sarrow![d, "{\action}"']
    \sarrow![dr, "{\tr_1}"]
  &
  &
  \process_2
    \sarrow![dl, "{\tr_2}"']
    \sarrow![d, "{\action}"]
  \\
  \process_1'
    \sarrow![dr, dotted, "{\tr_1'}"']
  &
  \process
  &
  \process_2'
    \sarrow![dl, dotted, "{\tr_2'}"]
  \\
  &
  \exists \process'
  &
\end{tikzcd}
   \]
\end{lemma}
\begin{proof}
  Since $\process_1$ and $\process_2$ are joinable,
  there exist $\process$ and internal $\tr_1, \tr_2$ such that
  \( \process_1 \idealstepsto!{\tr_1} \process \)
  and
  \( \process_2 \idealstepsto!{\tr_2} \process \).
If $\action$ is internal, then the result follows by two applications of
  \cref{thm:joinable-after-internal}.
Otherwise,
  \( \aindep{\action}{\tr_1} \)
  and
  \( \aindep{\action}{\tr_2} \).
The result follows from two applications of \cref{thm:diamond-process},
  and one application of \cref{thm:internal-determinism}:
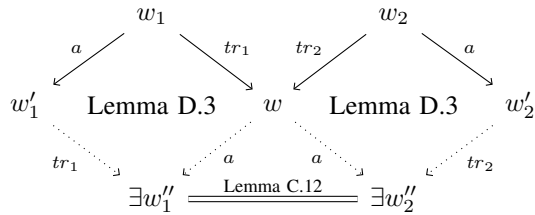
\[
    \begin{tikzcd}
  &
  \process_1
    \sarrow![dl, "{\action}"']
    \sarrow![dr, "{\tr_1}"]
  &
  &
  \process_2
    \sarrow![dl, "{\tr_2}"']
    \sarrow![dr, "{\action}"]
  \\
  \process_1'
    \sarrow![dr, dotted, "{\tr_1}"']
    \arrow[rr, phantom, "\text{\Cref{thm:diamond-process}}" description]
  &
  &
  \process
    \sarrow![dl, dotted, "{\action}"]
    \sarrow![dr, dotted, "{\action}"']
  &
  &
  \process_2'
    \sarrow![dl, dotted, "{\tr_2}"]
    \arrow[ll, phantom, "\text{\Cref{thm:diamond-process}}" description]
  \\
  &
  \exists \process_1''
    \arrow[rr, equal, "\text{\Cref{thm:internal-determinism}}"]
  &
  &
  \exists \process_2''
\end{tikzcd}
   \]
\end{proof}

If a well-typed, well-synchronized program can take an output step concurrently,
then it can take the same step sequentially
(after taking the series of internal steps leading up to the output).
We write
$\process \idealstepsto*{\aoutput{\msg}} \process'$
if
$\process \idealstepsto{\tr \append \aoutput{\msg}} \process'$
for some internal \tr.

\begin{lemma}[Sequential Execution]
  \label{thm:sequential-output}
  If
  \( \ssecure[\emptylist]{\process} \),
  \( \ssynched{\process} \),
  and
  \( \process \idealstepsto!{\aoutput{\msg}} \)
  for \( \msg \) external,
  then
  \( \process \idealstepsto*{\aoutput{\msg}} \).
\end{lemma}
\begin{proofsketch}
  By induction on the stepping relation.
If the step is by \cref{statement:sequential},
  then $\process \idealstepsto{\aoutput{\msg}}$ and we are done.
Otherwise, it must be by \cref{statement:delay}.
We need to show that we can take a sequential internal step by casing on the
  evaluation context $\ectx$.
Note that the top statement in \ectx must be internal,
  otherwise we get a contradiction by \cref{thm:ssynched-blocked-output}.
Since $\process$ has no free variables, it can take an internal step.
  
\end{proofsketch}

\begin{lemma}[Step Over]
  \label{thm:concurrent-then-sequential}
  If
  \( \process_1 \idealstepsto!{\action_1} \process_2 \idealstepsto{\action_2} \),
  then either
  \( \process_1 \idealstepsto{\action_1} \process_2 \),
  or
  \( \process_1 \idealstepsto{\action_2} \)
  and
  \( \aindep{\action_1}{\action_2} \).
\end{lemma}
\begin{proof}
  Assume the first step cannot be done sequentially
  (otherwise we are done).
Then, the first step uses a delay rule (\cref{statement:delay,statement:if-delay}),
  which ignores the head statement of $\process_1$.
The second step is sequential so it only depends on the head statement
  of $\process_2$, which is the same as the head statement of $\process_1$,
  thus, $\process_1$ can step with $\action_2$.
We get that
  \( \aindep{\action_1}{\action_2} \)
  from the side conditions on delay rules,
  which require \( \actor{\action_1} \) to be different from
  \( \actor{\action_2} \).
\end{proof}

\begin{lemma}[Delayed Step]
  \label{thm:delayed-steps}
  If
  \( \process_1 \idealstepsto!{\action} \process_2 \idealstepsto*{\aoutput{\msg}} \)
  for \( \action \) internal,
  then
  \( \process_1 \idealstepsto*{\aoutput{\msg}} \).
\end{lemma}
\begin{proof}
  Unfolding the definition of $\idealstepsto*{}$ gives
  \(
    \process_2 \idealstepsto{\tr} \process_3 \idealstepsto{\aoutput{\msg}}
  \)
  for some $\tr$ and $\process_3$.
We proceed by induction on \tr.
In either case, we are done if
  \( \process_1 \idealstepsto{\action} \process_2 \),
  so assume $\process_1$ cannot perform $\action$ sequentially.
\begin{pcases}
    \pcase{$\tr = \emptylist$}
      We have
      \( \process_1 \idealstepsto!{\action} \process_2 \idealstepsto{\aoutput{\msg}} \).
\Cref{thm:concurrent-then-sequential} gives
      \( \process_1 \idealstepsto{\aoutput{\msg}} \),
      and thus
      \( \process_1 \idealstepsto*{\aoutput{\msg}} \).

    \pcase{$\tr = \action' \append \tr'$}
      We have
      \(
        \process_1
        \idealstepsto!{\action}
        \process_2
        \idealstepsto{\action'}
        \process_3
        \idealstepsto{\tr'}
        \process_4
        \idealstepsto{\aoutput{\msg}}
      \).
\Cref{thm:concurrent-then-sequential} gives
      \( \process_1 \idealstepsto{\action'} \process_2' \)
      for some $\process_2'$.
\Cref{thm:general-diamond-process} with \cref{thm:half-diamond-statement}
      then gives
      \(
        \process_1
        \idealstepsto{\action'}
        \process_2'
        \idealstepsto!{\action}
        \process_3'
      \)
      for some $\process_3'$,
      and \cref{thm:internal-determinism} shows $\process_3' = \process_3$.
We use the induction hypothesis on
      \(
        \process_2'
        \idealstepsto!{\action}
        \process_3
        \idealstepsto{\tr'}
        \process_4
        \idealstepsto{\aoutput{\msg}}
      \)
      to get
      \(
        \process_2'
        \idealstepsto*{\aoutput{\msg}}
      \),
      and combined with
      \(
        \process_1
        \idealstepsto{\action'}
        \process_2'
      \),
      we get
      \(
        \process_1
        \idealstepsto*{\aoutput{\msg}}
      \).

      Diagrammatically (we use tails to denote sequential steps):
      
      
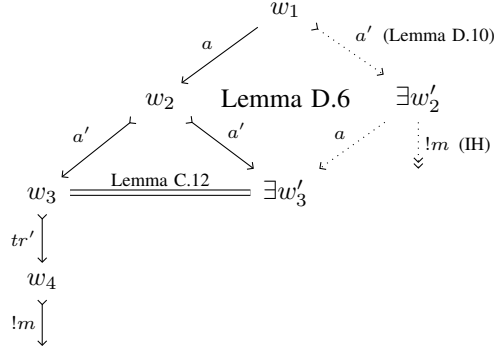
\[
        \begin{tikzcd}
  &
  &
  \process_1
    \sarrow![dl, "{\action}"']
    \sarrow[dr, dotted, "{\action'}~\text{(\Cref{thm:concurrent-then-sequential})}"]
  \\
  &
  \process_2
    \sarrow[dl, "{\action'}"']
    \sarrow[dr, "{\action'}"]
    \arrow[rr, phantom, "\text{\Cref{thm:general-diamond-process}}" description]
  &
  &
  \exists \process_2'
    \sarrow![dl, dotted, "{\action}"']
    \sarrow*[d, dotted, "{\aoutput{\msg}}~\text{(IH)}"]
  \\
  \process_3
    \sarrow[d, "{\tr'}"']
    \arrow[rr, equal, "\text{\Cref{thm:internal-determinism}}"]
  &
  &
  \exists \process_3'
  &
  \phantom{\process_3''}
  \\
  \process_4
    \sarrow[d, "{\aoutput{\msg}}"']
  \\
  \phantom{\process_5}
\end{tikzcd}
       \]
    \qedhere
  \end{pcases}
\end{proof}

If two processes are joinable and one of them can concurrently perform an
output action, then the other can perform the same action sequentially
(after unwinding).

\begin{lemma}[Matching Outputs]
  \label{thm:matching-outputs}
  Let \( \process_1 \) and \( \process_2 \) be such that
  \( \joinable{\process_1}{\process_2} \),
  \( \ssecure[\emptylist]{\process_2} \),
  and
  \( \ssynched{\process_2} \).
If
  \( \process_1 \idealstepsto!{\aoutput{\msg}} \)
  for \( \msg \) external,
  then
  \( \process_2 \idealstepsto*{\aoutput{\msg}} \).
\end{lemma}
\begin{proof}
  Since $\process_1$ and $\process_2$ are joinable,
  there exist $\process$ and internal $\tr_1, \tr_2$ such that
  $\process_1 \idealstepsto!{\tr_1} \process$
  and
  $\process_2 \idealstepsto!{\tr_2} \process$.
Because $\msg$ is external and $\tr_1$ is internal,
  we have $\aindep{\aoutput{\msg}}{\tr_1}$,
  so \cref{thm:diamond-process} gives
  $\process \idealstepsto!{\aoutput{\msg}}$.
Now we have
  \( \process_2 \idealstepsto!{\tr_2 \append \aoutput{\msg}} \),
  and we want to show
  \( \process_2 \idealstepsto*{\aoutput{\msg}} \).
We proceed by induction on $\tr_2$.
When $\tr_2$ is empty, \cref{thm:sequential-output} completes the proof.
When $\tr_2 = \action \append \tr_2'$,
  the induction hypothesis gives
  \( \process_2 \idealstepsto!{\action} \process_2' \idealstepsto*{\aoutput{\msg}} \),
  and \cref{thm:delayed-steps} gives the desired result.
\end{proof}

\begin{proofof}{thm:correctness-of-sequentialization}
  The simulator maintains a public view of the concurrent process,
  and runs the adversary against this view.
When the adversary schedules an output action, the simulator schedules
  the sequential process until it performs the same output;
  the simulator does nothing for input and internal actions.
\Cref{thm:matching-outputs} guarantees the sequential program can perform
  the output.
The primary invariant, that the concurrent and sequential processes
  remain joinable, is ensured by
  \cref{thm:joinable-after-internal,thm:joinable-after-matching}.

  More formally, we show UC simulation as follows.
  \begin{ucproof}
    \newcommand{\simof}[3][\adv]{\operatorname{\simulator}(#1 \parallel #2, #3)}\isimulator
      The simulator has the form
      \( \simof{\process_1'}{\process_2'} \)
      where $\process_1'$ is the public view of the concurrent process,
      and $\process_2'$ is the public view of the sequential process.
When the simulator receives an input from the environment
      or a \kdeclassify message from the sequential process,
      it feeds the message to \adv, $\process_1'$, and $\process_2'$.
When the adversary outputs a value
      (for the environment or for an \kendorse expression),
      the simulator feeds it to $\process_1'$ and $\process_2'$,
      and outputs the same value.
When the simulator receives an internal message from the sequential
      process (which indicates the sequential process has taken a step),
      it steps $\process_2'$.
The simulator only allows an output step for the sequential process
      if $\process_1'$ can perform same output.

    \irelation
      Let
      \(
        \adv \parallel \process_1
        \rel
        \simof[\adv']{\process_1'}{\process_2'} \parallel \process_2
      \)
      if:
      \begin{relations}
        \item
          \label{proof:seq-adv}
          \( \adv = \adv' \)

        \item
          \label{proof:seq-concurrent}
          \( \process_1 \publiceq \process_1' \)

        \item
          \label{proof:seq-sequential}
          \( \process_2 \publiceq \process_2' \)

        \item
          \label{proof:seq-joinable}
          \( \joinable{\process_1}{\process_2} \)

        \item
          \label{proof:seq-typed}
          \( \ssecure[\emptylist]{\process_2} \)
          and
          \( \ssynched{\process_2} \)
          for some \sctx.
      \end{relations}

    \iproof
      We claim $\rel$ is a weak bisimulation.

      \Cref{proof:seq-joinable,proof:seq-typed} ensure
      \cref{thm:matching-outputs} is applicable,
      which in turn ensures the external behavior of both systems
      is the same.

      \Cref{proof:seq-adv} is preserved since $\process_1'$ is an
      accurate public view of $\process_1$ (\cref{proof:seq-concurrent}).
\Cref{proof:seq-concurrent,proof:seq-sequential} are preserved
      since messages from \kdeclassify are sufficient to maintain a
      public view.
\Cref{thm:joinable-after-internal,thm:joinable-after-matching}
      ensure \cref{proof:seq-joinable} is preserved.
\Cref{thm:ssecure-preservation,thm:ssynched-preservation}
      ensure \cref{proof:seq-typed} is preserved.
    \qedhere
  \end{ucproof}
\end{proofof}
 \end{toappendix}

\subsection{Correctness of Ideal Execution}
\label{sec:correctness-of-ideal-execution}

Finally, we move from a concurrent choreography stepping with ideal rules
to a concurrent choreography stepping with real rules:
\begin{theorem}
  \label{thm:correctness-of-ideal-execution}
  If
  \( \ssecure[\emptylist]{\process} \),
  then
  \(
    \using{\corrupt{\process}}{\idealstepsto!{}}
    \simulatedBy
    \using{\corrupt{\process}}{\realstepsto!{}}
  \).
\end{theorem}

The main difference between the two semantics is the interface with the
adversary.
In the real semantics,
dishonest hosts actively leak data to the adversary
(through \ksend expressions and communication statements),
and the adversary controls all data coming from malicious hosts
(through \kreceive).
In contrast, the ideal semantics interacts with the adversary
only via \kdeclassify and \kendorse.
In effect, the ideal semantics causes leakage and corruption to become
coarse-grained.
Additionally, by eliminating all blocking \kreceive expressions
(which communicate with the adversary), the ideal semantics can make
progress in a manner independent of the adversary;
this aids the sequentialization proof in \cref{sec:correctness-of-sequentialization}.

To bridge the gap between the real and ideal semantics,
we show that the simulator can use \kdeclassify expressions
to recreate all data no longer leaked through communication,
and \kendorse expressions to corrupt all data no longer corruptible through
\kreceive expressions.
This is possible because our type system ensures secrets are not
directly sent to dishonest hosts (\cref{esecure:send,ssecure:move}),
and data from malicious hosts cannot directly influence trusted data
(\cref{esecure:receive}).

\begin{proofsketchof}{thm:correctness-of-ideal-execution}
  \newcommand{\simof}[1]{\simulator(\adv \parallel #1)}
The simulator maintains a public view of the real process,
  and runs the adversary against this view.
The simulator flips the input/output behavior of \kdeclassify and \kendorse
  expressions:
  when the ideal process \emph{outputs} data through a \kdeclassify expression,
  the simulator inputs this data instead.
Similarly, the simulator outputs data with \kendorse expressions,
  which it sends to the ideal process.
The key invariant is that the simulator's version of the process
  matches the real one on \emph{public values},
  and the ideal process matches the real one on \emph{trusted values}.

  This invariant is strong enough to witness simulatability.
Since the simulator's version of the process matches the real one on public values,
  the adversary in the real configuration has a view identical to the adversary
  running inside of the simulator
  (the adversary only sees public data).
Similarly, since the real process matches the ideal one on trusted values,
  the environment has the same view in both
  (the environment is only sent trusted data).

  Next, we argue that our simulator can preserve the invariant.
For the simulator to have an accurate view of public values in the real process,
  the ideal process must output values through \kdeclassify expressions that
  match the values in the real process.
The information-flow type system provides necessary restrictions.
\emph{Robust declassification}~\cite{nmifc} guarantees only trusted values
  are declassified.
Since the ideal process matches the real one on trusted values,
  data coming from \kdeclassify expressions is accurate.
Dually, when the simulator sends data to the ideal process
  through \kendorse expressions, the values must match the ones in the real process.
\emph{Transparent endorsement}~\cite{nmifc} guarantees only public values
  are endorsed.
Thus, the result follows from the simulator holding accurate public values.
See \apxref{sec:correctness-of-ideal-execution-details} for details.
\end{proofsketchof}

\nosectionappendix
\begin{toappendix}
  \appendixfor{sec:correctness-of-ideal-execution}
  \label{sec:correctness-of-ideal-execution-details}
  Consider choreography \process and its corruption $\corrupt{\process}$
when \alice is malicious:
\begin{mathpar}
  \begin{program}
    \scomment{\process}
    \\
    \slethhead{\x}{\alice}{\einput{\alice}}
    \\
    \smovehead{\alice}{\x}{\bob}{\y}
    \\
    \smovehead{\alice}{\x}{\chuck}{\z}
    \\
    \slethhead{\y'}{\bob}{\y + 1}
    \\
    \smovehead{\bob}{\y'}{\alice}{\x'}
  \end{program}

  \begin{program}
    \scomment{$\corrupt{\process}$}
    \\
    \\
    \slethhead{\y}{\bob}{\ereceive{\alice}}
    \\
    \slethhead{\z}{\chuck}{\ereceive{\alice}}
    \\
    \slethhead{\y'}{\bob}{\y + 1}
    \\
    \slethhead{\_}{\bob}{\esend{\y'}{\alice}}
  \end{program}
\end{mathpar}
The function $\corrupt{\cdot}$
erases all code on \alice (the first \klet statement) since a malicious
host does not follow the choreography and has arbitrary behavior.
Additionally, it replaces all communication statements involving \alice
with \kreceive/\ksend statements,
capturing the fact that \alice need not use the variables specified in the
choreography (\x and $\x'$).
In particular, even though the original choreography specifies \alice sends
\emph{the same} value to \bob and \chuck, a malicious \alice can send
\emph{different} values.
Giving \alice the power to \emph{equivocate} in this manner can compromise
security,
for instance, \alice could cause \bob and \chuck to disagree on control flow
if \x is used as a conditional guard.
Information-flow checking prevents \alice from exploiting this power.

Information-flow checking ensures untrusted data (from malicious hosts)
cannot influence trusted data (of nonmalicious hosts).
We formalize this intuition by erasing all data from malicious hosts in
the ideal semantics:
instead of receiving the value of \y from \alice (i.e., the adversary),
\bob simply assigns \vunit to \y (\chuck does the same for \z).
The adversary cannot possibly have any control over trusted data if all data
coming from the adversary is replaced with \vunit.
Note that erasing untrusted data can change the adversary's view.
In the example, \bob sends $\y' = \y + 1$ to \alice,
which is different from sending $\vunit + 1$.
The simulator can compute the correct value in this case since \y comes from the
adversary (which the simulator has access to), and $1$ is a fixed constant.
In the general case, the simulator can compute all public values, and
our type system ensures only public values are sent to dishonest hosts
(\cref{esecure:send,ssecure:move}).

In addition to preventing the adversary from corrupting trusted values,
we must prevent the adversary from learning secrets.
In the real semantics, the adversary witnesses all communication
and can read any message if at least one endpoint is dishonest.
Information-flow checking ensures the adversary does not learn anything new by
reading these messages.
We formalize this intuition by erasing communication:
in the ideal semantics, communication statements step internally.
The simulator must again recreate these hidden messages for the adversary,
which is possible since our type system ensures the messages the adversary
can read are public.

Simply discarding all untrusted data and hiding all secret data weakens
the adversary in the ideal semantics too much.
We bridge the gap between the real and ideal semantics through downgrade
expressions.
An \kendorse expression indicates that some untrusted data should be treated as
trusted, so in the ideal semantics, an \kendorse inputs data from the adversary.
Dually, a \kdeclassify expression indicates some secret data should be treated
as public, so a \kdeclassify outputs data to the adversary.
Explicit \kdeclassify/\kendorse expressions capture programmer intent.
Going back to the example,
our type system requires \bob and \chuck to \kendorse \x before using it
in a trusted context.
If \bob and \chuck separately \kendorse \x, then they might get two different
values.
If there is only one \kendorse
(e.g., a separate trusted host performs the \kendorse and shares the result),
then there can only be one value.

The core of the simulation result is showing that the simulator can
use \kdeclassify expressions to recreate all data no longer leaked
through communication,
and \kendorse expressions to influence all data no longer corruptible through
\kreceive expressions.
For this to work, we need to ensure the ideal choreography outputs the correct
value to the simulator when performing a \kdeclassify,
and we need to ensure the simulator can input the correct value to the ideal
choreography for an \kendorse.
For example, when the ideal choreography performs $\ddeclassify{\x}$,
we must ensure the value of \x is the same in the real and ideal choreographies.
This is nontrivial since the ideal semantics replaces all untrusted data with
\vunit.
\emph{Robust declassification} requires only trusted data is declassified,
and type checking ensures untrusted data does not influence trusted data.
Thus, \x is trusted and erased values cannot influence its value.
Similarly, when the ideal choreography performs $\dendorse{\x}$,
the simulator must compute the value \x would have in the real choreography
and send that to the ideal choreography.
\emph{Transparent endorsement} requires only public data is endorsed,
and the simulator can recreate all public data.

\begin{figure}
  \begin{minipage}{\linewidth}
    \begin{judgments}
  \judgment{\h \says \e \simstepsto{\action} \val}
\end{judgments}
\begin{mathpar}
\inferrule
    [\named{\e-Declassify-Simulator}{expression:declassify-sim}]
    {
      \issecret{\fromlbl}
      \\
      \ispublic{\tolbl}
    }
    {
      \h \says
      \edeclassify{\val}{\fromlbl}{\tolbl}
      \simstepsto{\aglobalreceive{\h}{\advprot}{\val'}}
      \val'
    }

\inferrule
    [\named{\e-Endorse-Simulator}{expression:endorse-sim}]
    {
      \isuntrusted{\fromlbl}
      \\
      \istrusted{\tolbl}
    }
    {
      \h \says
      \eendorse{\val}{\fromlbl}{\tolbl}
      \simstepsto{\aglobalsend{\advprot}{\h}{\val'}}
      \val
    }
\end{mathpar}
   \end{minipage}
  \caption{Stepping rules used internally by the simulator.
    These override \cref{fig:real-stepping}.}
  \label{fig:simulator-stepping}
\end{figure}

The simulator maintains a public view of the real process,
and runs the adversary against this view.
It uses the rules in \cref{fig:simulator-stepping},
which flip the roles of \kdeclassify and \kendorse.
We maintain the invariant that the simulator's version of the choreography
matches the real one on \emph{public values},
and the ideal choreography matches the real one on \emph{trusted values}.
Next, we define what it means for two terms to agree on public/trusted values.

\begin{definition}[Closing Substitution]
  A closing substitution
  \( \csecure{\cs}{\ctx} \)
  is a mapping from variables to values
  \( \cs : \dom(\ctx) \to \val! \).
\end{definition}

\begin{definition}[Channel Label]
  A channel's label derives from the labels of its endpoints:
  \begin{mathpar}
    \labelof{\mglobal{\ch_1}{\ch_2}{\val}}
    =
    \labelof{\ch_1 \ch_2}
    =
    \labelof{\ch_1} \vee \labelof{\ch_2}

    \labelof{\ainput{\msg}}
    =
    \labelof{\aoutput{\msg}}
    =
    \labelof{\msg}
    .
  \end{mathpar}

  We let
  \( \labelof{\advprot} = \ifbottom \)
  and
  \( \labelof{\envprot} = \strongest \),
  which leads to
  \(
    \labelof{\envprot\, \ch}
    =
    \labelof{\ch\, \envprot}
    =
    \labelof{\ch}
  \)
  and
  \(
    \labelof{\advprot\, \ch}
    =
    \labelof{\ch\, \advprot}
    =
    \integrity{\labelof{\ch}}
  \)
  (communication with the adversary is public,
  and is trusted only if the other endpoint is).
\end{definition}

\begin{definition}[Syntactic \lbl*-Equivalence]
  \label{def:low-equivalence}
  For a set of labels \( \lbl* \subseteq \lbl! \),
  define \( \equalsat{\lbl*} \) as follows.
  \begin{itemize}
    \item
      \(
        \mglobal{\ch_1}{\ch_2}{\val_1}
        \equalsat{\lbl*}
        \mglobal{\ch_1}{\ch_2}{\val_2}
      \)
      if
      \( \labelof{\ch_1 \ch_2} \in \lbl* \)
      implies
      \( \val_1 = \val_2 \).

    \item
      \( \ainput{\msg_1} \equalsat{\lbl*} \ainput{\msg_2} \)
      and
      \( \aoutput{\msg_1} \equalsat{\lbl*} \aoutput{\msg_2} \)
      if
      \( \msg_1 \equalsat{\lbl*} \msg_2 \).

    \item
      \( \s_1 \equalsat{\lbl*} \s_2 \)
      if there exist
      \( \ctx_1 \), \( \ctx_2 \), and \s
      with
      \( \ssecure[(\ctx_1 \munion \ctx_2)]{\s} \),
      and substitutions \( \csecure{\cs_1, \cs_2}{\ctx_2} \) such that
      \( \cs_1(\s) = \s_1 \)
      and
      \( \cs_2(\s) = \s_2 \).
Additionally, for
      \( (\csing{\x}{\lbl}{\h}) \in (\ctx_1 \munion \ctx_2) \),
      we require
      \( \labelof{\h} \actsfor \lbl \),
      and for
      \( (\csing{\x}{\lbl}{\h}) \in \ctx_2 \),
      we require
      \( \lbl \not\in \lbl* \).

    \item
      \( \buffer_1 \equalsat{\lbl*} \buffer_2 \)
      if
      \( \buffer_1(\ch_1 \ch_2) = \buffer_2(\ch_1 \ch_2) \)
      for all \( \ch_1 \) and \( \ch_2 \) such that
      \( \labelof{\ch_1 \ch_2} \in \lbl* \).

    \item
      \(
        \proc{\h*}{\buffer_1}{\s_1}
        \equalsat{\lbl*}
        \proc{\h*}{\buffer_2}{\s_2}
      \)
      if
      \( \buffer_1 \equalsat{\lbl*} \buffer_2 \)
      and
      \( \s_1 \equalsat{\lbl*} \s_2 \).
  \end{itemize}
\end{definition}

We instantiate \cref{def:low-equivalence} with $\lbl* = \public$
for agreement on public values,
and with $\lbl* = \trusted$ for agreement on trusted values.
\Cref{def:low-equivalence} requires the two terms to have the same structure,
but allows some values $\val \in \val!$ (those with labels \emph{not} in \lbl*)
to differ between them.
For example, two messages can only be equivalent if they are on the same channel.
Additionally, they must carry the same value if the channel is public and we
are considering public equality ($\publiceq$);
they are allowed to carry different values otherwise.
Action and buffer equivalence simply lift the definition for messages.
Equivalence for statements demands further explanation.

Values are fixed constants and can be assigned any label.
It is therefore not immediate which values should be allowed to differ
between statements.
For example, consider the following statements that have the
same structure but differ in the value of \x:
\begin{mathpar}
  \begin{program}
    \scomment{\( \s_1 \)}
    \\
    \slethhead{\x}{\alice}{\vfalse}
    \\
    \smovehead{\alice}{\x}{\bob}{\y}
  \end{program}

  \begin{program}
    \scomment{\( \s_2 \)}
    \\
    \slethhead{\x}{\alice}{\vtrue}
    \\
    \smovehead{\alice}{\x}{\bob}{\y}
  \end{program}
\end{mathpar}
The intuition behind \cref{def:low-equivalence} is that $\s_1$ and $\s_2$
are equivalent if \x can be treated as secret/untrusted.
To check that, \cref{def:low-equivalence} abstracts out values where the two
statement differ to find a common statement,
and type-checks the generalized statement in a context where all introduced
variables are marked as secret/untrusted.
For example, we could pick \s as follows
\begin{program}
  \scomment{\( \s \)}
  \\
  \slethhead{\x}{\alice}{\x'}
  \\
  \smovehead{\alice}{\x}{\bob}{\y}
\end{program}
along with substitutions
\( \cs_1 = \set{ \x' \mapsto \vfalse } \)
and
\( \cs_2 = \set{ \x' \mapsto \vtrue } \).
If \s can be typed under a context where $\x'$ is considered secret,
then $\s_1 \publiceq \s_2$.
However, if \bob has a public label (is dishonest), for example,
then there is no such context.

\Cref{def:low-equivalence} splits the context into $\ctx_1$ and $\ctx_2$,
with the substitutions only assigning values for variables in $\ctx_2$.
Context $\ctx_1$ is added to allow relating open terms,
which is needed for inductive cases of some proofs.

For the rest of this section,
we assume $\lbl* = \public$ or $\lbl* = \trusted$.
Moreover, whenever $\ereceive{\h}$ or $\esend{\aexp}{\h}$ appears in a
program, we assume $\malicious{\h}$
(which is ensured by $\corrupt{\cdot}$).\footnote{
  Our results hold for more general \lbl*, but we do not need this generality.
}

Equivalent choreographies remain equivalent given equivalent inputs
and after producing outputs on the same host.

\begin{lemma}[Equivalence Preservation]
  \label{thm:equivalence-preservation}
  Assume
  \( \s_1 \equalsat{\lbl*} \s_2 \)
  and
  \( \s_1 \realstepsto!{\action_1} \s_1' \)
  without using \cref{expression:declassify-real} or \cref{expression:endorse-real}.
  \begin{itemize}
    \item
      If
      \( \action_1 = \ainput{\msg_1} \),
      then
      \( \s_2 \realstepsto!{\ainput{\msg_2}} \s_2' \)
      with
      \( \s_1' \equalsat{\lbl*} \s_2' \)
      for \emph{all}
      \( \msg_2 \equalsat{\lbl*} \msg_1 \).

    \item
      If
      \( \action_1 = \aoutput{\msg_1} \),
      then
      \( \s_2 \realstepsto!{\aoutput{\msg_2}} \s_2' \)
      with
      \( \s_1' \equalsat{\lbl*} \s_2' \)
      for \emph{some} \( \msg_2 \) with
      \( \actor{\aoutput{\msg_2}} = \actor{\aoutput{\msg_1}} \).
  \end{itemize}
\end{lemma}
\begin{proof}
  By \cref{def:low-equivalence},
  there exists \s such that
  \( \ssecure[\ctx]{\s} \),
  \( \cs_1(\s) = \s_1 \),
  and
  \( \cs_2(\s) = \s_2 \)
  for some $\ctx = (\ctx_1 \munion \ctx_2)$
  and $\csecure{\cs_1, \cs_2}{\ctx_2}$.
We proceed by induction on the stepping judgment.
In all cases, stepping on $\s_1$ forces certain atomic expressions \aexp
  to be values \val (as opposed to variables \x);
  the same expressions in $\s_2$ must also be values since $\cs_2$ substitutes
  for the same variables as $\cs_1$.
We appeal to this fact implicitly.
  \begin{pcases}
    \prule{statement:let}
      We have
      \begin{mathpar}
        \s_1
        =
        \sleth{\x}{\h}{\e_1}{\s_1''}

        \s_2
        =
        \sleth{\x}{\h}{\e_2}{\s_2''}

        \s
        =
        \sleth{\x}{\h}{\e}{\s''}
      \end{mathpar}
      and
      \begin{mathpar}
        \h \says \e_1
        \realstepsto{\action_1}
        \val_1

        \h \says \e_2
        \realstepsto{\action_2}
        \val_2
      \end{mathpar}
with $\actor{\action_1} = \actor{\action_2} = \h$.
Inversion on \( \ssecure{\s} \) gives
      \begin{mathpar}
        \inferrule
          [\cref{ssecure:let}]
          {
            \esecure{\e}{\lbl}
            \\
            \labelof{\h} \actsfor \lbl
            \\
            \ssecure[\ctx \munion \csing{\x}{\h}{\lbl}]{\s''}
          }
          {
            \ssecure{\s}
          }
      \end{mathpar}

      In each case, we either prove $\val_1 = \val_2$ or $\lbl \not\in \public$.
When $\val_1 = \val_2$, we define
      \( \s' = \subst{\x}{\val}{\s''} \).
We then have
      \( \ssecure[\ctx]{\s'} \) by \cref{thm:substitution},
      \( \cs_1(\s') = \subst{\x}{\val}{\s_1''} = \s_1' \),
      and
      \( \cs_2(\s') = \subst{\x}{\val}{\s_2''} = \s_2' \),
      so
      \( \s_1' \equalsat{\lbl*} \s_2' \).

      When $\lbl \not\in \public$, we define
      \( \ctx_2' = (\ctx_2 \munion \csing{\x}{\h}{\lbl}) \),
      \( \cs_1' = \cs_1 \cup \set{\x \mapsto \val_1} \),
      \( \cs_2' = \cs_2 \cup \set{\x \mapsto \val_2} \),
      which ensures
      \( \cs_1'(\s'') = \subst{\x}{\val_1}{\s_1''} = \s_1' \)
      and
      \( \cs_2'(\s'') = \subst{\x}{\val_2}{\s_2''} = \s_2' \).
Note that $\ctx_2$ satisfies the requirements of \cref{def:low-equivalence},
      and $\csecure{\cs_1', \cs_2'}{\ctx_2'}$,
      so
      \( \s_1' \equalsat{\lbl*} \s_2' \).

      We case on the expression stepping relation to show one of the requirements.
\begin{pcases}
        \prule{expression:operator}
          We have
          \begin{mathpar}
            \e_1
            =
            \eapplyop{\f}{\aexp_1^1, \dots, \aexp_1^n}

            \e_2
            =
            \eapplyop{\f}{\aexp_2^1, \dots, \aexp_2^n}

            \e
            =
            \eapplyop{\f}{\aexp^1, \dots, \aexp^n}
            .
          \end{mathpar}
If all $\aexp^i$ are values, then $\aexp_1^i = \aexp^i = \aexp_2^i$
          and $\val_1 = \val_2$.
Otherwise, let $\aexp^i = \x^i$.
Inversion on $\esecure{\e}{\lbl}$ gives
          \( (\csing{\x^i}{\h}{\lbl'}) \in \ctx_2 \)
          for $\lbl' \not\in \lbl*$ and $\lbl' \flowsto \lbl$,
          which implies $\lbl \not\in \lbl*$.

        \prule{expression:declassify-real}
          Deliberately excluded;
          handled by \cref{thm:matching-after-declassify}.

        \prule{expression:declassify-skip}
          Same as the case for \cref{expression:operator}.

        \prule{expression:endorse-real}
          Deliberately excluded;
          handled by \cref{thm:matching-after-endorse}.

        \prule{expression:endorse-skip}
          Same as the case for \cref{expression:operator}.

        \prule{expression:input}
          We have
          \begin{mathpar}
            \e_1
            =
            \e_2
            =
            \e
            =
            \einput{\h}
            .
          \end{mathpar}
Assume $\lbl \in \lbl*$ since we are done otherwise.
Inversion on $\esecure{\e}{\lbl}$ gives
          \( \labelof{\h} \flowsto \lbl \),
          so \( \labelof{\h} \in \lbl* \),
          which means
          \( \labelof{\envprot \h} = \labelof{\h} \in \lbl* \).
Thus,
          \(
            \aglobalreceive{\envprot}{\h}{\val_1}
            =
            \action_1
            =
            \action_2
            =
            \aglobalreceive{\envprot}{\h}{\val_2}
          \)
          and
          \( \val_1 = \val_2 \).

        \prule{expression:input-malicious}
          We have $\val_1 = \vunit = \val_2$.

        \prule{expression:output}
          We have $\val_1 = \vunit = \val_2$.

        \prule{expression:output-malicious}
          We have $\val_1 = \vunit = \val_2$.

        \prule{expression:receive-real}
          We have
          \begin{mathpar}
            \e_1
            =
            \e_2
            =
            \e
            =
            \ereceive{\h'}
            .
          \end{mathpar}
We have $\malicious{\h'}$ by assumption,
          and $\ispublic{\labelof{\h'}}$ by \cref{def:valid-attack}.

          If $\lbl* = \public$, then
          $\labelof{\h'} \in \lbl*$ so $\labelof{\h' \h} \in \lbl*$.
This gives
          \(
            \aglobalreceive{\h'}{\h}{\val_1}
            =
            \action_1
            =
            \action_2
            =
            \aglobalreceive{\h'}{\h}{\val_2}
          \),
          so $\val_1 = \val_2$.

          If $\lbl* = \trusted$, then
          inversion on $\esecure{\e}{\lbl}$ gives
          \( \integrity{\labelof{\h'}} \flowsto \lbl \).
Since $\labelof{\h'} \not\in \trusted = \lbl*$,
          we have $\lbl \not\in \lbl*$.

        \prule{expression:send-real}
          We have $\val_1 = \vunit = \val_2$.
      \end{pcases}

    \prule{statement:move-real}
      We have
      \begin{mathpar}
        \s_1
        =
        \smove{\h_1}{\val_1}{\h_2}{\x}{\s_1''}

        \s_2
        =
        \smove{\h_1}{\val_2}{\h_2}{\x}{\s_2''}

        \s
        =
        \smove{\h_1}{\aexp}{\h_2}{\x}{\s''}
      \end{mathpar}
      and
      \begin{mathpar}
        \s_1
        \realstepsto{\aglobalsend{\h_1}{\h_2}{\val_1}}
        \subst{\x}{\val_1}{\s_1''}

        \s_2
        \realstepsto{\aglobalsend{\h_1}{\h_2}{\val_2}}
        \subst{\x}{\val_2}{\s_2''}
        .
      \end{mathpar}

      By inversion on \( \ssecure{\s} \), we have
      \begin{mathpar}
        \inferrule
          [\cref{ssecure:move}]
          {
            \aesecure{\aexp}{\lbl}[\h_1]
            \\
            \labelof{\h_2} \actsfor \lbl
            \\
            \ssecure[\ctx \munion \csing{\x}{\h_2}{\lbl}]{\s''}
          }
          {
            \ssecure{\smove{\h_1}{\aexp}{\h_2}{\x}{\s''}}
          }
      \end{mathpar}
We case on \aexp.
If $\aexp = \val$ for some \val, then
      $\val_1 = \cs_1(\val) = \val = \cs_2(\val) = \val_2$.
Additionally,
      \( \ssecure[\ctx]{\subst{\x}{\val}{\s''}} \)
      by \cref{thm:substitution},
      \( \cs_1(\subst{\x}{\val}{\s''}) = \subst{\x}{\val}{\s_1''} \),
      and
      \( \cs_2(\subst{\x}{\val}{\s''}) = \subst{\x}{\val}{\s_2''} \),
      so
      \( \subst{\x}{\val_1}{\s_1''} \equalsat{\lbl*} \subst{\x}{\val_2}{\s_2''} \).

      Otherwise, $\aexp = \x'$ for some $\x' \in \dom(\ctx_2)$.
By inversion on
      \( \aesecure{\aexp}{\lbl}[\h_1] \),
      we have $(\csing{\x}{\h_1}{\lbl'}) \in \ctx_2$
      for some $\lbl' \flowsto \lbl$.
Since $\lbl' \not\in \lbl*$ and $\lbl' \flowsto \lbl$,
      we have $\lbl \not\in \lbl*$.
Define
      \( \ctx_2' = (\ctx_2 \munion \csing{\x}{\h_2}{\lbl}) \),
      \( \cs_1' = \cs_1 \cup \set{\x \mapsto \val_1} \),
      \( \cs_2' = \cs_2 \cup \set{\x \mapsto \val_2} \),
      which ensures
      \( \cs_1'(\s'') = \subst{\x}{\val_1}{\s_1''} \)
      and
      \( \cs_2'(\s'') = \subst{\x}{\val_2}{\s_2''} \).
Note that $\ctx_2$ satisfies the requirements of \cref{def:low-equivalence},
      and $\csecure{\cs_1', \cs_2'}{\ctx_2'}$,
      so
      \( \subst{\x}{\val_1}{\s_1''} \equalsat{\lbl*} \subst{\x}{\val_2}{\s_2''} \).

    \prule{statement:select-real}
      We have
      \begin{mathpar}
        \s_1
        =
        \sselect{\h_1}{\val_1}{\h_2}{\s_1'}

        \s_2
        =
        \sselect{\h_1}{\val_2}{\h_2}{\s_2'}

        \s
        =
        \sselect{\h_1}{\val}{\h_2}{\s'}
      \end{mathpar}
      and
      \begin{mathpar}
        \s_1
        \realstepsto{\aglobalsend{\h_1}{\h_2}{\val_1}}
        \s_1'

        \s_2
        \realstepsto{\aglobalsend{\h_1}{\h_2}{\val_2}}
        \s_2'
        .
      \end{mathpar}

      $\ssecure{\s'}$ (by inversion on $\ssecure{\s}$),
      $\cs_1(\s') = \s_1'$,
      and
      $\cs_2(\s') = \s_2'$,
      thus
      $\s_1' \equalsat{\lbl*} \s_2'$.

    \prule{statement:if}
      We have
      \begin{mathpar}
        \s_1
        =
        \sifh{\val_1}{\h}{\s_1^1}{\s_1^2}

        \s_2
        =
        \sifh{\val_2}{\h}{\s_2^1}{\s_2^2}

        \s
        =
        \sifh{\aexp}{\h}{\s^1}{\s^2}
      \end{mathpar}
      and
      \begin{mathpar}
        \s_1
        \realstepsto{\aglobalinternal{\h}}
        \s_1^i

        \s_2
        \realstepsto{\aglobalinternal{\h}}
        \s_2^j
        .
      \end{mathpar}

      Inversion on \( \ssecure{\s} \)
      (which must be by \cref{ssecure:if})
      gives
      \( \aesecure{\aexp}{\ifbottom}[\h] \).
Additionally,
      $\cs_1(\aexp)$ and $\cs_2(\aexp)$ are values,
      so $\freevars{\aexp} \subseteq \ctx_2$,
      meaning
      \( \aesecure[\ctx_2]{\aexp}{\ifbottom}[\h] \).
Since $\ifbottom \in \lbl*$ for all attacks
      (recall \cref{def:valid-attack}),
      and $\ctx_2$ only contains variables with labels not in \lbl*,
      $\aexp$ must be a value,
      that is, $\aexp = \val$ for some $\val$.
Then,
      $\val_1 = \cs_1(\val) = \val = \cs_2(\val) = \val_2$,
      so $i = j$.
Finally, we have
      $\s_1^i \equalsat{\lbl*} \s_2^j$
      since
      $\ssecure{\s_i}$ (by inversion on $\ssecure{\s}$),
      $\cs_1(\s^i) = \s_1^i$,
      and
      $\cs_2(\s^i) = \s_2^i = \s_2^j$.

    \prule{statement:case}
      Impossible by inversion on \( \ssecure{\s} \).

    \prule{statement:sequential}
      Immediate by the induction hypothesis.

    \prule{statement:delay}
      We have
      \begin{mathpar}
        \s_1
        =
        \ectxof[\ectx_1]{\s_1''}

        \s_2
        =
        \ectxof[\ectx_2]{\s_2''}

        \s
        =
        \ectxof[\ectx]{\s''}
      \end{mathpar}
      and
      \begin{mathpar}
        \inferrule
          {
            \s_1'' \realstepsto!{\action} \s_1'''
            \\
            \actor{\action} \notin \hosts{\ectx_1}
          }
          {
            \s_1
            \realstepsto!{\action}
            \ectxof[\ectx_1]{\s_1'''}
          }

        \inferrule
          {
            \s_2'' \realstepsto!{\action} \s_2'''
            \\
            \actor{\action} \notin \hosts{\ectx_1}
          }
          {
            \s_2
            \realstepsto!{\action}
            \ectxof[\ectx_2]{\s_2'''}
          }
      \end{mathpar}
Note that
      \( \ssecure[(\ctx_1 \munion \ctx_1' \munion \ctx_2)]{\s''} \)
      where $\ctx_1'$ are the variables defined by \ectx
      (which must be the same as the ones defined by $\ectx_1$ and $\ectx_2$).
Thus,
      \( \s_1'' \equalsat{\lbl*} \s_2'' \)
      through $\s''$, $\cs_1$, and $\cs_2$,
      and we can apply induction hypothesis to get
      \( \s_1''' \equalsat{\lbl*} \s_2''' \).
This then gives
      \(
        \s_1'
        =
        \ectxof[\ectx_1]{\s_1'''}
        \equalsat{\lbl*}
        \ectxof[\ectx_2]{\s_2'''}
        \s_2'
      \).

    \prule{statement:if-delay}
      Using the induction hypotheses similar to \cref{statement:delay}.
    \qedhere
  \end{pcases}
\end{proof}

Public-equivalent choreographies produce public-equivalent outputs.

\begin{lemma}[Public Outputs]
  \label{thm:public-outputs}
  If
  \( \s_1 \publiceq \s_2 \),
  \( \s_1 \realstepsto!{\aoutput{\msg_1}} \),
  and
  \( \s_2 \realstepsto!{\action_2} \)
  with
  \( \actor{\aoutput{\msg_1}} = \actor{\action_2} \),
  then
  \( \aoutput{\msg_1} \publiceq \action_2 \).
\end{lemma}
\begin{proof}
  By \cref{def:low-equivalence},
  there exists \s such that
  \( \ssecure[\ctx]{\s} \),
  \( \cs_1(\s) = \s_1 \),
  and
  \( \cs_2(\s) = \s_2 \)
  for some $\ctx = (\ctx_1 \munion \ctx_2)$
  and $\csecure{\cs_1, \cs_2}{\ctx_2}$.
We proceed by induction on the two stepping judgments,
  which must be by the same rule since $\s_1$ and $\s_2$ have the same structure.

  Cases for \cref{statement:delay,statement:if-delay} follow from the
  induction hypotheses.
Cases for \cref{expression:operator,,expression:declassify-real,,expression:declassify-skip,,expression:endorse-real,,expression:endorse-skip,,expression:output-malicious,,statement:if}
  are immediate since both actions are internal, i.e.,
  \( \aoutput{\msg_1} = \aglobalinternal{\h} = \action_2 \)
  for some \h, which implies
  \( \aoutput{\msg_1} \publiceq \action_2 \).
We detail the remaining cases.
  \begin{pcases}
    \prule{statement:let}
      We have
      \begin{mathpar}
        \s_1
        =
        \sleth{\x}{\h}{\e_1}{\s_1'}

        \s_2
        =
        \sleth{\x}{\h}{\e_2}{\s_2'}

        \s
        =
        \sleth{\x}{\h}{\e}{\s'}
      \end{mathpar}
      and
      \begin{mathpar}
        \h \says \e_1
        \realstepsto{\aglobalsend{\h}{\ch}{\val_1}}

        \h \says \e_2
        \simstepsto{\aglobalsend{\h}{\ch}{\val_2}}
        .
      \end{mathpar}
Inversion on \( \ssecure{\s} \) gives
      \( \esecure{\e}{\lbl} \) for \( \labelof{\h} \actsfor \lbl \).
We case on the expression stepping relations.
\begin{pcases}
        \prule{expression:output}
          We have \( \ch = \envprot \), \( \nonmalicious{\h} \), and
          \begin{mathpar}
            \e_1
            =
            \eoutput{\val_1}{\h}

            \e_2
            =
            \eoutput{\val_2}{\h}

            \e
            =
            \eoutput{\aexp}{\h}
            .
          \end{mathpar}
If \( \labelof{\h} \not\in \public \),
          then
          \(
            \aoutput{\msg_1}
            =
            \aglobalsend{\h}{\envprot}{\val_1}
            \publiceq
            \aglobalsend{\h}{\envprot}{\val_2}
            =
            \action_2
          \)
          immediately, so assume \( \ispublic{\labelof{\h}} \).
Inversion on \( \esecure{\e}{\lbl} \) gives
          \( \aesecure{\aexp}{\labelof{\h}} \).
Since \( \labelof{\h} \in \public \),
          $\aexp = \val$ for some \val, meaning
          \( \val_1 = \cs_1(\aexp) = \val = \cs_2(\aexp) = \val_2 \),
          so
          \(
            \aoutput{\msg_1}
            =
            \aglobalsend{\h}{\envprot}{\val}
            =
            \aglobalsend{\h}{\envprot}{\val}
            =
            \action_2
          \),
          and
          \(
            \aoutput{\msg_1}
            \publiceq
            \action_2
          \).

        \prule{expression:send-real}
          We have \( \ch = \h' \) and
          \begin{mathpar}
            \e_1
            =
            \esend{\val_1}{\h'}

            \e_2
            =
            \esend{\val_2}{\h'}

            \e
            =
            \esend{\aexp}{\h'}
            .
          \end{mathpar}
If $\aexp = \val$ for some \val, then $\val_1 = \val_2$ and
          we are done, so assume $\aexp = \x'$ for some $\x'$.
Inversion on \( \esecure{\e}{\lbl} \) gives
          \( \aesecure{\aexp}{\confidentiality{\labelof{\h'}}} \).
We then have
          \( (\csing{\x'}{\lbl'}{\h}) \in \ctx_2 \)
          with
          $\labelof{\h} \actsfor \lbl'$,
          $\lbl' \not\in \public$,
          and
          $\lbl' \flowsto \confidentiality{\labelof{\h'}}$.
Then,
          \begin{align*}
            \lbl' \not\in \public
              \wedge
              \labelof{\h} \actsfor \lbl'
            &\implies
            \labelof{\h} \not\in \public
            \\
            \lbl' \flowsto \confidentiality{\labelof{\h'}}
              \wedge \lbl' \not\in \public
            \implies
            \confidentiality{\labelof{\h'}} \not\in \public
            &\implies
            \labelof{\h'} \not\in \public
            .
          \end{align*}
          Thus,
          \( \labelof{\h \h'} \not\in \public \)
          and
          \(
            \aoutput{\msg_1}
            =
            \aglobalsend{\h}{\h'}{\val_1}
            \publiceq
            \aglobalsend{\h}{\h'}{\val_2}
            =
            \action_2
          \).
      \end{pcases}

    \prule{statement:move-real}
      We have
      \begin{mathpar}
        \s_1
        =
        \smove{\h_1}{\val_1}{\h_2}{\x}{\s_1''}

        \s_2
        =
        \smove{\h_1}{\val_2}{\h_2}{\x}{\s_2''}

        \s
        =
        \smove{\h_1}{\aexp}{\h_2}{\x}{\s''}
      \end{mathpar}
      and
      \begin{mathpar}
        \s_1
        \realstepsto{\aglobalsend{\h_1}{\h_2}{\val_1}}

        \s_2
        \simstepsto{\aglobalsend{\h_1}{\h_2}{\val_2}}
        .
      \end{mathpar}

      By inversion on \( \ssecure{\s} \), we have
      \begin{mathpar}
        \inferrule
          [\cref{ssecure:move}]
          {
            \aesecure{\aexp}{\lbl}[\h_1]
            \\
            \labelof{\h_2} \actsfor \lbl
            \\
            \ssecure[\ctx \munion \csing{\x}{\h_2}{\lbl}]{\s''}
          }
          {
            \ssecure{\smove{\h_1}{\aexp}{\h_2}{\x}{\s''}}
          }
      \end{mathpar}

      We case on \aexp.
If $\aexp = \val$ for some \val, then
      $\val_1 = \cs_1(\val) = \val = \cs_2(\val) = \val_2$.
So
      $\action_1 = \aglobalsend{\h_1}{\h_2}{\val} = \action_2$
      and
      $\action_1 \publiceq \action_2$.

      Otherwise, $\aexp = \x'$ for some $\x' \in \dom(\ctx_2)$.
By inversion on
      \( \aesecure{\aexp}{\lbl}[\h_1] \),
      we have $(\csing{\x}{\h_1}{\lbl'}) \in \ctx_2$
      for some $\lbl' \flowsto \lbl$.
Then,
      \begin{align*}
        \lbl' \not\in \public
          \wedge \lbl' \flowsto \lbl
        &\implies
        \lbl \not\in \public
        \\
        \lbl' \not\in \public
          \wedge \labelof{\h_1} \actsfor \lbl'
        &\implies
        \labelof{\h_1} \not\in \public
        \\
        \lbl \not\in \public
          \wedge \labelof{\h_2} \actsfor \lbl
        &\implies
        \labelof{\h_2} \not\in \public
        \\
        \labelof{\h_1} \not\in \public
          \wedge \labelof{\h_2} \not\in \public
        &\implies
        \labelof{\h_1 h_2} \not\in \public
        .
      \end{align*}
Since $\labelof{\h_1 h_2} \not\in \public$,
      \(
        \action_1
        =
        \aglobalsend{\h_1}{\h_2}{\val_1}
        \publiceq
        \aglobalsend{\h_1}{\h_2}{\val_2}
        =
        \action_2
      \).

    \prule{statement:select-real}
      We have
      \begin{mathpar}
        \s_1
        =
        \sselect{\h_1}{\val_1}{\h_2}{\s_1'}

        \s_2
        =
        \sselect{\h_1}{\val_2}{\h_2}{\s_2'}

        \s
        =
        \sselect{\h_1}{\val}{\h_2}{\s'}
      \end{mathpar}
      and
      \begin{mathpar}
        \s_1
        \realstepsto{\aglobalsend{\h_1}{\h_2}{\val_1}}

        \s_2
        \simstepsto{\aglobalsend{\h_1}{\h_2}{\val_2}}
        .
      \end{mathpar}

      Note that selection statements do not allow variables to be communicated,
      so \s sending \val (rather than \aexp) is not a mistake.
Thus, we have
      $\val_1 = \cs_1(\val) = \val = \cs_2(\val) = \val_2$,
      which means
      $\action_1 = \aglobalsend{\h_1}{\h_2}{\val} = \action_2$,
      which in turn means
      $\action_1 \publiceq \action_2$.
    \qedhere
  \end{pcases}
\end{proof}

Trusted-equivalent choreographies produce trusted-equivalent outputs
\emph{for the environment}.
The statement does not apply to intermediate messages between hosts
because untrusted values can be sent on trusted channels
(e.g., a trusted third party can process untrusted values from other hosts).

\begin{lemma}[Trusted Outputs]
  \label{thm:trusted-outputs}
  If
  \( \s_1 \publiceq \s_2 \),
  \( \s_1 \realstepsto!{\aglobalsend{\h}{\envprot}{\val_1}} \),
  and
  \( \s_2 \realstepsto!{\action_2} \)
  with
  \( \actor{\action_2} = \h \),
  then
  \( \aglobalsend{\h}{\envprot}{\val_1} \publiceq \action_2 \).
\end{lemma}
\begin{proof}
  By induction on the stepping relations.
Inductive cases are handled similarly to \cref{thm:public-outputs}.
The only other relevant case is under \cref{statement:let}
  with \cref{expression:output}.
The argument is similar to the case in \cref{thm:public-outputs},
  but holds because only trusted values can be \koutput to nonmalicious hosts.
\end{proof}

The simulator's view of the real choreography stays accurate on public values.

\begin{lemma}[Matching Steps for Public Equivalence]
  \label{thm:public-matching}
  Assume
  \( \s_1 \publiceq \s_2 \)
  and
  \( \s_1 \realstepsto!{\action_1} \s_1' \)
  without using \cref{expression:declassify-real} or \ref{expression:endorse-real}.
  \begin{itemize}
    \item
      If
      \( \action_1 = \ainput{\msg_1} \),
      then
      \( \s_2 \simstepsto!{\ainput{\msg_2}} \s_2' \)
      with
      \( \s_1' \publiceq \s_2' \)
      for \emph{all}
      \( \msg_2 \publiceq \msg_1 \).

    \item
      If
      \( \action_1 = \aoutput{\msg_1} \),
      then
      \( \s_2 \simstepsto!{\aoutput{\msg_2}} \s_2' \)
      with
      \( \s_1' \publiceq \s_2' \)
      for \emph{some}
      \( \msg_2 \publiceq \msg_1 \).
  \end{itemize}

  In addition, the statement holds with the roles of \( \realstepsto!{} \) and
  \( \simstepsto!{} \) reversed
  (excluding \cref{expression:declassify-sim,expression:endorse-sim} instead).
\end{lemma}
\begin{proof}
  Follows immediately from
  \cref{thm:equivalence-preservation,thm:public-outputs}
  since \( \simstepsto!{} \) is equivalent to \( \realstepsto!{} \) except for
  \cref{expression:declassify-real,expression:endorse-real},
  which we exclude.
\end{proof}

The ideal choreography stays accurate to the real choreography on trusted values.

\begin{lemma}[Matching Steps for Trusted Equivalence]
  \label{thm:trusted-matching}
  Assume
  \( \s_1 \trustedeq \s_2 \)
  and
  \( \s_1 \realstepsto!{\action_1} \s_1' \)
  without using \cref{expression:declassify-real} or \ref{expression:endorse-real}.
  \begin{itemize}
    \item
      If
      \( \action_1 = \aglobalreceive{\envprot}{\h}{\val_1} \),
      then
      \( \s_2 \idealstepsto!{\action_2} \s_2' \)
      with
      \( \s_1' \trustedeq \s_2' \)
      for \emph{all}
      \( \action_2 \trustedeq \action_1 \).

    \item
      If
      \( \action_1 = \aglobalsend{\h}{\envprot}{\val_1} \),
      then
      \( \s_2 \idealstepsto!{\action_2} \s_2' \)
      with
      \( \s_1' \trustedeq \s_2' \)
      for \emph{some}
      \( \action_2 \trustedeq \action_1 \).

    \item
      Otherwise,
      \( \s_2 \idealstepsto!{\aglobalinternal{\h}} \s_2' \)
      with
      \( \s_1' \trustedeq \s_2' \)
      and
      \( \actor{\action_1} = \h \).
  \end{itemize}

  In addition, the statement holds with the roles of \( \realstepsto!{} \) and
  \( \idealstepsto!{} \) reversed
  (excluding \cref{expression:declassify,expression:endorse} instead).
\end{lemma}
\begin{proof}
  Follows from
  \cref{thm:equivalence-preservation,thm:trusted-outputs}.
judgment \( \realstepsto!{} \) behaves the same as \( \realstepsto!{} \)
  except it replaces some output messages with internal steps.
Since this does not affect the resulting choreographies (only the actions),
  \cref{thm:equivalence-preservation} applies and shows that the resulting
  choreographies are equivalent.
The one exception to this \cref{expression:receive-real,expression:receive},
  where the ideal choreography proceeds with \vunit instead of receiving a value;
  this value is treated as untrusted so the choreographies still agree on
  trusted values as required.
\end{proof}

Assume the simulator's view agrees with the real choreography on public values,
and the ideal choreography agrees with the real choreography on trusted values.
If the ideal choreography declassifies a value and we feed that value to
the simulator, then all three choreographies remain in agreement.
Only trusted values are declassified,
so the ideal choreography outputs the correct value to the simulator.

\begin{lemma}[Equivalence After Declassify]
  \label{thm:matching-after-declassify}
  Let
  \( \s_1 \publiceq \s_2 \)
  and
  \( \s_1 \trustedeq \s_3 \).
If
  \( \s_1 \realstepsto!{\aglobalinternal{\h}} \s_1' \),
  \( \s_2 \simstepsto!{\aglobalreceive{\h}{\advprot}{\val}} \s_2' \),
  and
  \( \s_3 \idealstepsto!{\aleak{\h}{\val}} \s_3' \),
  then
  \( \s_1' \publiceq \s_2' \)
  and
  \( \s_1' \trustedeq \s_3' \).
\end{lemma}
\begin{proof}
  \newcommand{\pe}[1]{#1_{\mathrm{p}}}\newcommand{\te}[1]{#1_{\mathrm{t}}}By \cref{def:low-equivalence},
  there exist $\pe{\s}$ and $\te{\s}$ such that
  \( \ssecure[\pe{\ctx}]{\pe{\s}} \),
  \( \cs_1(\pe{\s}) = \s_1 \),
  \( \cs_2(\pe{\s}) = \s_2 \),
  and
  \( \ssecure[\te{\ctx}]{\te{\s}} \),
  \( \cs_2'(\te{\s}) = \s_2 \),
  \( \cs_3(\te{\s}) = \s_3 \)
  for
  \( \pe{\ctx} = (\ctx_1 \munion \ctx_2) \),
  \( \csecure{\cs_1, \cs_2}{\ctx_2} \),
  \( \te{\ctx} = (\ctx_3 \munion \ctx_4) \),
  and
  \( \csecure{\cs_2', \cs_3}{\ctx_4} \).

  We proceed by induction on the stepping relations.
Inductive cases (\cref{statement:delay,statement:if-delay}) are handled
  similarly to \cref{thm:public-matching}.
The only remaining case is when the steps are by
  \cref{expression:declassify-real,expression:declassify-sim,expression:declassify},
  respectively.
We have
  \begin{mathpar}
    \s_1
    =
    \sleth{\x}{\h}{\edeclassify{\val_1}{\fromlbl}{\tolbl}}{\s_1''}

    \pe{\s}
    =
    \sleth{\x}{\h}{\edeclassify{\pe{\aexp}}{\fromlbl}{\tolbl}}{\pe{\s}''}

    \s_2
    =
    \sleth{\x}{\h}{\edeclassify{\val_2}{\fromlbl}{\tolbl}}{\s_2''}

    \te{\s}
    =
    \sleth{\x}{\h}{\edeclassify{\te{\aexp}}{\fromlbl}{\tolbl}}{\te{\s}''}

    \s_3
    =
    \sleth{\x}{\h}{\edeclassify{\val}{\fromlbl}{\tolbl}}{\s_3''}
  \end{mathpar}
  where $\issecret{\fromlbl}$ and $\ispublic{\tolbl}$, and
  \begin{mathpar}
    \s_1'
    =
    \subst{\x}{\val_1}{\s_1''}

    \s_2'
    =
    \subst{\x}{\val}{\s_2''}

    \s_3'
    =
    \subst{\x}{\val}{\s_3''}
    .
  \end{mathpar}
We claim $\val_1 = \val$
  ($\val_2$ is ignored by $\s_2$, so it is irrelevant).
By \cref{thm:robust-declassification} and inversion on
  \( \ssecure[\te{\ctx}]{\te{\s}} \),
  we have $\istrusted{\fromlbl}$.
Assume for contradiction that $\te{\aexp} = \te{\x}$ for some $\te{\x}$.
Then, $(\csing{\te{\x}}{\h}{\lbl}) \in \ctx_4$ for
  $\lbl \flowsto \fromlbl$.
However, $\lbl \not\in \trusted$ so $\fromlbl \not\in \trusted$,
  which is a contradiction.
Thus, $\te{\aexp} = \te{\val}$ for some $\te{\val}$.
Then,
  \(
    \val_1
    =
    \cs_2'(\te{\aexp})
    =
    \cs_2'(\te{\val})
    =
    \te{\val}
    =
    \cs_3(\te{\val})
    =
    \cs_3(\te{\aexp})
    =
    \val
  \).

  Finally, let
  \( \pe{\s}' = \subst{\x}{\val}{\pe{\s}''} \)
  and
  \( \te{\s}' = \subst{\x}{\val}{\te{\s}''} \).
We have
  \( \s_1' \publiceq \s_2' \)
  since
  \( \ssecure[\pe{\ctx}]{\pe{\s}'} \)
  (inversion on
  \( \ssecure[\pe{\ctx}]{\pe{\s}} \)
  followed by \cref{thm:substitution}),
  \( \cs_1(\pe{\s}') = \s_1' \),
  and
  \( \cs_2(\pe{\s}') = \s_2' \);
  and we have
\( \s_2' \trustedeq \s_3' \)
  since
  \( \ssecure[\te{\ctx}]{\te{\s}'} \)
  (inversion on
  \( \ssecure[\te{\ctx}]{\te{\s}} \),
  then \cref{thm:substitution}),
  \( \cs_2'(\te{\s}') = \s_2' \),
  and
  \( \cs_3(\te{\s}') = \s_3' \).
\end{proof}

Similarly, if the simulator recreates a value that the ideal choreography
endorses, all three choreographies remain in agreement.
Only public values are endorsed,
so the simulator outputs the correct value to the ideal choreography.

\begin{lemma}[Equivalence After Endorse]
  \label{thm:matching-after-endorse}
  Let
  \( \s_1 \publiceq \s_2 \)
  and
  \( \s_1 \trustedeq \s_3 \).
If
  \( \s_1 \realstepsto!{\aglobalinternal{\h}} \s_1' \),
  \( \s_2 \simstepsto!{\aglobalsend{\advprot}{\h}{\val}} \s_2' \),
  and
  \( \s_3 \idealstepsto!{\amaul{\h}{\val}} \s_3' \),
  then
  \( \s_1' \publiceq \s_2' \)
  and
  \( \s_1' \trustedeq \s_3' \).
\end{lemma}
\begin{proof}
  Dual to \cref{thm:matching-after-declassify},
  but focusing on \( \s_1 \publiceq \s_2 \)
  and using \cref{thm:transparent-endorsement}.
\end{proof}

\Cref{thm:public-matching,thm:trusted-matching} straightforwardly lift from
choreographies \s to processes \process.
\Cref{thm:public-matching} needs an additional condition on buffer equivalence:
for
\( \buffer_1 \publiceq \buffer_2 \),
we require
\( \sizeof{\buffer_1(\ch_1 \ch_2)} = \sizeof{\buffer_2(\ch_1 \ch_2)} \)
when
\( \labelof{\ch_1 \ch_2} \not\in \public \).
That is, the buffers must agree exactly on public channels,
and agree on the number of messages on secret channels.
This condition allows the simulator to keep track of messages on secret channels
even though it cannot read message contents.

\begin{proofsketchof}{thm:correctness-of-ideal-execution}
  We prove simulation as follows.
  \begin{ucproof}
    \newcommand{\simof}[2][\adv]{\operatorname{\simulator}(#1 \parallel #2)}\isimulator
      The simulator has the form \( \simof{\process} \)
      where \process is a public view of the real process.
The simulator runs \process against \adv for all internal messages.
The simulator forwards inputs from \envprot to \adv and \process,
      and forwards messages from \adv destined for \envprot to \envprot.
When the ideal process outputs data through a \kdeclassify expression,
      the simulator inputs this data to \process.
Similarly, when \process outputs data through an \kendorse expressions,
      the simulator forwards this data to the ideal process.\footnote{
        The simulator needs to step the ideal process an additional time
        so that the ideal process pulls the message from its buffer.
This is due to how we define operational rules for processes.
This extra step forces us to use weak bisimulation instead of strong
        bisimulation.
      }

    \irelation
      We maintain the invariant that the simulator's version of the process
      matches the real one on \emph{public values},
      and the ideal process matches the real one on \emph{trusted values}.
More concretely, we define
      \(
        \adv \parallel \process_1
        \rel
        \simof[\adv']{\process} \parallel \process_2
      \)
      if:
      \begin{relations}
        \item
          \label{proof:ideal-adv}
          \( \adv = \adv' \)

        \item
          \label{proof:ideal-public}
          \( \process_1 \publiceq \process \)

        \item
          \label{proof:ideal-trusted}
          \( \process_1 \trustedeq \process_2 \)
        .
      \end{relations}

    \iproof
      We claim $\rel$ is a weak bisimulation.

      Since the simulator's version of the process matches the real one on public
      values (\cref{proof:ideal-public}),
      the adversary in the real configuration has a view identical to the
      adversary running inside of the simulator
      (the adversary only sees public data).
Similarly, since the real process matches the ideal one on trusted values,
      the environment has the same view in both
      (the environment is only sent trusted data).

      \Cref{proof:ideal-adv} is preserved since $\process$ is an
      accurate public view of $\process_1$ (\cref{proof:ideal-public}).
When there are no downgrade actions,
      \cref{thm:public-matching} ensures \cref{proof:ideal-public} is preserved,
      and \cref{thm:trusted-matching} ensures \cref{proof:ideal-trusted} is
      preserved.
\Cref{thm:matching-after-declassify,thm:matching-after-endorse} cover
      the cases with downgrades.
    \qedhere
  \end{ucproof}
\end{proofsketchof}
 \end{toappendix}
 
\section{Endpoint Projection}
\label{sec:endpoint-projection}

The second stage of compilation, endpoint projection,
transforms a choreography into a distributed program.

\subsection{Hybrid Distributed Language}

\begin{figure}
  \begin{syntax}
  \category[Expressions]{\e}
  \alternative{\dots}
  \alternative{\ereceive{\h}}
  \alternative{\esend{\aexp}{\h}}

  \category[Statements]{\s}
  \alternative{\dots}
  \alternative{\scase{\h_1}{\h_2}{\set{\val \mapsto \s_{\val}}_{\val \in \val*}}}

  \category[Processes]{\process}
  \alternative{\proc{\{\h \in \h!\}}{\buffer}{\s}}
\end{syntax}
   \caption{Hybrid distributed syntax as an extension to source
  (\cref{fig:source-syntax}).}
  \label{fig:target-syntax}
\end{figure}

As for choreographies, \cref{fig:target-syntax} specifies the syntax of
hybrid distributed programs by extending the syntax of source programs.
Hybrid distributed programs are configurations containing multiple
processes, with each
process acting on behalf of a single host $\h \in \h!$.
These processes communicate data via \ksend/\kreceive
and agree on control flow via \ksend and \kcase.
The \kcase statement, on receiving a value on the expected channel,
steps to the specified branch.

\paragraph*{Operational Semantics}

Each process transitions using the real stepping rules ($\realstepsto{}$),
and the distributed configuration steps using the parallel composition rules in
\cref{fig:parallel}.

\subsection{Compiling to Distributed Programs}
\label{sec:definition-of-projection}

Given a choreography \s and a host \h, the endpoint projection $\project[\h]{\s}$
defines the local program that \h should run.
The distributed system $\partition{\s}$ is derived by independently projecting
onto each host in the choreography.

Our notion of endpoint projection is entirely standard,
so we defer the formal definition to
\apxref{sec:definition-of-projection-details}.
Our proof is agnostic to how endpoint projection is defined,
and only relies on its soundness and completeness,
properties extensively studied in prior work~\cite{Montesi23, Montesi13, FilipeM20, FilipeMP22, HirschG22}.

\begin{toappendix}
  \label{sec:definition-of-projection-details}

  \begin{figure}
    \begingroup \newcommand{\projectH}[1]{{\project[\h]{#1}}}\newcommand{\ow}{\text{o/w}}\begin{judgments}
  \judgment{\project[\h]{\s} = \s'}
\end{judgments}
\begin{function}{\projectH}
\case{\sleth{\x}{\h'}{\e}{\s}}
  \begin{cases}
    \sleth{\x}{\h'}{\e}{\projectH{\s}}
      & \h' = \h
    \\
    \projectH{\s}
      & \ow
  \end{cases}

\case{\smove{\h_1}{\aexp}{\h_2}{\x}{\s}}
  \begin{cases}
    \sleth{\_}{\h_1}{\esend{\aexp}{\h_2}}{\projectH{\s}}
      & \h_1 = \h
    \\
    \sleth{\x}{\h_2}{\ereceive{\h_1}}{\projectH{\s}}
      & \h_2 = \h
    \\
    \projectH{\s}
      & \ow
  \end{cases}

\case{\sselect{\h_1}{\val}{\h_2}{\s}}
  \begin{cases}
    \sleth{\_}{\h_1}{\esend{\val}{\h_2}}{\projectH{\s}}
      & \h_1 = \h
    \\
    \scase{\h_1}{\h_2}{\set{\val \mapsto \projectH{\s}}}
      & \h_2 = \h
    \\
    \projectH{\s}
      & \ow
  \end{cases}

\case{\sifh{\aexp}{\h'}{\s_1}{\s_2}}
  \begin{cases}
    \sifh{\aexp}{\h'}{\projectH{\s_1}}{\projectH{\s_2}}
      & \h' = \h
    \\
    \merge{\projectH{\s_1}}{\projectH{\s_2}}
      & \ow
  \end{cases}

\case{\sskip}
  \sskip
\end{function}

\bigskip
\newcommand{\mergeH}[2]{{\merge{#1}{#2}}}\begin{judgments}
  \judgment{\merge{\s_1}{\s_2} = \s}
\end{judgments}
\begin{function}{\mergeH}
  \case{\s_1}{\s_2}
  \scase{\h_1}{\h_2}{\set{\val \mapsto \s_{\val}}_{\val \in \val*_1 \cup \val*_2}}
    \where{
      \s_1
      =
      \scase{\h_1}{\h_2}{\set{\val \mapsto \s_{\val}}_{\val \in \val*_1}}
    }
    \where{
      \s_2
      =
      \scase{\h_1}{\h_2}{\set{\val \mapsto \s_{\val}}_{\val \in \val*_2}}
    }
    \where{\text{$\val*_1$ and $\val*_2$ disjoint}}

  \case{\s_1}{\s_2}
  \sleth{\x}{\h}{\e}{\merge{\s_1'}{\s_2'}}
    \where{\s_1 = \sleth{\x}{\h}{\e}{\s_1'}}
    \where{\s_2 = \sleth{\x}{\h}{\e}{\s_2'}}
\end{function}

\bigskip
\begin{judgments}
  \judgment{\project[\h]{\buffer} = \buffer'}
  \and
  \judgment{\project{\process} = \process'}
\end{judgments}
\begin{align*}
  \projectH{\buffer}(\ch_1 \ch_2)
  &=
  \begin{cases}
    \buffer(\ch_1 \ch_2) & \ch_1 \neq \h \wedge \ch_2 = \h
    \\
    \emptylist & \text{otherwise}
  \end{cases}
  \\
  \projectH{\proc{\h*}{\buffer}{\s}}
  &=
  \proc{\h* \cap \set{\h}}{\projectH{\buffer}}{\projectH{\s}}
\end{align*}
\endgroup

     \caption{Endpoint projection: statements, buffers, processes.}
    \label{fig:endpoint-projection}
  \end{figure}

  \Cref{fig:endpoint-projection} formalizes projecting onto a host \h.
Projection keeps \klet statements assigned to \h,
  and removes ones assigned to other hosts.
Communication statements become a \ksend or a \kreceive, or are entirely removed
  depending on whether \h is the sending host, the receiving host, or neither.
Selection statements follow the same logic,
  but are projected as a \ksend expression
  or a \kcase statement with a single branch.

  The most interesting case is \kif statements.
If the \kif statement is placed at \h,
  we perform the \kif as usual and project the branches.
Otherwise, \h does not store the conditional and cannot determine which
  branch should be taken.
In this case, the projections of the two branches must be compatible,
  formalized by a merge function.
Merging requires the two branches to have the same syntactic structure,
  but allows \kcase statements to have disjoint branches,
  which are combined into one.
We elide most cases of the merge function,
  since the proof is agnostic to the details.

  We lift projection to processes:
  projecting a buffer onto \h keeps only messages destined for \h,
  and projecting processes is done componentwise.
The \emph{projection} $\partition{\process}$ of process $\process$
  is the configuration formed by projecting onto each host in $\process$:
  \[
    \partition{\process = \proc{\h*}{\_}{\_}}
    =
    \bigparallel_{\h \in \h*}{\project[\h]{\process}}
    .
  \]
\end{toappendix}

\subsection{Correctness of Endpoint Projection}
\label{sec:correctness-of-projection}

Let $\corrupt{\config}$ remove from $\config$ all processes for
malicious hosts.

\begin{theorem}
  \label{thm:correctness-of-projection}
  If
  \( \ssecure[\emptylist]{\process} \),
  then
  \(
    \using{\corrupt{\process}}{\realstepsto!{}}
    \simulatedBy
    \using{\corrupt{\partition{\process}}}{\realstepsto{}}
  \).
\end{theorem}

A choreography and its endpoint projection match each other action-for-action;
once we prove this fact, showing simulation is trivial since we can pick
$\simulator = \adv$.
This perfect correspondence between a choreography and its projection is studied
extensively in the literature~\cite{Montesi23, Montesi13, FilipeM20, FilipeMP22, HirschG22},
and formalized as \emph{soundness} and \emph{completeness}
of endpoint projection.
However, the standard methods of proving soundness and completeness must
be modified to handle malicious corruption and asynchronous communication.
Existing work relates $\process$ to $\partition{\process}$,
which we must extend to relate $\corrupt{\process}$ to $\corrupt{\partition{\process}}$.
This is trivial since
$\corrupt{\partition{\process}} = \partition{\corrupt{\process}}$.

The presence of \emph{asynchrony} breaks the perfect correspondence between
the projected program and the choreography:
a \ksend/\kreceive pair reduces in two steps in a projected program,
but the corresponding communication statement reduces in
only one.
We follow prior work~\cite{async-choreo, PohjolaGSN22} and add syntactic forms
to choreographies for partially reduced \ksend/\kreceive pairs:
messages that have been sent and buffered but not yet received.
These run-time terms exist only to restore the correspondence,
and are never generated by the compiler.
An additional, simple simulation then shows that a choreography
with these run-time terms simulates one without,
removing the need to reason about run-time terms in other proof steps.
For details, see \apxref{sec:correctness-of-projection-details}.

\nosectionappendix
\begin{toappendix}
  \appendixfor{sec:correctness-of-projection}
  \label{sec:correctness-of-projection-details}

  A choreography and its endpoint projection match each other action-for-action;
once we prove this fact, showing simulation is trivial since we can pick
$\simulator = \adv$.
The choreographic programming literature~\cite{Montesi23, Montesi13, FilipeM20, FilipeMP22, HirschG22}
extensively studies this perfect correspondence between a choreography and
its projection, and formalizes the correspondence as strong bisimulation.

To prove that a choreography \process is bisimilar to its endpoint projection
$\partition{\process}$, we must define a relation $\rel$ between an arbitrary
configuration and process, $\config_1 \rel \process_2$,
and show that $\rel$ is a bisimulation.
The obvious approach is to define
\( \config_1 \rel \process_2 \)
if
\( \config_1 = \partition{\process_2} \),
but this idea fails because $\rel$ is not preserved under stepping.

\begin{lemma}
  \label{thm:bad-choreography-bisimulation}
  Define
  \( \config_1 \rel \process_2 \)
  if
  \( \config_1 = \partition{\process_2} \).
We claim \( \rel \) is \emph{not} a bisimulation.
\end{lemma}
\begin{proof}
  Consider the following choreography and its projection:
\begin{program}
    \scomment{Choreography}
    \\
    \process_2 =
    \sifh{\vtrue}{\alice}
      {\sselect{\alice}{\bob}{\vtrue}{\s_1}}
      {\sselect{\alice}{\bob}{\vfalse}{\s_2}}
\\
    \scomment{\alice}
    \\
    \project[\alice]{\process_2} =
    \sifh{\vtrue}{\alice}
      {\sseq{\esend{\vtrue}{\bob}}{\project[\alice]{\s_1}}}
      {\sseq{\esend{\vfalse}{\bob}}{\project[\alice]{\s_2}}}
\\
    \scomment{\bob}
    \\
    \project[\bob]{\process_2} =
    \scase{\alice}{\bob}{
      \set{
        \vtrue \mapsto \project[\bob]{\s_1}
        ,
        \vfalse \mapsto \project[\bob]{\s_2}
      }
    }
  \end{program}
Let $\config_1 = \partition{\process_2}$;
  we have $\config_1 \rel \process_2$.
Now, host \alice can reduce the \kif statement with an internal step in both
  $\config_1$ and $\process_2$, which gives:
  \begin{program}
    \scomment{Choreography}
    \\
    \process_2' =
    \sselect{\alice}{\bob}{\vtrue}{\s_1}
\\
    \scomment{\alice}
    \\
    \config_1'(\alice) =
    \sseq{\esend{\vtrue}{\bob}}{\project[\alice]{\s_1}}
\\
    \scomment{\bob}
    \\
    \config_1'(\bob) =
    \scase{\alice}{\bob}{
      \set{
        \vtrue \mapsto \project[\bob]{\s_1}
        ,
        \vfalse \mapsto \project[\bob]{\s_2}
      }
    }
  \end{program}
Note that the process for \bob in $\config_1'$ does not match
  $\project[\bob]{\process_2'}$, which is
  \[
    \project[\bob]{\process_2'}
    =
    \scase{\alice}{\bob}{
      \set{
        \vtrue \mapsto \project[\bob]{\s_1}
      }
    }
  \]
  (there is no case for $\vfalse$).

  Thus, we have
  \( \config_1 \rel \process_2 \),
  \( \config_1 \stepsto{\aglobalinternal{\alice}} \config_1' \),
  \( \process_2 \stepsto{\aglobalinternal{\alice}} \process_2' \),
  but it is not the case that \( \config_1' \rel \process_2' \).
\Cref{thm:internal-determinism} implies $\config_1'$ is uniquely determined,
  so there is no other $\config_1''$ related to $\process_2'$ that $\config_1$
  can step to.
Therefore, $\rel$ is not a bisimulation.
\end{proof}

Intuitively, when a choreography reduces an \kif statement,
the branch that is not taken disappears in one step for all hosts.
However, in the projected program,
each host reduces its corresponding \kcase statement separately,
which results in extraneous dead branches during simulation.
This is a known issue in the choreography literature~\cite{Montesi23},
and it does \emph{not} break bisimilarity, but
we need to be smarter about how to define $\rel$.

The solution to the issue raised by \cref{thm:bad-choreography-bisimulation}
is to ignore extraneous branches when defining $\rel$.
Even though the configuration $\config_1'$ has ``leftover'' branches
that projecting the choreography $\process_2'$ does not create,
we know that these branches will never be taken.
So we can ignore these branches when defining $\rel$.

Following \citet{Montesi23},
we define $\process_1 \refinedby \process_2$
if $\process_1$ and $\process_2$ are structurally identical,
except $\process_1$ has at least as many branches in \kcase
statements as $\process_2$.
We lift $\refinedby$ pointwise to configurations.
Now, we define
\( \config_1 \rel \process_2 \)
if
\( \config_1 \refinedby \partition{\process_2} \).
We claim $\rel$ is a bisimulation.
Further,
the proof is split into showing the \emph{soundness} and \emph{completeness}
of endpoint projection.

\begin{lemma}[Soundness of Endpoint Projection]
  \label{thm:projection-soundness}
  If
  \( \config \refinedby \partition{\process} \)
  and
  \( \process \realstepsto!{\action} \process' \),
  then
  \( \config \realstepsto{\action} \config' \)
  for some \( \config' \refinedby \partition{\process'} \).
\end{lemma}

\begin{lemma}[Completeness of Endpoint Projection]
  \label{thm:projection-completeness}
  If
  \( \config \refinedby \partition{\process} \)
  and
  \( \config \realstepsto{\action} \config' \),
  then
  \( \process \realstepsto!{\action} \process' \),
  for some \( \process' \) with
  \( \config' \refinedby \partition{\process'} \).
\end{lemma}

\subsection{Handling Asynchronous Communication}

\begin{figure}
  \begin{syntax}
  \category[Statements]{\s}
  \dots
  \\
  \alternative{\smovepartial{\h_1}{\val}{\h_2}{\x}{\s}}
  \\
  \alternative{\sselectpartial{\h_1}{\val}{\h_2}{\s}}
\end{syntax}
   \caption{The syntax of asynchronous choreographies
    (extends \cref{fig:choreography-syntax}).}
  \label{fig:async-choreography}
\end{figure}

\begin{figure}
  \begin{minipage}{\linewidth}
    \begin{judgments}
  \judgment{\s \asyncstepsto{\action} \s'}
\end{judgments}
\begin{mathpar}
\inferrule
    [\named{\s-Communicate-Send}{statement:move-send}]
    {}
    {
      \smove{\h_1}{\val}{\h_2}{\x}{\s}
      \asyncstepsto{\aglobalsend{\h_1}{\h_2}{\val}}
      \smovepartial{\h_1}{\val}{\h_2}{\x}{\s}
    }

\inferrule
    [\named{\s-Communicate-Receive}{statement:move-receive}]
    {}
    {
      \smovepartial{\h_1}{\val}{\h_2}{\x}{\s}
      \asyncstepsto{\aglobalinternal{\h_2}}
      \subst{\x}{\val}{\s}
    }

\inferrule
    [\named{\s-Select-Send}{statement:select-send}]
    {}
    {
      \sselect{\h_1}{\val}{\h_2}{\s}
      \asyncstepsto{\aglobalsend{\h_1}{\h_2}{\val}}
      \sselectpartial{\h_1}{\val}{\h_2}{\s}
    }

\inferrule
    [\named{\s-Select-Receive}{statement:select-receive}]
    {}
    {
      \sselectpartial{\h_1}{\val}{\h_2}{\s}
      \asyncstepsto{\aglobalinternal{\h_2}}
      \s
    }
\end{mathpar}
   \end{minipage}
  \caption{Stepping rules for asynchronous choreographies.
    These override the rules for $\realstepsto{}$ in \cref{fig:real-stepping}.}
  \label{fig:async-stepping}
\end{figure}

We follow prior work~\cite{async-choreo, PohjolaGSN22} and add syntactic forms
to choreographies to represent partially reduced \ksend/\kreceive pairs
given in \cref{fig:async-choreography}.
We extend endpoint projection so that the new syntactic forms are projected
as a \kreceive statement and a message on the receiver's buffer.
For example, while
\( \smove{\alice}{\val}{\bob}{\x}{\s} \)
becomes a \ksend on \alice and a \kreceive on \bob,
\( \smovepartial{\alice}{\val}{\bob}{\x}{\s} \)
becomes a \kreceive on \bob and a message (from \alice) in \bob's buffer.
We update the stepping rules for communication and selection statements
so that they reduce to the corresponding run-time terms,
which in turn reduce to their continuations.
\Cref{fig:async-stepping} gives the updated rules.

These run-time terms are sufficient to restore perfect correspondence,
and make \cref{thm:projection-soundness,thm:projection-completeness} go through.
We refer to \citet{async-choreo} for details.

\begin{lemma}
  \label{thm:projection-to-async}
  If
  \( \ssecure[\emptylist]{\process} \),
  then
  \(
    \using{\corrupt{\process}}{\asyncstepsto!{}}
    \simulatedBy
    \using{\corrupt{\partition{\process}}}{\realstepsto{}}
  \).
\end{lemma}
\begin{proof}
  Let \( \simulator = \adv \).
\Cref{thm:projection-soundness,thm:projection-completeness} immediately
  give a strong bisimulation.
\end{proof}

\subsection{Restoring Original Choreography Syntax}

\Cref{thm:projection-to-async} proves the correctness of endpoint projection
for \emph{extended} choreographies that have run-time terms.
Next, we show a simple simulation that a choreography
with run-time terms simulates one without,
removing the need to reason about run-time terms in other proof steps.

\begin{lemma}
  \label{thm:async-to-original}
  If
  \( \ssecure[\emptylist]{\process} \),
  then
  \(
    \using{\corrupt{\process}}{\realstepsto!{}}
    \simulatedBy
    \using{\corrupt{\process}}{\asyncstepsto!{}}
  \).
\end{lemma}
\begin{proof}
  The simulator follows the control flow by maintaining a public view of
  the extended choreography
  \( \using{\corrupt{\process}}{\asyncstepsto!{}} \)
  and runs the adversary against this view.
The simulator behaves the same as the adversary,
  except when the adversary schedules a run-time term,
  the simulator takes an internal step
  (and does not schedule the original choreography).
\end{proof}

\begin{proofof}{thm:correctness-of-projection}
  By \cref{thm:projection-to-async,thm:async-to-original}
  using the transitivity of UC simulation.
\end{proofof}
 \end{toappendix}
 
\section{Cryptographic Instantiation}
\label{sec:uc}

Our simulation result is a necessary and novel first step toward constructing a verified,
secure compiler for distributed protocols that use cryptography.
We have abstracted all cryptographic mechanisms into idealized hosts
(e.g., \MPCalicebob); thus, to achieve a full end-to-end security proof,
these idealized hosts must be securely instantiated with cryptographic subprotocols (e.g.,
BGW~\cite{bgw88} for multiparty computation).
Such an instantiation would imply UC security for all compiled programs,
in contrast to existing formalization efforts for individual protocols~\cite{EasyUC,BarbosaBGKS21,gancher2023core}.

To this end, we show how the distributed protocols arising from
our compilation correspond to hybrid protocols in the Simplified Universal
Composability (SUC) framework~\cite{SUC}.
Then, we show how to take advantage of the \emph{composition} theorem
in SUC to obtain secure, concrete instantiations of cryptographic
protocols.

\newcommand{\SUC}[1]{\partition{#1}^{\mathsf{SUC}}}

\paragraph{Simplified UC}

Let $\s$ be a choreography with partitioning $\partition{\s}$.
We construct a corresponding SUC protocol $\SUC{\s}$
which behaves identically to the partitioned choreography, with minor differences due
to the differing computational models.
Each host in $s$ is either a \emph{local} host (e.g., \alice), or an \emph{idealized}
host standing in for cryptography, such as \MPCalicebob. Local hosts
map onto SUC parties, while idealized hosts map onto \emph{ideal functionalities} in SUC.

Protocol execution in SUC happens through
\emph{activations} scheduled by the adversary: a party runs for some steps,
delivers messages to a central \emph{router}, and cedes execution to the
adversary. To faithfully capture the behavior of host \h in $\partition{\s}$,
the party/functionality for \h in $\SUC{\s}$ is essentially a \emph{wrapper}
around the projected host $\project{\s}$, who steps $\project{s}$ accordingly
and forwards correct messages to the router.

Each wrapper needs to explicitly model \emph{corruption}, which our
framework captures by labels: if host $\h$ is semi-honest
($\semihonest{\h}$), the wrapper for $\h$ allows the adversary to query
$\h$ for its current message transcript so far. Similarly, if $\h$ is malicious
($\malicious{\h}$), the wrapper for $\h$ should enable the adversary to
take complete control over $\h$.
By using labels to model corruption, we model \emph{static} security in SUC.

\paragraph{Communication Model}

In SUC, all messages between local hosts are fully public, while messages
between hosts and functionalities contain \emph{public headers} (e.g., the
source/destination addresses) and \emph{private content} (the message payload).
In our system, we do not stratify message privacy along the
party/functionalities axis, but rather along the information flow lattice:
the adversary can read the messages intended for semi-honest hosts, and can forge
messages from malicious hosts.
Indeed, information flow policies allow
more flexible security policies for communication.

However, we can encode our communication model into SUC with the aid of
additional functionalities.
To do so, we make use of a secure channel functionality $\mathcal{R}_\mathsf{sec}$,
which guarantees in-order message delivery and enables secret communication between
honest hosts.
We can realize $\mathcal{R}_\mathsf{sec}$ in SUC via a standard subprotocol using
a public key infrastructure.

For ideal functionalities in $\SUC{\s}$, we need to
ensure that they only communicate with local hosts, and not with other ideal
functionalities.
This property is preserved by compilation, so we only need to
ensure that host selection produces a choreography $\s$ that has this property.
Indeed, our synchronization judgment $\ssynched{\s}$ makes it possible for
choreographies to stay well-synchronized, even when the ideal hosts do not
communicate with each other.

\paragraph{Adversaries and Environments}

In our framework, we prove perfect security against non-probabilistic
adversaries.
However, allowing the adversary to use probability (as in SUC) does not weaken
our simulation result.\footnote{
  The dummy adversary theorem~\cite{UC} implies that security against
  non-probabilistic adversaries guarantees security against probabilistic
  adversaries.}
Additionally, in UC/SUC, the environment is given by a concurrently running
process that outputs a \emph{decision bit}, whereas our model uses a trace
semantics to model the environment.
Security for the latter easily implies the
former, since our simulation result proves equality of environment views between
the two worlds.

Finally, UC (and SUC) require all parties---the adversary,
environment, simulator, and hosts---to run in polynomial time since
cryptographic schemes can be broken given enough time.
Since our simulators query the adversary
and emulate source and target programs in a straightforward manner,
all simulators we define run in time bounded by a polynomial
in the run times of the adversary and the source and target programs.

\subsection{Secure Instantiation of Cryptography}
\label{sec:secureinst}
To securely instantiate cryptographic mechanisms, we appeal
to the \emph{composition} theorem in SUC, which says that ideal
SUC\hyp{}functionalities $\mathcal{F}$ may be substituted for SUC protocols that
securely realize them. Concrete cryptographic protocols are
obtained by applying this theorem iteratively to each ideal host.

Ideal hosts in our model correspond closely to the broad class of
\emph{reactive, deterministic
straight-line} functionalities in SUC, including MPC~\cite{SUC,mpc-sok} and
Zero-Knowledge Proofs (ZKP)~\cite{mpcinthehead}.
The main difference is that our model allows the
adversary to corrupt ideal functionalities (both semi-honestly and maliciously),
while SUC functionalities are incorruptible.
However, we guarantee that the adversary does not gain more power in our model by
restricting the possible corruption models via authority labels for ideal hosts.

For example, we have \MPCalicebob has label $\iflabel{A \wedge B}$,
meaning that
\MPCalicebob is semi-honest (resp. malicious) only if \emph{both} \alice and
\bob are semi-honest (resp. malicious). Thus, any power the adversary gains in
corrupting \MPCalicebob can be instead achieved using \alice and \bob alone.
Similar security concerns for label-based host
selection have been discussed for Viaduct~\cite{viaduct-pldi21}.
We can formalize this intuition via a simulation of the form
$\using{\config}{\idealstepsto{}} \simulates \using{\config}{\idealstepsto{}'}$,
where $\config$ uses $\MPCalicebob$, and $\idealstepsto{}'$ is modified so that
corruption of $\MPCalicebob$ is impossible.
 
\section{Security Preservation}
\label{sec:security-preservation}

We use simulation to define the correctness of compilation, and show
that it corresponds to a well-studied correctness criterion,
\emph{robust hyperproperty preservation} (RHP)~\cite{abateB0HPT19}.
RHP states that hyperproperties~\cite{cs08} satisfied by source programs
under any context are also satisfied by target programs under any context.
RHP is important because common notions of information-flow
security such as termination-insensitive noninterference, observational
determinism~\cite{zm03}, and nonmalleable information flow control~\cite{nmifc}
are hyperproperties.
With a compiler that satisfies RHP, one only needs to prove
security of source programs; security of target programs immediately
follows.

\begin{definition}[Robust Hyperproperty Preservation (RHP)]
  \label{def:rhp}
  Let $\compile{}$ be a compiler from a source program to a target program,
  $\progcompose$ be an operator that composes a program with its context,
  and $\behavior$ be a behavior function that returns the set of possible
  traces generated from a whole program (i.e., a program composed with a
  context).
Then $\compile{}$ satisfies RHP
  over source program set $\prog!$, source context set \srcctx!,
  and target context set \targetctx! when,
  given program $\prog \in \prog!$,
  for all $\targetctx \in \targetctx!$
  there exists $\srcctx \in \srcctx!$ such that
  $\behavior(\srcctx \progcompose \prog)
  = \behavior(\targetctx \progcompose \compile{\prog})$.\footnote{
    This is the ``property-free'' definition of RHP as given by \citet{abateB0HPT19}.
An equivalent but more direct definition is that $\compile{}$ satisfies RHP
    given that if for some hyperproperty HP it is the case that
    $\behavior(\srcctx \progcompose \prog) \in \mathrm{HP}$
    for any $\srcctx \in \srcctx!$,
    then $\behavior(\targetctx \progcompose \compile{\prog}) \in \mathrm{HP}$
    for any $\targetctx \in \targetctx!$.
  }
\end{definition}

\Citet{rhp-uc, rhp-uc-extended} previously observed a correspondence between
UC simulation and robust hyperproperty preservation;
it also holds for our notion of simulation.

\begin{theorem}[Simulation Implies RHP]
  \label{thm:sim-implies-rhp}
  Define \(\adv \progcompose \config = \adv \parallel \config\).
  Then given behavior function $\behavior(\cdot) = \envtraces{\cdot}$
  and an operator \(\compile{}\) between configurations such that
  \(\compile{\config} \simulates \config\) for any configuration
  \(\config \in \config!\), we have that
  \(\compile{}\) satisfies RHP over source program set \config! and
  source and target context set \(\set{\adv}\).
  
\end{theorem}

\begin{corollary}[Partitioning Satisfies RHP]
  \label{thm:partitioning-satisfies-rhp}
  The function
  $\lambda \process. \using{\corrupt{\partition{\process}}}{\realstepsto{}}$
  satisfies RHP over source and target context set $\set{\adv}$ and
  source program set
  \[
   \{ \using{\source{\process}}{\idealstepsto{}} \mid
   \ssecure[\emptylist]{\process}, \ssynched[\sctx]{\process} \}
  \]
\end{corollary}
\begin{proof}
  \Cref{thm:sim-implies-rhp} follows from \cref{def:simulation,def:rhp}.
\Cref{thm:partitioning-satisfies-rhp} follows from
  \cref{thm:correctness-of-partitioning,thm:sim-implies-rhp}.
  
\end{proof}
 
\section{Related Work}
\label{sec:related-work}

\subsubsection*{Secure Program Partitioning and Compilers for Secure Computation}

Prior work on secure program partitioning~\cite{jorchestra,zznm02,zcmz03,viaduct-pldi21}
focuses largely on the engineering effort of compiling security-typed
source programs to distributed code with the aid of cryptography.
\Citet{scvm} does give an informal UC simulation proof of a compiler, but it
is limited to two party semi-honest MPC and oblivious RAM, and they do not
consider integrity.

There is also a large literature on compilers for secure computation
that target specific cryptographic mechanisms such as
MPC~\cite{fairplay, ScaleMamba, hycc, wysteria, silph},
ZKP~\cite{OzdemirBW22, pinocchio, geppetto},
and homomorphic encryption~\cite{EVA, CHET, porcupine, heco, coyote}.
Again, the focus of this literature is efficient implementations, not formal
guarantees.

A long line of work~\cite{fr-popl08, fgr09, AskarovHS08, Laud08, KustersTG12, KustersTBBKM15}
focuses on enforcing computational noninterference for information-flow typed
programs by using standard cryptographic mechanisms, such as encryption.
However, computational noninterference guarantees little
in the presence of downgrading.
In contrast, our compiler enjoys \emph{simulation-based security},
which guarantees preservation of \emph{all} hyperproperties, even
for programs using declassification and endorsement.

Our compilation model is closest to that of Viaduct~\cite{viaduct-pldi21},
which like this work also approximates security guarantees of cryptographic
mechanisms with information-flow labels.
However, this work differs from Viaduct, and from most of the literature
on secure program partitioning and compilers for secure computation, in that our
goal is to provide a \emph{model} of a secure end-to-end compilation process,
not to provide an implementation of an actual compiler.
Our model currently does not analyze a source language with functions, loops, or
mutable arrays, which Viaduct supports.
Additionally, our model does not capture some minor subtleties of Viaduct, such
as allowing secret-dependent conditionals whenever they may be eliminated via
multiplexing the two branches.
An interesting research direction would be to close the gap between our model
and existing compilers by providing a verified compiler implementation akin to
CompCert~\cite{leroy2009}.

\subsubsection*{Simulation-based Security}

Simulation-based cryptographic frameworks, such as
Universal Composability~\cite{UC},
Reactive Simulatability~\cite{RSIM}, and
Constructive Cryptography~\cite{constructivecrypto},
allow modular proofs of distributed cryptographic protocols,
and \citet{ilc} give a core language for formalizing UC protocols.
We abstract away concrete cryptography, so
we do not explicitly model some subtleties of these systems:
probability, computational complexity, and cryptographic
hardness assumptions. But
our approach should be compatible with these frameworks.

Prior verification efforts~\cite{EasyUC,cryptholcsf19,BarbosaBGKS21}
show simulation-based security for concrete
cryptographic mechanisms.
Our work is orthogonal: simulation-based
security for compiler correctness, rather than
proofs for individual protocols.

\subsubsection*{Secure Compilation}

Standard notions of compiler correctness
are derived from full abstraction and hyperproperty preservation~\cite{abateB0HPT19}.
\Citet{rhp-uc, rhp-uc-extended} argue that robust hyperproperty preservation
and Universal Composability are directly analogous.
We affirm this hypothesis by proving that our simulation-based security result
guarantees RHP.
To our knowledge, we are the first to make this connection formally.

\subsubsection*{Choreographies}

The use of choreographies is central to our compilation process
and to the proof of its correctness. The primary concern in the
extensive literature on choreographies~\cite{Montesi23, Montesi13, FilipeM20, FilipeMP22, HirschG22},
is proving deadlock freedom;
very little prior work considers security~\cite{Lluch-LafuenteN15}.
Our extension of choreographies with an information-flow type system,
modeling semi-honest and malicious corruption, is novel.
 
\section{Conclusion and Future Work}
\label{sec:conclusion}

This work presents a novel simulation-based security result for
a compiler from sequential source programs to distributed programs,
using idealizations of cryptographic mechanisms.
Our simulation result guarantees that the security properties
of source programs are preserved in the compiled protocols.

We believe this work opens up many opportunities for future research;
in particular, it gives a clear path toward building a fully verified
cryptographic compiler.
Aside from adding more language features (e.g., subroutines, loops, and mutable arrays)
to our source language to better match real-world cryptographic compilers,
the remaining verification effort largely consists of verifiably instantiating
all abstract components we assume.
For example, we model protocol selection via abstract judgments on
choreographies; in turn, outputs of a concrete protocol selection algorithm must satisfy these
judgments.
Additionally, we assume that parties communicate over secure channels, and use
ideal functionalities for performing cryptography; these
assumptions can be eliminated using cryptographic instantiations verified in the Simplified UC framework.

Our security result holds for a strong attacker that knows when
communication occurs between hosts and controls its scheduling,
but secret control flow is therefore precluded.
To improve expressive power, weaker attacker models should be explored.

We have focused on confidentiality and integrity properties; however,
some protocols also explicitly address availability~\cite{zm14-plas},
which would be another interesting direction for exploration.
  
\ifacknowledgments
  \section{Acknowledgments}
  This work was supported by the
National Science Foundation under
NSF grant 1704788. Any opinions, findings, conclusions or recommendations
are those of the authors and do not necessarily reflect the views of
any of these funders.
We thank the reviewers for their thoughtful suggestions.
We also thank
Silei Ren,
Ethan Cecchetti,
Drew Zagieboylo,
and Suraaj Kannawadi
for discussions and feedback.
 \fi

\iffinalformat
\relax
\else
\finalpage
\fi

\begingroup
\small
  \bibliography{override, bibtex/pm-master}
\endgroup

\iffinalformat
\finalpage
\else
\relax
\fi

\end{document}